\documentclass{article}

\usepackage[utf8]{inputenc}
\usepackage{authblk}
\usepackage{setspace}
\usepackage{graphicx}
\graphicspath{ {./figures/} }
\usepackage{subcaption}
\usepackage{amsmath}
\usepackage{lineno}
\usepackage{float}
\usepackage{helvet}
\usepackage{etoolbox}
\usepackage{graphicx}
\usepackage{titlesec}
\usepackage{titletoc}
\usepackage{caption}
\usepackage{nicematrix}
\usepackage{booktabs}
\usepackage{xcolor} 
\usepackage{amsmath}
\usepackage{amssymb}
\usepackage{amsthm}
\usepackage[colorlinks=true,citecolor=cyan,urlcolor=cyan,linkcolor=black]{hyperref}
\usepackage{subcaption}
\usepackage{physics}
\usepackage{mathtools,delarray}
\usepackage{algpseudocode}
\usepackage{phfcc}
\usepackage{revsymb}
\usepackage{nicematrix}
\usepackage{tikz}
\usepackage{enumitem}

\usepackage[a4paper,top=2truecm,left=2truecm,right=2truecm,bottom=2truecm,includefoot]{geometry}

\usetikzlibrary{positioning, arrows.meta}
\usetikzlibrary{calc, positioning, decorations.pathreplacing}
\usetikzlibrary{arrows.meta}

\NiceMatrixOptions{parallelize-diags=false}
\usepackage{comment}

\phfMakeCommentingCommand[initials={NT}]{NT}
\phfMakeCommentingCommand[initials={AW}]{AW}
\phfMakeCommentingCommand[initials={JA}]{JA}

\newcommand{\CC}{\mathbb{C}}
\newcommand{\RR}{\mathbb{R}}

\newcommand{\M}{\mathcal{M}}

\renewcommand{\Tr}{\operatorname{Tr}}
\newcommand{\1}{\openone}

\newcommand{\ex}{\operatorname{ex}}
\newcommand{\diag}{\operatorname{diag}}
\newcommand{\spec}{\operatorname{spec}}

\newcommand{\relint}{\mathrm{relint}} 
\newcommand{\aff}{\mathrm{aff}}

\renewcommand{\ket}[1]{\vert #1\rangle}
\renewcommand{\bra}[1]{\langle #1\vert}

\newtheorem{theorem}{Theorem}
\numberwithin{theorem}{section}

\newtheorem{conjecture}[theorem]{Conjecture}
\newtheorem{corollary}[theorem]{Corollary}

\newtheorem{lemma}[theorem]{Lemma}

\newtheorem{proposition}[theorem]{Proposition}

\usepackage{phfcc}

\theoremstyle{definition}
\newtheorem{definition}[theorem]{Definition}
\newtheorem{example}[theorem]{Example}
\newtheorem{remark}[theorem]{Remark}

\usepackage[style=ieee, backend=biber, sorting=none]{biblatex}
\begin{document}

\title{The geometry of absolute separability \protect\\ and other convex matrix properties from spectrum}

\author[1,3,a]{Jennifer Ahiable}
\author[1,b]{Naga Bhavya Teja Kothakonda}
\author[1,2,3,c]{Andreas Winter}

\renewcommand\Authfont{\fontsize{10}{10.8}\selectfont}
\renewcommand\Affilfont{\fontsize{10}{10.8}\selectfont}
\renewcommand*{\Authsep}{, }
\renewcommand*{\Authand}{, }
\renewcommand*{\Authands}{, }
\setlength{\affilsep}{2em}  
\newsavebox\affbox

\affil[1]{{\small\it Grup d'Informaci\'o Qu\`antica, Departament de F\'isica,\protect\\ Universitat Aut\`onoma de Barcelona, 08193 Bellaterra (Barcelona), Spain\vspace{1mm}}}
\affil[2]{{\small\it ICREA--Instituci\'o Catalana de Recerca i Estudis Avan\c{c}ats,\protect\\ Pg.~Lluis Companys, 08010 Barcelona, Spain\vspace{1mm}}}
\affil[3]{{\small\it Department Mathematik/Informatik--Abteilung Informatik,\protect\\ Universit\"at zu K\"oln, Albertus-Magnus-Platz, 50923 K\"oln, Germany\vspace{1mm}}}

\affil[a]{{\small\emph{Email:} jennifer.ahiable@autonoma.cat}}
\affil[b]{{\small\emph{Email:} nagabhavyateja.kothakonda@uab.cat}}
\affil[c]{{\small\emph{Email:} andreas.winter@uni-koeln.de}}

\date{\vspace{-4mm}\small (4 August 2026)\vspace{-3mm}}


\maketitle

\thispagestyle{empty}

\begin{abstract}
We investigate the geometric structure of the set of spectra of bipartite absolute separable states ($\mathrm{ASEP}_{m,n}$) and absolute positive partial transpose states ($\mathrm{APPT}_{m,n}$), \emph{i.e.}, bipartite quantum states that remain separable or PPT, respectively, under all global unitary transformations. First, we establish general geometric properties of absolute convex sets of matrices, their spectra and extreme points. 
Regarding absolute separability,  we  present a permutation-symmetric reformulation of the absolute PPT criterion and use it to demonstrate that $\mathrm{APPT}_{m,n}$ is a spectrahedron for all $m\leq n$: in particular,  all its faces are exposed. While $\mathrm{ASEP}_{2,n}$ is a spectrahedron, too, for $\mathrm{ASEP}_{m,n}$  we can in general only show that it is a semialgebraic set for all $m,n\geq 3$. Furthermore, we provide a  complete characterization of the faces and extreme points of $\mathrm{APPT}_{m,n}$ and demonstrate that the dimension of a face is determined by the rank of a certain matrix, with maximal proper faces having dimension $mn-m-1$. 

In the quantitative setting, we provide a rigorous lower bound on the maximal attainable purity of $\mathrm{APPT}_{m,n}$ via an inscribed polytope $\mathcal{P}_{m,n}$ and conjecture that the maximal purity of $\mathrm{APPT}_{m,n}$ (along with its spectra) coincides with the polytope for arbitrary dimensions except when $m=n=2$. Additionally, we also provide a rigorous upper bound on the minimal von Neumann entropy of $\mathrm{APPT}_{m,n}$ and demonstrate numerically that the minimum entropy eventually coincides with the polytope  $\mathcal{P}_{m,n}$ as the local system dimension $n$ increases. Finally, we show that the relative spectral volume of $\mathrm{APPT}_{m,n}$ decays exponentially in $n$ by a constant multiplicative factor of the relative volume of the inscribed polytope $\mathcal{P}_{m,n}$. 
\end{abstract}

\setcounter{tocdepth}{1}

\startcontents
\renewcommand{\baselinestretch}{0.7}\normalsize
\printcontents{}{0}{}
\renewcommand{\baselinestretch}{1.0}\normalsize


\section{Introduction}
At the heart of quantum information theory is the phenomenon of quantum entanglement \cite{schrodinger1935discussion,einstein1935can,horodecki2009quantum}. Over the years, many applications of entanglement have been proven to be crucial for information processing and computation, such as quantum cryptography \cite{ekert1991quantum} and quantum teleportation \cite{bennett1993teleporting}. However, although this phenomenon has been experimentally demonstrated to exist in Nature \cite{Clauser1969}, a closed-form description is particularly difficult. Mathematically, entanglement is described as the complementary property to separability. Formally, a bipartite quantum state $\rho\in \M_m\otimes \M_n$ is said to be \emph{separable} \cite{werner1989quantum} if and only if it can be written as $\rho = \sum_j X_j \otimes Y_j$ where $X_j \in \M_m$ and $Y_j \in \M_n$ are positive semidefinite matrices. Any state that does not admit this decomposition is referred to as \emph{entangled}. Thus, a better understanding of separability directly correlates to a better understanding of entanglement. However, determining whether a given state is separable or not is generally known to be difficult: more precisely, it is NP-hard \cite{gurvits2003classical,gharbian2010}. Regardless, several necessary and/or sufficient conditions have been provided to address the separability/entanglement detection problem \cite{asherperes1996,Horodecki1996Separability,guhne2009entanglement}. 
Among these is the \emph{positive partial transpose (PPT)} criterion \cite{asherperes1996,Horodecki1996Separability}: a state $\rho \in\M_m \otimes \M_n$ is said to have \emph{positive partial transpose} (or \emph{to be PPT}) if $\rho^\Gamma$ is positive semidefinite for the linear map
$\Gamma(X \otimes Y) = X \otimes Y^T$,
where $X \in \M_m$ and $Y \in \M_n$, and $T$ denotes the standard matrix transpose. The PPT criterion is only sufficient whenever the total dimension of the system is less or equal to $6$ but necessary for separability in arbitrary dimensions. 

Within this class of states lies a convex and compact subset of separable states that remain separable after the transformation $U\rho U^\dagger$ under all global unitaries $U$ in the unitary group $\mathcal{U}(mn)$, widely referred to as \emph{absolutely separable states} \cite{kus2001geometry}. Since the eigenvalues of a state are the precise invariants under unitary rotations, by definition, the absolute separability property of the state is determined by its spectrum. Thus, \emph{absolute separability} is a property of the eigenvalues of a state (occasionally also referred to as \emph{separability from spectrum}  \cite{knill2003,johnston2014separability,Verstraete}) and asks for a complete characterization of such states entirely in terms of their eigenvalues. 
Just like for separability, the most natural step to solve this problem is to consider the  relaxation known as \emph{absolute PPT}. Defined analogously to absolutely separable states, absolute PPT states are those states which remain PPT after global unitary rotations and have been completely characterized across all dimensions \cite{hildebrand2007positive}. Remarkably, it has been shown that the set of bipartite absolute separable states coincides with the set of absolute PPT states for qubit-qudit systems \cite{johnston2013sepfromspectra} and this has been conjectured to hold for any arbitrary dimension with numerical evidence to support the claim \cite{johnston2018inverse}. Even though partial characterizations have been presented \cite{abellanet2025sufficient,XiongSze2026,kondra2026fundamental}, the problem remains open. Previous studies such as in \cite{fawzi2021set} have shown that while the bipartite separable states form semialgebraic sets in general, they are  spectrahedral shadows if and only if the total dimension of the system is less or equal to $6$, demonstrating the equivalence with PPT states. Since the set of absolute separable and PPT states are convex and compact \cite{ganguly2014witness}, it is only natural that we consider the absolute separability problem from a similar geometric perspective.
More recently, the boundary characteristics particularly focused on the extreme points have been studied for qubit-qudit systems and qutrit-qudit systems \cite{SongChen2025,halder2021characterizing,wang2026extreme}. 

In this paper, we seek to understand the geometric properties of \emph{the set of spectra} of absolutely separable states (which we denote by $\mathrm{ASEP}_{m,n}$) and absolutely PPT states ($\mathrm{APPT}_{m,n}$), respectively, and illustrate how underlying symmetries within the sets could provide more insight into solving the absolute separability problem and the open question of the equivalence of both sets. In particular, we begin by defining general absolute convex sets of matrices and investigate equivalence properties of these sets in terms of their spectra. Following that, we present a permutation-invariant reformulation of Hildebrand's absolute PPT criterion and and identify that the set of spectra $\mathrm{ASEP}_{2,n}$ and $\mathrm{APPT}_{m,n}$ are spectrahedra. In the more general sense, we show that  $\mathrm{ASEP}_{m,n}$ are semialgebraic sets for arbitrary dimensions whenever $m\leq n$. Due to the spectrahedral description of $\mathrm{APPT}_{m,n}$, we show that
its topological boundary can be expressed as the feasible part of $(mn-2)$-dimensional determinantal hypersurfaces. With this, we are able to bound the number of irreducible determinant polynomials describing the topological boundary. Additionally, we provide a kernel-based description of the facial structure of $\mathrm{APPT}_{m,n}$ which allows us to characterize its boundary, faces and extreme points by a constraint matrix with linear entries generated from the kernel-subspaces. Thus, we determine the complete set of possible facial dimensions of $\mathrm{APPT}_{m,n}$ via the rank of this constraint matrix such that the maximum dimension of a proper face of $\mathrm{APPT}_{m,n}$ is $(mn-m-1)$. 

We also extend our results to study an equivalent permutation-invariant polytope $\mathcal{P}_{m,n}\subseteq \mathrm{APPT}_{m,n}$
and compute its maximal attainable purity and minimum von Neumann entropy. This gives us rigorous lower and upper bounds on the maximum purity and minimal von Neumann entropy of $\mathrm{APPT}_{m,n}$, respectively, which are inherently different from those obtained via the separable  ball. We conjecture that the maximal purity of this polytope and  $\mathrm{APPT}_{m,n} $ are equal except when $m=n=2$, and demonstrate this numerically for $m=2,3,4$. Additionally, we show that the relative volume of $\mathrm{APPT}_{m,n}$ follows a similar exponential decay as that of the inscribed polytope rather the separable which decays much faster. We utilize mathematical tools from real semialgebraic geometry and convexity theory.

The rest of the paper is organized as follows. In Section \ref{sec:Preliminaries}, we introduce our notations and various mathematical details used throughout the paper. In Section \ref{sec:AM}, we define absolute convex sets of matrices and their sets of spectra, and prove general relations between the extreme points of the two. In Section \ref{sec:spectrahedron}, we present a permutation-invariant reformulation of $\mathrm{APPT}_{m,n}$  and demonstrate it is a spectrahedron, while  $\mathrm{ASEP}_{m,n}$ is semialgebraic. In Section~\ref{sec:faces and extreme points}, we introduce the kernel-based description of the facial structure and provide necessary and/or sufficient conditions for the faces and extreme points of the set. In Section \ref{sec:purity-etc}, we move to quantitative aspects: we provide a lower bound on the maximum purity, upper bound on the minimum von Neumann entropy and assess the relative volume of $\mathrm{APPT}_{m,n}$ via the separable ball and the inscribed polytope. In Section \ref{sec:conclusion}, we discuss our results, open questions and potential geometric approaches for solving them. The paper ends with Appendix \ref{app:Appendix_A1} where we provide the proofs of some results in Subsection \ref{subsec:boundary}.


\section{Preliminaries}\label{sec:Preliminaries}
Let us introduce our notations and briefly introduce some mathematical concepts relevant for the subsequent chapters.

\subsection{Convex geometry}
A set $C\subset \RR^n$ is \emph{convex} if for any $\mathbf{x},\mathbf{y}\in C$, the line segment connecting $\mathbf{x}$ and $\mathbf{y}$ is contained in $C$: $[\mathbf{x},\mathbf{y}] = \{ \alpha \mathbf{x} +(1-\alpha)\mathbf{y} \mid 0\leq\alpha\leq 1 \} \subseteq C$. We denote by $\mathrm{aff}(C)$, the \emph{affine hull} of the convex set $C\subset \RR^n$ and by  $\mathrm{relint}(C)$, the \emph{relative interior} of $C$.
A \emph{face} of $C$ is a convex subset $F\subseteq C$ such that for $\alpha \in (0,1)$ and $\mathbf{x},\mathbf{y} \in C$, if $\alpha \mathbf{x} +(1-\alpha)\mathbf{y} \in F$, then $\mathbf{x},\mathbf{y}\in F$. We refer to the non-empty faces of $C$ with $F\neq C$ as \emph{proper} faces and the faces $\emptyset$ and $C$ as \emph{trivial} faces. An \emph{extreme} point of $C$ is simply a singleton face. That is, a point $\mathbf{p} = \alpha \mathbf{x} +(1-\alpha)\mathbf{y} $ with $\mathbf{x},\mathbf{y}\in C$ and $\alpha\in (0,1)$ must satisfy $\mathbf{x}=\mathbf{y} = \mathbf{p}.$ The dimension of any convex set $C$, denoted $\dim(C)$, is the dimension of the smallest affine subspace that contains $C$ and thus dimensions of non-empty faces are well-defined given that they are themselves convex sets. With the exception of $C$ itself, all faces lie on the \emph{topological boundary}, denoted $\partial C$, of the convex set $C$, i.e., the set difference between the closure and the interior of the given set. The closure $\overline{C}$ is the intersection of all closed sets containing $C$ and interior, $\mathrm{int}(C)$ is the union of all open sets contained in $C$. Note that $\mathrm{int}(C) = (\overline{C^c})^c$, where $\bullet^c$ denotes the set complement. A proper face $F$ of $C$ is said to be \emph{exposed} if and only if $F = C \cap H$, for some \emph{supporting} affine hyperplane $H$.

Occasionally, it is convenient to study convex sets by working with their corresponding convex cones especially in the context of duality. A set $K \subseteq \mathbb{R}^n$ is a \emph{convex cone} if for any $\mathbf{x}, \mathbf{y} \in K$ and any scalars $\alpha, \beta \geq 0$, the linear combination $\alpha\mathbf{x} + \beta\mathbf{y}$ belongs to $K$. The \emph{dual cone} of a cone $K \subseteq \mathbb{R}^n$, denoted $K^*$, is the set defined as
$
    K^* = \{ \mathbf{y} \in \mathbb{R}^n \mid \langle \mathbf{y}, \mathbf{x} \rangle \geq 0 \text{ for all } \mathbf{x} \in K \}.
$ 
We define the \emph{polar} of the convex set $K$ (with respect to the origin $0$) as 
$
    K^\circ = \{\mathbf{y}\in \RR^n : \langle \mathbf{y}, \mathbf{x}\rangle\leq 1,\; \text{ for all } \mathbf{x}\in K\}.
$
The latter notion is most suited to compact convex sets containing the origin in its interior, because then $K^\circ$ again has all those properties. A cone is \emph{pointed} if it contains no lines (i.e., $K \cap -K = \{\mathbf{0}\}$). 
An \emph{extreme ray} of a convex cone $K$ is a ray $\{\mu\mathbf{x} \in \RR^n\mid \mu \geq 0\}\subseteq K$ such that if $\mathbf{x} = \mathbf{y} +\mathbf{z}$ for $\mathbf{y}, \mathbf{z} \in K$, then both $\mathbf{y}$ and $\mathbf{z}$ must lie on the same ray. We denote the set of extreme rays of $K$ as $\mathrm{ex}(K).$ If $K$ is a compact convex set, $\ex(K)$ should be used for the extreme points and $\exp(K)$, for exposed points.
We refer the reader to \cite{rockafellar2015convex,barvinok:convexity} for more on the background on convex sets, see also \cite{boyd2004convex}.

\subsection{Real algebraic geometry}
We recall some basic facts about real algebraic sets, semialgebraic sets and spectrahedra. More comprehensive references on foundations of the subjects are, among others, \cite{BCR:RAG,cynthiaspectrahedra,scheiderer2022extreme,netzerspectrahedra,ramana1995some}.

For $\mathbf{x}=(x_1,\ldots, x_n)\in \RR^n$, we define $\mathbb{R}[\mathbf{x}] = \mathbb{R}[x_1,\ldots, x_n]$ to be the real commutative ring of polynomials. For any subset of polynomials $S \subset \mathbb{R}[\mathbf{x}]$, the \emph{real algebraic set} generated by $S$ is the subset of $\mathbb{R}^n$ defined as
$
    \mathcal{Z}(S) = \{ \mathbf{x} \in \RR^n \mid \forall g\in S\ g(\mathbf{x}) = 0 \}.
$
If $S = \{g\}$ is a single polynomial, $\mathcal{Z}(g)$ is referred to as an \emph{algebraic hypersurface}. A real algebraic set $C \subset \mathbb{R}^n$ is \emph{irreducible} if it cannot be expressed as the union $C = C_1 \cup C_2$ of two proper algebraic subsets $C_1, C_2 \subsetneq C$. Every real algebraic set $C$ admits a unique minimal decomposition into a finite union of irreducible algebraic sets $ C = C_1 \cup C_2 \cup \dots \cup C_k,$ where $C_i \not\subseteq C_j$ for all $i \neq j$.

A \emph{semialgebraic} subset of $\RR^n$ is a subset of the form 
    $$ \bigcup_{i=1}^{s}\bigcap_{j=1}^{t_i} \{ \mathbf{x} \in \mathbb{R}^n \mid g_{i,j}(\mathbf{x}) *_{i,j} 0 \}$$ where $g_{i,j}(\mathbf{x}) \in \mathbb{R}[\mathbf{x}]$ and $*_{i,j} $ is either $< $ or $=$, for $i =1,\ldots, s$ and $j = 1,\ldots,t_i$. 
\noindent An important class of convex semialgebraic sets are spectrahedra. We denote by $\mathcal{S}_N^+$, the positive semidefinite cone of $N\times N$ matrices and by $\mathcal{M}_N$, the space of complex Hermitian $N\times N$ matrices. A set $C\subseteq \RR^n$ is called a \emph{spectrahedron} if there exists an affine linear map, $\mathcal{L}: \RR^n \to \M_N$, i.e, $\mathcal{L}(\mathbf{x})= A_0 +  x_1 A_1 +\cdots+ x_n A_n, $ where $ A_i \in \M_N$ such that 
    \begin{equation*}
         C = \{ \mathbf{x} \in \mathbb{R}^n |\; \mathcal{L}(\mathbf{x})\geq 0\}.
    \end{equation*}
 Here, $\geq 0$ denotes positive semidefiniteness, and more generally $\geq$ the semidefinite (L\"owner) order: $A\geq B$ if and only if $A-B\geq 0$.


\section{Absolute convex sets of matrices, their sets of spectra and extreme points}
\label{sec:AM}
In the present section we discuss general relations between convex sets of matrices, their absolute versions and their spectra. Here it is convenient to work in the cone setting (for example separable and PPT matrices without normalisation). The observations made here motivate why subsequently one focuses on the sets of spectra of absolute convex sets.

The space $\mathcal{M}_d$ of complex Hermitian $d\times d$-matrices is equipped with the Hilbert-Schmidt inner product $\langle A, B \rangle = \mathrm{Tr}(A^\dagger B)$ where $\bullet^\dagger$ refers to the Hermitian adjoint (conjugate transpose in matrix form). This makes it into a standard Euclidean space, in particular the notion of cone duality carries over.
Let $M \subset \M_d$ be a closed convex cone in the space of Hermitian $d\times d$-matrices, containing the identity $\1$. Its dual $M^*$, which thanks to the trace inner product on $\M_d$ can itself be regarded as a cone in $\M_d$:
\[
  M^* = \{ X\in \M_d : \forall Y\in M\ \Tr XY \geq 0 \}. 
\]
In general, $M^*$ is a closed convex cone, too, and $(M^*)^* = M$ by Farkas' Lemma \cite{rockafellar2015convex}.

We define the \emph{absolute $M$}-set as the set
\[
  AM \coloneq \bigcap_{U \in \mathcal{U}(d)} UMU^\dagger, 
\]
which is clearly a closed convex cone as well, and has the dual cone 
\begin{equation}
  \label{eq:absolute_dual}
  (AM)^* = \sum_{U\in\mathcal{U}(d)} U M^* U^\dagger,
    \end{equation}
the set of all finite sums of elements of the form $UXU^\dagger$, $X\in M^*$ and $U\in\mathcal{U}(d)$. Note that by Caratheodory's Lemma, every element of the right-hand side of Eq.~\eqref{eq:absolute_dual} can be realized as a sum of at most $d$ points of the form $UXU^\dagger$, $X\in M^*$, which is why (via compactness) the right-hand side is already a closed convex cone.
If $M$ does not contain the identity $\1$ or a multiple of it, then trivially $AM = \emptyset$, so we shall generally assume that $\1\in M$. 
If furthermore $M$ is not generating, or indeed if it does not contain an open neighbourhood of $\1$, then $AM \subset \{t\1 : t \in \RR \}$ (leaving the possibilities that $AM = \RR\1$ or $AM = \RR_{\geq0}\1$). This is also not very interesting, so we shall generally assume that $M$ is generating and contains an open neighbourhood of $\1$, which makes $AM$ generating too: namely, note that the open neighbourhood of $\1$ in $M$ may be assumed a ball in a unitarily invariant norm, say the Hilbert-Schmidt norm, and then $AM$ contains the same open ball. We will assume that $M$ is pointed (making $AM$ pointed, too, automatically), because pointed cones are generated by their extremal rays, denoted $\ex(M)$, due to the conic version of the Krein-Milman theorem. Special unitary matrices are the permutation matrices $U_{\pi}$ for permutations $\pi \in S_d$.



Define furthermore the diagonal matrices among $M$ (which could be an arbitrary set for this definition) as $\diag(M) \coloneq M \cap D_d$, where $D_d = {\text{diagonal real matrices}}$, and
$$\spec(M) \coloneq \{\diag(\lambda_1,...,\lambda_d) \mid (\lambda_1, \dots, \lambda_n) \text{ is the spectrum of some } A \in  M\}$$
the spectra of matrices in $M$ (with multiplicities), written as diagonal matrices. 
Note that $\diag(M) \subset \spec(M)$, but in general of course they are not equal, even for convex sets. However, it evidently holds that 
\begin{equation}\label{eq:AM-diag=spec}
     \diag(AM) = \spec(AM)
     \end{equation}
for any closed convex cone $M$, as well as 
\begin{equation}\label{eq:AM-spec}
     AM = \bigcup_{U\in\mathcal{U}(d)} U\left(\spec(AM)\right)U^\dagger.
     \end{equation}
The latter property means that membership of a matrix $X$ in $AM$ is decided entirely by the spectrum $\spec(X)$, with no role played by the eigenbasis. This is the reason why $AM$ is sometimes interpreted as describing ``$M$-ness from spectrum'': it is the largest convex cone of Hermitian matrices $X$ such that $U\left(\spec(X)\right)U^\dagger \in M$ for all unitaries $U$. Intuitively, all geometric properties of $AM$ are thus encoded in $\diag(AM)$, and next we shall derive some results supporting this idea.

The following properties follow from these definitions for arbitrary closed convex cones $M$:

\begin{theorem}
\label{thm:AM-dual}
$\diag\left((AM)^*\right) = \sum_{\pi\in S_d} U_{\pi} \left(\spec(M^*)\right) U_{\pi}^\dagger$.
\end{theorem}

\begin{proof}
The r.h.s. is clearly contained in the l.h.s., so 
we only have to worry about the opposite inclusion. Let $X \in \diag((AM)^*)$, meaning $X = \sum_i U_i T_i U_i^\dagger$ for some $T_i \in M^*$, and $D(X) = X$, where $D$ is the completely dephasing map that projects $\M_d$ onto $D_d$. Thus, $X = D(X) = \sum_i D(U_i T_i U_i^\dagger)$, and from general properties of majorisation (Schur) \cite{marshall1979inequalities}, it follows that for every $i$, $D(U_i T_i U_i^\dagger)$ is majorised by $\spec(T_i)$. 
This on the other hand (Hardy-Littlewood-P{\'o}lya \cite{hardy1952inequalities}) implies that there are probabilities $q(\pi|i)$ on permutations $\pi$ such that 
$D(U_i T_i U_i^\dagger) = \sum_{\pi} q(\pi|i) U_{\pi}(\spec(T_i))U_{\pi}^\dagger$. 
Inserting this into the previous equation for $X$, we get $X = \sum_i \sum_{\pi} q(\pi|i) U_{\pi}\left (\spec(T_i) \right) U_{\pi}^\dagger$, 
and we are done. 
\end{proof}

\begin{theorem}\label{thm:AM-Dephasing}
    $ \diag(AM) = D(AM)$ where $D$ is the completely dephasing map that projects the space of Hermitian matrices $\M_d$ onto the space of diagonal matrices $D_d$.
\end{theorem}

\begin{proof} The forward inclusion $\diag(AM) \subset D(AM)$ is straightforward. Let $X \in \diag(AM)$, meaning $X \in AM$ and $X \in D_d$. Since $X$ is diagonal, it is invariant under the completely dephasing map, i.e., $X = D(X) $.  Thus, $X = D(X) \in D(AM)$. 

For the opposite inclusion, suppose $X \in D(AM)$. Then there exists $Y\in AM$ such that $X = D(Y).$ By majorization properties, it follows that $D(Y)$ is majorised by $\spec(Y)$. Consequently, by the Hardy-Littlewood-Polya theorem, there exist probabilities $q_\pi \ge 0$ such that
$ D(Y) = \sum_{\pi \in S_d} q_\pi U_\pi (\spec(Y)) U_\pi^\dagger.$ Since $Y\in AM,$ and $AM$ is unitarily invariant $\spec(Y) \in \spec(AM)\subset AM$. This means that for all $\pi \in S_d$, $U_\pi(\spec(Y))U_\pi^\dagger \in AM$. Furthermore, since $AM$ is a closed convex cone, it follows that $X \in AM$. By definition, $X$ is diagonal, therefore, we have $X \in AM \cap D_d = \diag(AM)$.
\end{proof}

\begin{theorem}
\label{thm:AM-dual-diagonals}
$\diag((AM)^*) = (\diag(AM))^*$, where on the right hand side we regard $\diag(AM)$ as a cone in the space of diagonal matrices $D_d$, and 
the dual is also considered in $D_d$ (which is selfdual under the trace inner product). 
\end{theorem}

\begin{proof}
$X \in \diag((AM)^*)$ means that $X$ is diagonal and for all $T \in M$ and unitaries $U$, $\Tr XUTU^\dagger \geq 0$. In particular, this includes $UTU^\dagger = \spec(T)$, so $\Tr X(\spec(T)) \geq 0$ for any $T \in M$ and any diagonal arrangement of the spectrum of $T$, so certainly $X \in (\diag(AM))^*$, as $\diag(AM) = \spec(AM)$. 
Conversely, consider a diagonal $X \in (\diag(AM))^*$, i.e. for all $T \in M$, $\Tr X(\spec(T)) \geq 0$. To get the inclusion of $X \in \diag((AM^*))$, we have to show that for all unitaries $V$, $\Tr XVTV^\dagger \geq 0.$ 
Indeed, as $X=D(X)$, 
\[
  \Tr XVTV^\dagger = \Tr D(X)VTV^\dagger = \Tr X D(VTV^\dagger).
\]
As in the previous proof, we use that $ D(VTV^\dagger)$ is majorised by $\spec(T)$, i.e. there exists a distribution $q$ on $S_d$ such that 
$D(VTV^\dagger) = \sum_{\pi} q(\pi) U_{\pi}(\spec(T))U_{\pi}^\dagger$. 
But we already assume that $\Tr X U_{\pi}(\spec(T))U_{\pi}^\dagger \geq 0$, hence also $\Tr X D(VTV^\dagger) \geq 0$.
\end{proof}

\begin{theorem}
\label{thm:AM-extremal}
$\diag(\ex(AM)) = \ex(\diag(AM))$ and indeed it holds that
\[
  \ex(AM) = \bigcup_{U\in\mathcal{U}(n)} U(\ex(\diag(AM)))U^\dagger.
\]
\end{theorem}

\begin{proof}
The l.h.s. is clearly contained in the r.h.s., as this is a general fact: every extreme point of $AM$ that happens to be diagonal is automatically an extreme point of the subset $\diag(AM) \subset D_d$. 

For the opposite inclusion, assume that there were an extremal $X \in \diag(AM)$ that is however not extremal in AM. In other words, $X = X_1 + X_2$ with $X_1, X_2 \in AM$ linearly independent. But now,
\[
  X = D(X) = D(X_1) + D(X_2),
\]
so by extremality of $X \in \diag(AM)$, we must have that $D(X_1)$ and 
$D(X_2)$ are linearly dependent. This in turn implies that for at 
least one $i=1,2, X_i \neq D(X_i)$, otherwise $X_1$ and $X_2$ would be 
linearly independent and dependent at the same time. In fact, by 
the above equations, we conclude that for both $i=1,2$, $X_i \neq D(X_i)$.
Now we activate the majorisation insight a third time. We have 
$S_i = \spec(X_i)$ majorises $D(X_i)$, and since $X_i$ is different from 
$D(X_i)$, the majorisation is strict. I.e., there exist distributions 
$q(\pi|i)$ for $i=1,2$ such that 
$D(X_i) = \sum_{\pi} q(\pi|i) U_{\pi}(\spec(X_i))U_{\pi}^\dagger$, 
and the convex combination is nontrivial, meaning there are two permutations $\pi$ and $\tau$ with $q(\pi|i),\ q(\tau|i) > 0$ and also 
$U_{\pi}(\spec(X_i))U_{\pi}^\dagger \neq U_\tau(\spec(X_i))U_{\tau}^\dagger$. Note that 
the latter implies that the two permuted versions of $\spec(X_i)$ are indeed linearly independent. Thus, 
\begin{align*}
    X &= D(X_1) + D(X_2) \\
      &= \sum_{\pi} q(\pi|1) U_{\pi}(\spec(X_1))U_{\pi}^\dagger
       + \sum_{\pi} q(\pi|2) U_{\pi}(\spec(X_2))U_{\pi}^\dagger
\end{align*}
is a decomposition of $X$ into elements from $\diag(AM)$, at least two 
of which are linearly independent, contradicting the assumption of extremality of $X \in \diag(AM)$. So, $X$ must have been extremal in $AM$ all along. 

The union equality follows directly due to the unitary invariance of $\ex(AM)$.
\end{proof}

The same holds for the set of exposed extreme rays, $\exp(AM)$: 

\begin{theorem}
\label{thm:AM-exposed-extremal}
$\diag(\exp(AM)) = \exp(\diag(AM))$, 
and indeed it holds that
\[
  \exp(AM) = \bigcup_{U\in\mathcal{U}(d)} U(\exp(\diag(AM)))U^\dagger.
\]
\end{theorem}
\begin{proof}
We first show $\diag(\exp(AM)) \subseteq \exp(\diag(AM))$. Let $X \in \diag(exp(AM))$, then $X$ is a diagonal matrix in $D_d$ and there exists $N \in (AM)^*$  such that $\Tr(NY)\geq 0$ for all $Y\in AM $ with $\Tr(NY) =0$ if and only if $ Y= \mu X$, for some $\mu\geq 0$. However, since $X$ is diagonal, we consider $Y\in \diag(AM)$. As  $Y = D(Y)$, $\Tr(D(N)Y) = \Tr(N D(Y)) = \Tr(NY)\geq 0$ where equality holds for $Y = \mu X$. Since $D(N)$ exposes $X$ in the diagonal subspace, $X \in \exp(\diag(AM))$. 

For the opposite inclusion, assume $X \in\exp(\diag(AM))$. Then there exists a diagonal functional $\Lambda \in (\diag(AM))^*$ that strictly exposes $X$ in $\diag(AM)$. By Theorem \ref{thm:AM-dual-diagonals}, $(\diag(AM))^* = \diag((AM)^*)$, so $\Lambda \in (AM)^*$.
Suppose there exists $Z \in AM$ such that $\Tr(\Lambda Z) = 0$. Since $\Lambda$ is diagonal, $\Tr(\Lambda Z) = \Tr(\Lambda D(Z)) = 0$. By Theorem \ref{thm:AM-Dephasing}, $D(Z) \in \diag(AM)$. As $\Lambda$ exposes $X$ in $\diag(AM)$, it must be that $D(Z) = cX$ for some $c \ge 0$. By majorisation properties, $D(Z)$ is majorised by $\spec(Z)$. And consequently by Hardy-Littlewood-Polya, there exist probabilities $q_\pi \geq 0$ for all permutations $\pi\in S_d$ such that $cX = D(Z) = \sum_{\pi\in S_d} q_\pi U_\pi \spec(Z) U_\pi^\dagger$. However, since $X$ is extremal, $cX =   U_\pi \spec(Z) U_\pi^\dagger$ for all $\pi \in S_d$ with $q_\pi > 0$. Thus, diagonal elements of $Z$ equals its spectrum and so the Frobenius norm satisfies $\sum_i Z_{ii}^2 = \sum_i \lambda_i^2 = \Tr(Z^2)$, where $\lambda_i \in \spec(Z)$.
Expanding $\Tr(Z^2) = \sum_i Z_{ii}^2 + \sum_{i \neq j} |Z_{ij}|^2$ implies $Z_{ij} = 0$ for all $i \neq j$. Therefore, $Z$ is diagonal, meaning that $Z = D(Z) = cX$. Thus, $\Lambda$ exposes $X$ in $AM$, and therefore $X \in \exp(AM)$. 

The union equality follows similarly to Theorem \ref{thm:AM-extremal} due to the unitary invariance of the set $\exp(AM)$ of exposed extreme rays.
\end{proof}

\begin{theorem}
\label{thm:AM-semialgebraic}
If $M \subset \M_d$ is a closed convex semialgebraic cone of Hermitian $d\times d$ matrices, then $AM$ is also closed convex semialgebraic cone. Furthermore, $\diag(AM)$ and $\spec(AM)$ are also semialgebraic.
\end{theorem}
\begin{proof} 
The closed convex cone nature of $AM$ follows trivially as discussed earlier. 

Suppose $M$ is semialgebraic. Then the condition $Y \in M$ is defined by a finite boolean combination of real polynomial inequalities. Since $Y$ is a complex Hermitian matrix, it can be decomposed into its real and imaginary components, identifying $\M_d$ with the real vector space $\RR^{d^2}$. Similarly, for any $U\in \mathcal{U}(d)$, the condition $U U^\dagger = \1$ decomposes into a finite set of real polynomial equations defined by the real and imaginary parts of $U$, allowing us to embed $\mathcal{U}(d)$ within the space $\RR^{2d^2}$. By definition of the absolute $M$-set,  
$$X \in AM \iff X\in \M_d \textrm{ and }\forall  U\in \mathcal{U}(d)  \ UXU^\dagger \in M.$$ The matrix multiplication map $\mathcal{U}(d) \times \M_d \to \M_d$ defined by $(U, X) \mapsto U X U^\dagger$ entirely consists of addition and multiplication of these real components, and is therefore a real polynomial mapping. Thus, we can write the absolute $M$-set membership as a first-order logic formula over the real closed field $\RR$ such that
\[
    X \in AM \iff \forall U \in \RR^{2d^2}\, \big( U \in \mathcal{U}(d) \implies U X U^\dagger \in M \big).
\]

By quantifier elimination \cite{tent2012course}, there exists a quantifier-free equivalent of the first-order logic formula seen above. And therefore by the Tarski-Seidenberg principle \cite[Theorem.~1.4.2]{BCR:RAG}, any set definable by a first-order formula over the reals is semialgebraic. Thus, $AM$ is semialgebraic.

It is straightforward to see that $\diag(AM)$ is semialgebraic. Since the space of diagonal matrices $D_d$ is simply a linear subspace of $\M_d$, which is trivially defined by linear equations, $D_d$ is semialgebraic. By definition, $\diag(AM)= AM \cap D_d$ is the intersection of two semialgebraic sets and so, $\diag(AM)$ is semialgebraic. Also, as we have established that $ \diag(AM) = \spec(AM)$, it follows directly. 
\end{proof}



\section{$\mathbf{\text{APPT}_{m,n}}$ is a spectrahedron}
\label{sec:spectrahedron}
In this section, we discuss in detail the precise geometric properties of the sets of spectra of absolute PPT ($\mathrm{APPT}_{m,n}$) and absolute separable states ($\mathrm{ASEP}_{m,n}$). To do so, we first define the sets in a more general framework beyond the non-increasingly ordered spectra.

\subsection{The set of spectra of absolute PPT states}
First, recall that the set of absolute PPT states admits a complete characterization: a mixed state $\rho \in  \M_m \otimes \M_n$ with eigenvalues $\lambda_1\geq \lambda_2 \geq \cdots \geq \lambda_{mn}\geq 0$ is absolute PPT if and only if a finite system of linear matrix inequalities defined with respect to the eigenvalues $\lambda_i$ hold true \cite{hildebrand2007positive}. Although, this formulation is exact in itself, its direct application to characterize higher-dimensional absolute PPT states becomes intractable because in general, the minimum number of such matrix inequalities grows exponentially as the dimensions increase. Furthermore, using the strict eigenvalue ordering obscures the underlying symmetries of the state space. As such, the set of spectra (and equivalently the set of states), comprised of all unordered spectra of similar type, becomes difficult to properly understand. However, we identify that the underlying structure governing these defining linear matrix constraints is the permutation orbit of a single base matrix. With this, we provide a more definitive way to characterize the set of spectra of absolute PPT states and in some cases the set of spectra of absolute separable states. 

To characterize this symmetry, let $\Delta_{N-1} \subset \mathbb{R}^N$ denote the standard probability simplex defined as 
\begin{equation*}
  \Delta_{N-1} 
    = \left\{ \mathbf{x} \in \RR^N \,\Bigg|\, x_i \geq 0 \text{ for all } i = 1,2,\ldots,N, \text{ and } \sum_{i=1}^N x_i = 1 \right\}.
\end{equation*}
We consider the unordered vector of eigenvalues of a state $\rho \in \M_N$ to be the spectrum  $\lambda = (\lambda_1, \lambda_2, \dots, \lambda_N)^T \in \Delta_{N-1}$. Define a permutation bijection $\pi \in S_{mn}$ such that $\pi: \{1, \dots, mn\} \to \{1, \dots, mn\}$ acts on the coordinate indices of the unordered spectrum $\lambda \in \mathrm{APPT}_{m,n}$. Without loss of generality, we take $m\leq n$. For any permutation $\pi\in S_{mn}$, we define the permuted $m\times m$ symmetric matrix $L_\pi(\lambda)$ such that every entry $[L_\pi(\lambda)]_{(i,j)}$ with $1 \le i \le j \le m$ is evaluated as 
\begin{equation}\label{eq:base_matrix}
    [L_{\pi}(\lambda)]_{(i,j)} = 
    \begin{cases} 
        2\lambda_{\pi(p(i,i))} & \text{for } i = j,  \\
        \lambda_{\pi(p(i,j))} - \lambda_{\pi(q(i,j))} & \text{for } i < j.
    \end{cases}
\end{equation}
where the index mappings $q(i,j)$ and $p(i,j)$ are defined as 
 \begin{equation}
   \label{eq:index_functions}
   q(i,j) = \frac{(i-1)(2m-i)}{2} + (j-i) 
   \quad \text{ and } \quad
   p(i,j) = mn - \frac{(i-1)(2m-i+2)}{2} - (j-i).
\end{equation} 
The functions $q(i,j)$ and $p(i,j)$ enumerate the strict upper triangular entries sequentially starting from the index $1$, and the upper triangular entries (including the diagonal) in reverse order from the maximum index $mn$, respectively. Since the number of elements appearing in each row forms a decreasing arithmetic progression, the terms $q(i,j)$ and $p(i,j)$ are exactly the partial sums of the preceding row entries, offset accordingly for a given column index $j$.

This mapping explicitly constructs the matrix 
\begin{equation*}
    L_\pi(\lambda) 
     = \begin{pmatrix}
        2\lambda_{\pi(mn)} & \lambda_{\pi(mn-1)} - \lambda_{\pi(1)} & \cdots & \lambda_{\pi(m(n-1)+1)} - \lambda_{\pi(m-1)} \\
        \lambda_{\pi(mn-1)} - \lambda_{\pi(1)} & \ddots & \cdots & \cdots \\
        \vdots & \vdots & \ddots & \vdots \\
        \lambda_{\pi(m(n-1)+1)} - \lambda_{\pi(m-1)} & \cdots & \cdots & 2\lambda_{\pi\left(mn - \left(\frac{m(m+1)}{2} - 1\right)\right)}
    \end{pmatrix}_{m\times m.}
\end{equation*}
The two sets of indices defined by the functions in Eq.~\eqref{eq:index_functions} are disjoint since $mn\geq m^2$ and their total cardinality is $\binom{m}{2}+\binom{m+1}{2} = m^2$.
As such, any one of the matrices depends on exactly $m^2$ of the spectral coordinates, while $mn-m^2$ are unused within the matrix description.

Hildebrand \cite{hildebrand2007positive} characterized the absolute PPT spectra via a finite number of linear matrix inequalities generated by compatible orderings of eigenvalues of decomposable witnesses. Using the matrix definition in Eq.~\eqref{eq:base_matrix}, we present a symmetric reformulation of that absolute PPT criterion as follows: 

\begin{theorem}
\label{thm:APPT-full}
Suppose $\lambda \in \Delta_{mn-1}$ is the spectrum of a bipartite state and assume $m\leq n$. Then $\lambda \in\mathrm{APPT}_{m,n}$ if and only if  $L_{\pi}(\lambda) \ge 0,$ for all $\pi \in S_{mn}$. 
\end{theorem}

\begin{proof}
Suppose $\lambda$ is the spectrum of a bipartite state. Then $\lambda\in \mathrm{APPT}_{m,n}$ if and only if for all unitaries $U\in \mathcal{U}(mn)$, $(U\diag(\lambda)U^\dagger)^\Gamma \geq 0$. By definition, $(U\diag(\lambda)U^\dagger)^\Gamma \geq 0$ if and only if $\bra{\psi}(U\diag(\lambda)U^\dagger)^\Gamma \ket{\psi} \geq 0$
for $\ket{\psi}\in \CC^m\otimes \CC^n.$ This is equivalent to $$\Tr (U\diag(\lambda)U^\dagger\, (\ket{\psi}\!\bra{\psi})^\Gamma)\geq 0,$$ under the Hilbert-Schmidt inner product. Now suppose $x_1,x_2\ldots ,x_{m} \in \RR_{\geq 0}$ are the Schmidt coefficients of $\ket{\psi}$, then $(\ket{\psi}\!\bra{\psi})^\Gamma$ has eigenvalues 
\begin{equation}
 \label{eq:witness_eigenvals}
     x^2_i \text{ for } 1\leq i\leq m, \qquad \pm x_ix_j \text{ for } 1\leq i < j\leq m,  \qquad 0 \text{ with multiplicity } mn-m^2,
\end{equation}
following \cite{hildebrand2007positive} (see also \cite[Lemma~1]{johnston2018inverse}). Let  $\mathbf{v}=(v_1,\ldots,v_{m})^T\in \RR^{m}$ and set $x_i\coloneq |v_i|$. Define a vector $\mu(\mathbf{v})\in\RR^{mn}$ element-wise as
\begin{equation*}
\mu_k(\mathbf{v}) =
     \begin{cases}
         \phantom{-}v_i^2 &\text{for } k = p(i,i) \text{ and } 1\leq i\leq m, \\
         \phantom{-}v_i v_j &\text{for } k = p(i,j) \text{ and } 1\leq i<j\leq m, \\
         -v_i v_j &\text{for } k = q(i,j) \text{ and } 1\leq i<j\leq m, \\
         \phantom{-}0 &\text{otherwise}.
     \end{cases}
\end{equation*} 
Notice that for $k = p(i,j)$ where $i< j,$ if $v_iv_j \geq 0$, then $ \mu_{p(i,j)}(\mathbf{v}) = |v_iv_j|, $ and if $v_iv_j < 0$, then $ \mu_{p(i,j)}(\mathbf{v}) = -|v_iv_j| $. Similarly, for $k = q(i,j)$,  $ \mu_{q(i,j)}(\mathbf{v}) = -|v_iv_j|$ if 
$v_iv_j \geq 0$ and  $ \mu_{q(i,j)}(\mathbf{v}) = |v_iv_j|,$ if 
$v_iv_j < 0$. Thus, for every $i< j,$ $\{v_iv_j, -v_iv_j\} = \{|v_iv_j|,-|v_iv_j|\}$. As such, 
$\mu(\mathbf{v}) $ defines a fixed ordering of the eigenvalues of $(\ket{\psi}\!\bra{\psi})^\Gamma$.
Thus, for a fixed Schmidt coefficient vector $\mathbf{x}$, it follows that (see the proof of \cite[Lemma~3]{hildebrand2007positive})  $$\Tr(U\diag(\lambda)U^\dagger\, \ket{\psi}\!\bra{\psi}^\Gamma)\geq 0 \iff \sum_{k=1}^{mn}\lambda_{\pi(k)}\mu_k(\mathbf{v})\geq 0,$$
for all permutations $\pi\in S_{mn}$ and $\mu_k(\mathbf{v})\in \RR$. The latter inequality reduces to
\begin{align*}
\sum_{k=1}^{mn}\lambda_{\pi(k)}\mu_k(\mathbf{v}) = \sum_{i=1}^m \lambda_{\pi(p(i,i))}v_i^2+\sum_{1\leq i<j\leq m} (\lambda_{\pi(p(i,j))}-
    \lambda_{\pi(q(i,j))}) v_iv_j= \frac{1}{2}\mathbf{v}^T L_\pi(\lambda)\mathbf{v}\geq 0,
\end{align*} 
and since this holds now for all $\textbf{v}\in\RR^m$, we get $L_\pi(\lambda)\geq 0$.
\end{proof}

In this formulation, the matrix inequalities $L_\pi(\lambda)\geq 0$ are permutation-invariant and $\mathrm{APPT}_{m,n}$ is an $S_{mn}$-invariant convex body defined without any reliance on the ordering constraints on the spectra. Although we can define $\mathrm{APPT}_{m,n}$ over the full symmetric group $S_{mn}$, it is sufficient to restrict the constraints over a subset $\widetilde{S} \subset S_{mn}$ of size $|\widetilde{S}| = {mn\choose m^2}(m^2)! = \frac{(mn)!}{(mn-m^2)!}$ since the matrix structure only requires $m^2$ eigenvalues out of the total $mn$ eigenvalues to satisfy the constraint. As such, any permutation acting only on the remaining $(mn-m^2)$ eigenvalues generates identical matrix constraints.
Thus, even though the set $\widetilde{S}$ may not necessarily be the minimal set needed for the criteria, it simply removes the redundant permutations from unused spectra. Additional redundant permutations could also arise from equivalences within the matrix structure. 

Therefore, we identify the set of spectra of absolute PPT states as the spectrahedron
\begin{equation}
\label{eq:APPT_spectrahedron}
    \mathrm{APPT}_{m,n} = \Delta_{mn-1}\cap\left\{ \lambda \in \RR^{mn} \ \Bigg\vert \  \bigoplus_{\pi \in \widetilde{S}} L_{ \pi}(\lambda) \geq 0 \right\}. 
\end{equation}
Equivalently, 
\begin{equation}
        \label{eq:APPT_spectrahedron_equiv}
        \mathrm{APPT}_{m,n}
    = \left\{ \lambda \in \Delta_{mn-1} \ \Bigg\vert \ \mathcal{L}(\lambda) =\bigoplus_{\pi \in \widetilde{S}} L_{ \pi}(\lambda) \geq 0 \right\} 
    = \Delta_{mn-1}\cap \mathcal{L}^{-1} \left(\bigoplus_{\pi \in \widetilde{S}} \mathcal{S}^m_+ \right),
\end{equation}
where $\mathcal{S}^m_+$ is the set of $m \times m$ positive semidefinite matrices.
Notice that the size of the block matrix $\mathcal{L} (\lambda)$ is
\begin{equation}
  \sum_{\pi\in \widetilde{S}} m 
    = |\widetilde{S}|\cdot  m = \frac{m(mn)!}{(mn-m^2)!}.
\end{equation} 
The spectral sets $ \mathrm{APPT}_{m,n}$ are basic closed semialgebraic convex sets as they are spectrahedra. Since $\mathrm{APPT}_{2,n} =\mathrm{ASEP}_{2,n}$ \cite{johnston2013sepfromspectra},(see \cite{johnston2018inverse} for initial discussion of ) the same property applies for the set of spectra of absolute separable states whenever $m = 2, n\geq 2$ (see Figure \ref{fig:ASEP2X2} for $\mathrm{ASEP}_{2,2}$). 
\begin{figure}[ht]
    \centering
    \includegraphics[width=0.6\linewidth]{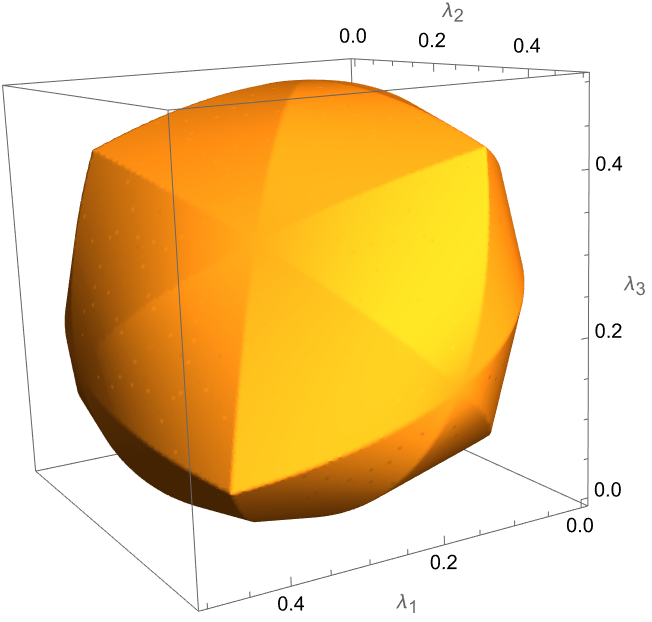}
    \caption{The absolute separable spectrahedron $\mathrm{ASEP}_{2,2}$  }
    \label{fig:ASEP2X2}
\end{figure}
Although we cannot say for certain whether or not this property carries over for $\mathrm{ASEP}_{m,n}$ whenever $m,n \geq 3$, we can deduce the following as seen for absolute $M$-sets in Theorem \ref{thm:AM-semialgebraic}: 

\begin{theorem}
\label{thm:ASEP-semialgebraic}
The set of spectra of absolutely separable states, $\mathrm{ASEP}_{m,n}$, is a convex and compact semialgebraic set.
\end{theorem}
\begin{proof}
Let $\mathrm{AS}_{m,n}\subset \M_m \otimes \M_n$ denote the set of absolute separable states. From \cite{ganguly2014witness}, the set $\mathrm{AS}_{m,n}$ is shown to be a convex and compact subset of the set of separable states.  We identify the spectral set as $\spec(\mathrm{AS}_{m,n}) = \mathrm{ASEP}_{m,n}\subset \Delta_{mn-1} $ such that from Eq.~\eqref{eq:AM-diag=spec}, $\mathrm{ASEP}_{m,n} =\diag(\mathrm{AS}_{m,n}) = \mathrm{AS}_{m,n} \cap D_{mn}$ where $D_{mn}$ is the set of all diagonal real matrices. Since both $\mathrm{AS}_{m,n}$ and $D_{mn}$ are convex, their intersection, $\mathrm{ASEP}_{m,n}$ is  convex. 
For compactness, it follows easily as the set of diagonal matrices $D_{mn}$ is topologically closed in $\M_{mn}$ and $\mathrm{AS}_{m,n}$ is a compact set. Therefore, their intersection $\mathrm{ASEP}_{m,n}$ is compact. 

The set of separable states, as the convex hull of the Segre variety of pure product states, is well-known to be semialgebraic \cite{fawzi2021set}. Thus, from Theorem \ref{thm:AM-semialgebraic}, it follows directly that $\spec(\mathrm{AS}_{m,n}) =\mathrm{ASEP}_{m,n}$ is also semialgebraic. 
\end{proof}

Consequently, the dual cones, $ \mathrm{ASEP}^*_{m,n}$ and $ \mathrm{APPT}^*_{m,n}$, which correspond to sets of spectra of \emph{absolute separability witnesses}
and \emph{absolute $\mathrm{PPT}$ witnesses} \cite{johnston2018inverse,ganguly2014witness}, respectively, are also convex semialgebraic sets. In particular, $\mathrm{APPT}^*_{m,n}$ is a spectrahedral shadow\cite{cynthiaspectrahedra}. Indeed, these dual cones equal the convex hull of the set of spectra of block positive matrices and decomposable block positive matrices, respectively \cite[Corollary~1]{johnston2018inverse}. In Figure \ref{fig:duals}, we illustrate the set of spectra of two-qubit decomposable block positive matrices, satisfying the polynomials inequalities from \cite[Theorem~3]{johnston2018inverse} and its the convex hull, corresponding to the polar set $\mathrm{ASEP}^\circ_{2,2}$ with respect to the spectrum of the maximally mixed state as origin. We normalize the spectral tuple $(\omega_1,\omega_2,\omega_3,\omega_4)$ so that $\sum_{i=1}^4 \omega_i = 1$, and may assume $\omega_1, \omega_2, \omega_3\geq 0$. 
 


\begin{figure}[ht]
    \centering
    \begin{subfigure}{0.48\textwidth}
        \includegraphics[width=\textwidth]{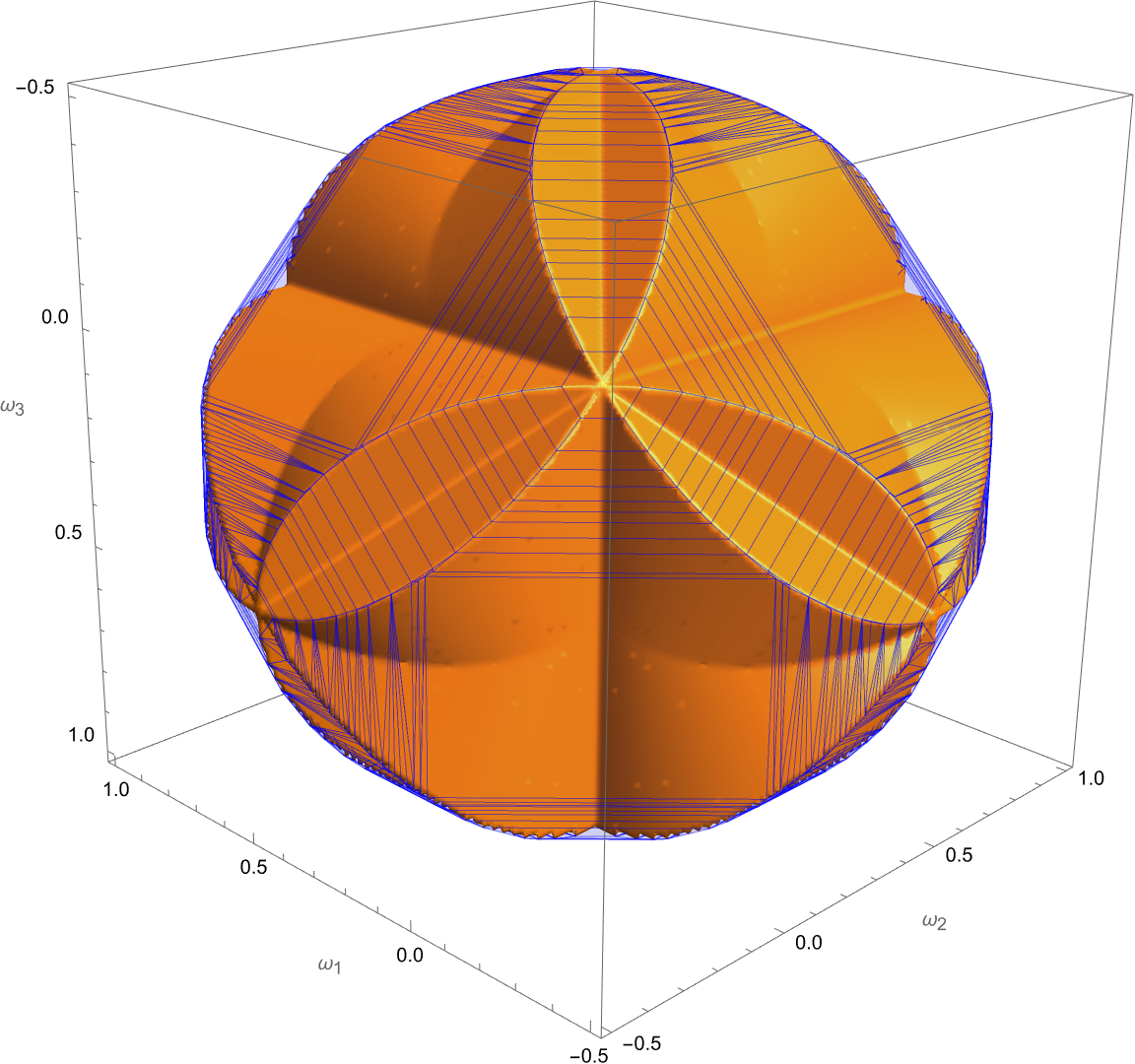}
        \caption{\label{fig:duals_a}}
    \end{subfigure}
    \begin{subfigure}{0.48\textwidth}
        \includegraphics[width=\textwidth]{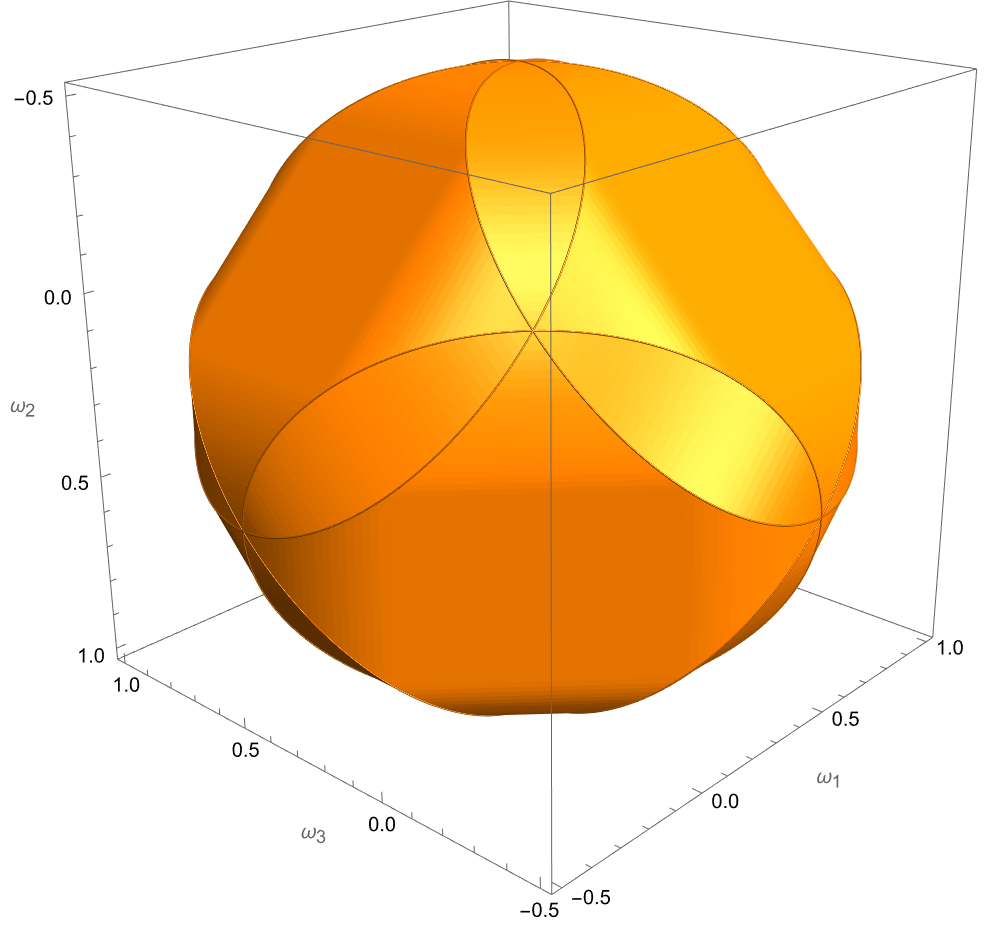}
        \caption{\label{fig:duals_b}}
    \end{subfigure}
    \caption{On the left (Fig.~\ref{fig:duals_a}), the inscribed object (orange) depicts the set of normalized spectra of two-qubit decomposable block-positive matrices \cite[Theorem~3]{johnston2018inverse} and the enclosing shell (purple), is the convex hull of this object.  This convex hull is exactly the normalized polar witness set $\mathrm{ASEP}_{2,2}^\circ$ (Fig.~\ref{fig:duals_b}) of the spectrahedron in Figure~\ref{fig:ASEP2X2}.}
    \label{fig:duals}
\end{figure}


\subsection{The \texorpdfstring{$\mathrm{APPT}_{m,n}$}{APPT\_{m,n}} boundary}\label{subsec:boundary}

Let $\partial \mathrm{APPT}_{m,n}$ denote the topological boundary of the set $\mathrm{APPT}_{m,n}$. Since $\mathrm{APPT}_{m,n}$ is defined over the probability simplex, we consider all topological notions relative to the affine hull of the simplex, $\aff(\Delta_{mn-1})$. Thus, we define the topological boundary $\partial \mathrm{APPT}_{m,n}$ as the set difference
$\partial \mathrm{APPT}_{m,n} = \mathrm{APPT}_{m,n} \setminus \mathrm{int}(\mathrm{APPT}_{m,n})$ since  $\mathrm{APPT}_{m,n}$ is closed. Because $\mathrm{APPT}_{m,n}$ is a spectrahedron and the spectrum of the maximally mixed state is strictly feasible within the set, the interior is given by $$\mathrm{int}(\mathrm{APPT}_{m,n})=
\left\{\lambda\in\operatorname{int}(\Delta_{mn-1}) \Big| L_\pi(\lambda)>0 \text{ for every }\pi\in\widetilde S \right\}.$$  

\begin{theorem}\label{thm:boundary_characterization}Suppose $\lambda$ is the spectrum of an bipartite state in $\M_m\otimes \M_n$.
  The spectrum  $\lambda \in \mathrm{APPT}_{m,n}$ lies on the topological boundary $\partial \mathrm{APPT}_{m,n}$ if and only if there exists at least one permutation $\pi\in \widetilde{S}$ such that $\det(L_\pi(\lambda)) =0$. Equivalently, 
  $$   \partial\mathrm{APPT}_{m,n} =\left\{\lambda\in\mathrm{APPT}_{m,n}\mid\det(\mathcal{L}(\lambda))=0\right\}. $$
\end{theorem}
\begin{proof}
    Since a feasible spectrum lies in the interior of $\mathrm{APPT}_{m,n}$ if and only if $\mathcal{L }$ is positive definite at that point, it follows that any spectrum $\lambda\in \partial \mathrm{APPT}_{m,n} $ if and only if $\det(\mathcal{L}) = 0$. Since $\mathcal{L}(\lambda)$ is block-diagonal and particularly singular on the boundary, its determinant decomposes as 
\begin{equation}
\label{eq:determinant-decomp}
 \det(\mathcal{L}(\lambda)) = \det\left(\bigoplus_{\pi \in \widetilde{S}} L_{ \pi}(\lambda)\right) = \prod_{\pi \in \widetilde{S}} \det(L_\pi(\lambda)).
\end{equation} Thus, it follows directly that  $\prod_{\pi \in \widetilde{S}} \det(L_\pi(\lambda)) = 0$  if and only if there exists at least one permutation $\pi\in\widetilde{S}$, $\det(L_\pi(\lambda)) =0$.
\end{proof}

The index sets defined by Eq.~\eqref{eq:index_functions} are disjoint, implying that the upper-triangular and diagonal entries of $L_\pi(\lambda)$ have distinct coordinates. Thus, the linear forms $[L_\pi(\lambda)]_{i,j} $ are linearly independent and may be regarded as independent variables, say $x_{i,j}:=[L_\pi(\lambda)]_{i,j} $ for $1\leq i\leq j\leq m$. In this notation, each $L_\pi(\lambda) $ is simply the generic symmetric $m\times m$ matrix $L_\pi(\lambda)=\{x_{i,j}\}_{i,j=1}^m $, i.e., no special structure is imposed in this form. Since the determinant of a generic symmetric matrix is irreducible, and
irreducibility is preserved under an invertible linear change of variables,
$\det (L_\pi(\lambda))$ is irreducible for each $\pi\in\widetilde{S}$. Furthermore, since all matrix entries here are linear forms, the matrix determinant $\det(L_{\pi}(\lambda))$  for each permutation $\pi \in \widetilde{S}$, is a homogeneous polynomial of degree $m$. As such the topological boundary $\partial \mathrm{APPT}_{m,n}$ is contained in the finite union of irreducible hypersurfaces generated by the determinant equations $\det(L_{\pi}(\lambda))=0$ within the probability simplex, after repeated polynomials are eliminated:
\begin{equation}
\label{eq:loci_union}
    \partial \mathrm{APPT}_{m,n} \subseteq \Delta_{mn-1} \cap \left( \bigcup_{\pi \in \widetilde{S}} \mathcal{Z}(\det(L_{\pi}(\lambda))) \right).
\end{equation}
The proofs of the subsequent results in the section can be found in Appendix \ref{app:Appendix_A1}.

\begin{lemma}
\label{lem:polynomial_orbit}
The set of determinant polynomials 
$P = \{ \det(L_\pi(\lambda))\}_{\pi \in S_{mn}}$ of $\mathrm{APPT}_{m,n}$ is the group orbit under $S_{mn}$ given by
\begin{equation*}
P = \mathrm{Orb}_{S_{mn}}(\det(L_{\mathrm{id}}(\lambda)))
\end{equation*}
where $\det(L_{\mathrm{id}}(\lambda))$ is the determinant polynomial under the identity permutation. 
\end{lemma}

\begin{lemma}
\label{lem:signed-permutation}
As polynomials, $\det(L_{\pi}(\lambda)) = \det(L_{id}(\lambda))$ for some permutation $\pi \in S_{mn}$, if and only if there exists a signed permutation matrix $A$ such that $L_{\pi}(\lambda) = A L_{id}(\lambda) A^T$.
\end{lemma}

We determine a bound on the number of permutations needed to describe the boundary as follows:

\begin{theorem}
\label{thm:kappa-irreducible-components} 
Let $\partial \mathrm{APPT}_{m,n}$ be  the topological boundary of the set of absolute $\mathrm{PPT}$ spectra and $V_{mn}=\{\lambda\in \RR^{mn}\mid \sum_{i=1}^{mn}\lambda_i=1 \}$, the affine hyperplane of normalized spectra. Then there exist permutations $ \pi_1,\ldots,\pi_{\kappa_{m,n}}\in \widetilde{S}$ such that

\begin{equation}\label{eq:boundary_hypersurface}
\partial\mathrm{APPT}_{m,n}=\mathrm{APPT}_{m,n}\cap\left(\bigcup_{i=1}^{\kappa_{m,n}}
    \mathcal{Z}_{V_{mn}} \left(
        \det L_{\pi_i}(\lambda)
    \right)\right),
\end{equation}
where 
$$\kappa_{m,n} = \frac{(mn)!}{(mn-m^2)!\cdot 2^{m-1} \cdot m!}.$$  
Furthermore, each algebraic set $\mathcal{Z}_{V_{mn}} \left(\det L_{\pi_i} \right)$ taken with respect to the affine space $V_{mn}$ 
is a distinct irreducible hypersurface of dimension $(mn-2)$.
\end{theorem}

\begin{remark}
\label{rem:boundary_relation}
The number of distinct irreducible determinant polynomials in the boundary representation \eqref{eq:boundary_hypersurface} satisfies
\begin{equation*}
    \kappa_{m,n} = \binom{mn}{m^2} \kappa_{m,m}.
\end{equation*} 
This demonstrates the choice of $m^2$ eigenvalues from the total $mn$ eigenvalues required by the matrix constraint $L_\pi(\lambda)$.

\end{remark}

\begin{figure}[ht]
    \centering
    \includegraphics[width=0.75\linewidth]{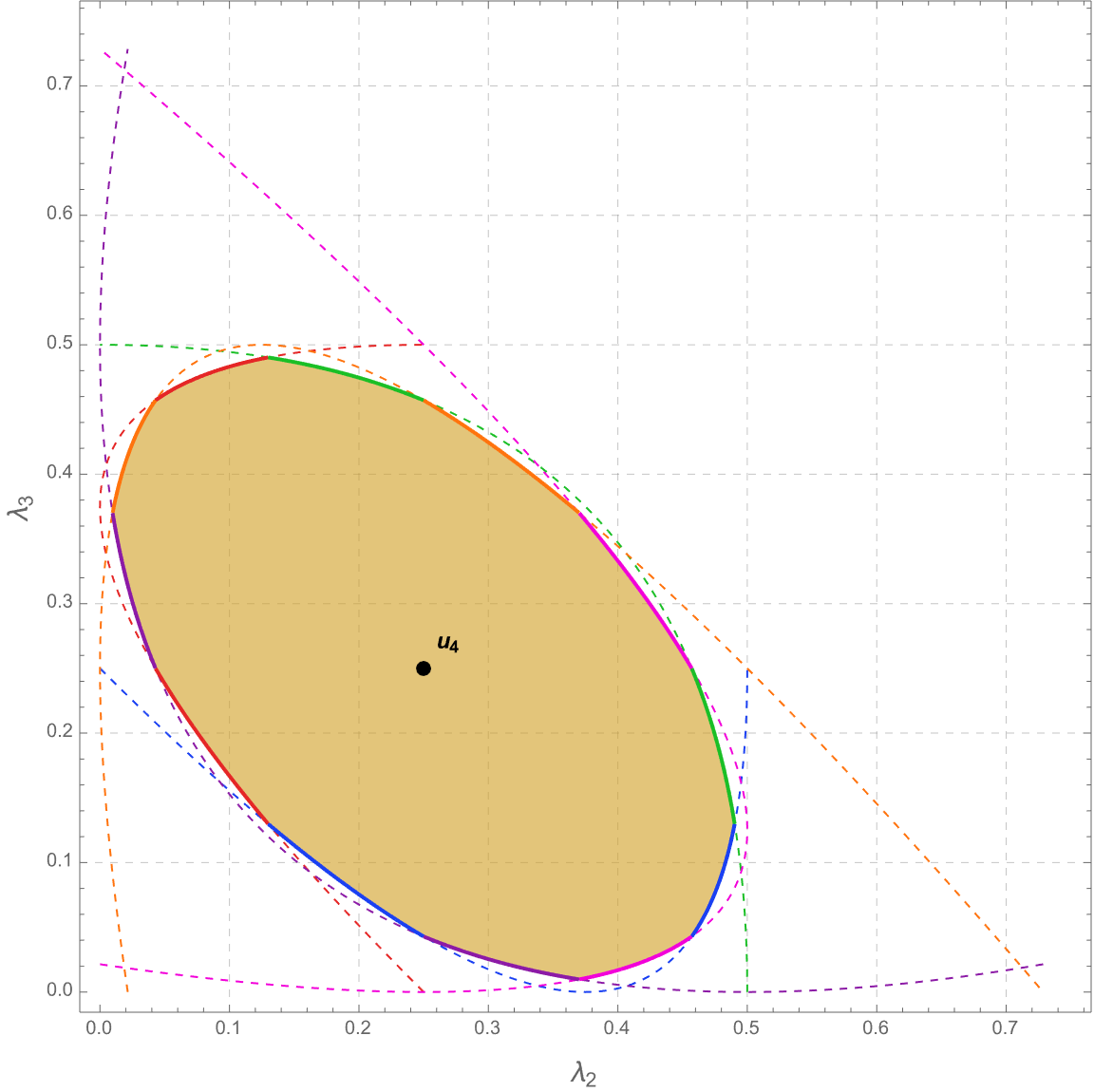}
    \caption{A cross-section of  $\mathrm{ASEP}_{2,2}$ at $\lambda_{1} = 0.25$ whose boundary is generated by $\kappa_{2,2} = 6$ irreducible hypersurfaces, each of dimension $2$ (brightly colored). The interior (yellow) includes the center $\mathbf{u}_4 = \frac{1}{4}(1,1)$ corresponding to the spectra of the maximally mixed state. Each  colored curve is a hypersurface $ \mathcal{Z}_{V_{mn}} \left(
        \det L_{\pi_i}
    \right) $ for $i=1,\ldots,6$ in the slice. The feasible regions of the topological boundary are indicated by thick brightly colored parts of the curves  while the infeasible regions outside the set are the dotted colored parts. }
    \label{fig:boundary}
\end{figure}


\section{Faces and extreme points of $\mathbf{\text{APPT}_{m,n}}$ }\label{sec:faces and extreme points}
The faces of spectrahedra are determined by the images, or equivalently kernels, of their defining matrices, a property inherited from the set of positive semidefinite matrices $\mathcal{S}^N_+$ \cite{ramana1995some,cynthiaspectrahedra}. More precisely, there exists a natural inclusion-preserving bijection between the non-empty faces of a spectrahedron and its associated image subspaces \cite[Prop.~2.10]{scheiderer2022extreme}. Taking orthogonal complements gives an equivalent inclusion-reversing correspondence with the associated kernel subspaces. This bijection adapted into the kernel-based description of spectrahedra implies the following: 

\begin{lemma}
\label{lem:minimal_subspace}
Let  $F $ be a non-empty face of a spectrahedron. Suppose $F$ is defined by the common kernel subspace $U = \bigcap_{x \in F} \ker(\mathcal{L}(x))$, its relative interior $\mathrm{relint}(F)$ is given by  
\begin{equation}
    \mathrm{relint}(F) = \{ y \in F \;|\; \ker(\mathcal{L}(y)) = U \}.
\end{equation}
\end{lemma}

\begin{proof}
Let $W =\sum_{x \in F} \mathrm{im}(\mathcal{L}(x))  $ be the image subspace of $F$. By \cite[Cor.~2.11]{scheiderer2022extreme}, 
$$\mathrm{relint}(F) = \{ y \in F \;|\; \mathrm{im}(\mathcal{L}(y)) = W \}.$$ 
Since $\mathcal{L}(x) $ is a real symmetric matrix, $\ker(\mathcal{L}(x)) = \mathrm{im}(\mathcal{L}(x))^\perp$. Taking the orthogonal complement of the image subspace, we see that $U =  \bigcap_{x \in F} \ker(\mathcal{L}(x)) =\bigcap_{x \in F} \mathrm{im}(\mathcal{L}(x))^\perp =\left( \sum_{x \in F} \mathrm{im}(\mathcal{L}(x)) \right)^\perp= W^\perp.$ Consequently, 
$$\mathrm{im}(\mathcal{L}(y)) = W \iff \ker(\mathcal{L}(y)) = U.$$ 
Thus, $\mathrm{relint}(F) = \{ y \in F \;|\; \ker(\mathcal{L}(y)) = U \}$.
\end{proof}

With this, we can define the faces of the set of spectra of absolute PPT states $\mathrm{APPT}_{m,n}$ as follows: 

\begin{theorem} \label{thm:faces_of_APPT}
Let $F_\mathcal{U}$ denote a face of $\mathrm{APPT}_{m,n}$ and let $\mathcal{U} = \{U_\pi\}_{\pi \in \widetilde{S}}$ be its associated collection of common kernel subspaces, where $U_\pi = \ker(L_\pi(\lambda_0)) \subseteq \mathbb{R}^m$ for $\lambda_0 \in \mathrm{relint}(F_\mathcal{U})$. Then the face $F_\mathcal{U} \subseteq \mathrm{APPT}_{m,n}$ is uniquely characterized as 
\begin{equation}\label{eq:exhaustive_face}
    F_{\mathcal{U}} = \left\{ \lambda \in \mathrm{APPT}_{m,n} \;\middle|\; U_{\pi} \subseteq \ker(L_{\pi}(\lambda)), \quad \forall\ \pi \in \widetilde{S} \right\}.
\end{equation}
\end{theorem}
More precisely, the kernel constraints associated with each face can be characterized as follows:

\begin{theorem}
\label{thm:face_matrix}
Let $F_\mathcal{U}$ be a proper face of $\mathrm{APPT}_{m,n}$, $\mathcal{U} = \{U_\pi\}_{\pi \in \widetilde{S}}$ its associated collection of kernel subspaces and choose a basis $\mathcal{B}_\pi$ of $U_\pi$. Define the set of active permutations as $\mathcal{K} \coloneq \{\pi \in \widetilde{S} \mid 0<\dim(U_\pi) <m  \}$.
For each permutation $\pi \in\mathcal{K}$ and each basis vector $\mathbf{u}^{(\pi)} = (u_1^{(\pi)}, \dots, u_m^{(\pi)})^T \in \mathcal{B}_\pi$, we can also define the $(m \times mn)$ matrix  $W_{\pi, \mathbf{u}^{(\pi)}}$ element-wise as
\begin{equation}
\label{eq:submatrices_of_C}
        [W_{\pi, \mathbf{u}^{(\pi)}}]_{(i,k)} = 
        \begin{cases} 
            2u_i^{(\pi)} & \text{if } k = \pi(p(i,i)) \\ 
            u_j^{(\pi)} & \text{if } \exists\; j\neq i \text{ such that }   k = \pi(p(i\wedge j, i\vee j))  \\ 
            -u_j^{(\pi)} & \text{if } \exists\;j\neq i \text{ such that } k = \pi(q(i\wedge j, i\vee j)) \\ 
            0 & \text{otherwise}
        \end{cases}
\end{equation}
where  $i\wedge j = \min(i,j) $, $i \vee j = \max(i,j) $ and the index functions $p(i,j), q(i,j)$ are defined as in Eq.~\ref{eq:index_functions} for $1 \le i \le j \le m$.
Then for the $\left(m \sum_{\pi \in \mathcal{K}} \dim(U_\pi) \times mn \right)$ row block matrix $C \coloneq \begin{bmatrix} W_{\pi, \mathbf{u}^{(\pi)}} \end{bmatrix}_{\substack{\pi \in \mathcal{K} \\ \mathbf{u}^{(\pi)} \in \mathcal{B}_\pi}}$ obtained by vertically stacking $W_{\pi, \mathbf{u}^{(\pi)}} $, the face $F_{\mathcal{U}}$ is characterized as 
\begin{equation}
    F_{\mathcal{U}} = \left\{ \lambda \in \mathrm{APPT}_{m,n} \;\middle|\;C\cdot \lambda = 0 \right\}.
\end{equation}
\end{theorem}

\begin{proof}
Following Theorem \ref{thm:faces_of_APPT}, $\lambda \in \mathrm{APPT}_{m,n}$ belongs to the face $F_{\mathcal{U}}$ if and only if  $U_\pi \subseteq \ker(L_\pi(\lambda))$ for every permutation $\pi \in \widetilde{S}$.
We will consider this condition over active and inactive permutations. 

Let  $\mathcal{K} = \{\pi \in \widetilde{S} \mid 0<\dim(U_\pi) <m  \}$ denote the set of active permutations of the face $F_{\mathcal{U}}$. Then, for any inactive permutation $\pi \notin \mathcal{K}$, by definition, $\dim(U_\pi) = 0$ or $\dim(U_\pi) = m.$ Notice that if $\dim(U_\pi) = 0$, the subspace is trivial, $U_\pi = \{0\}$ and the kernel constraint $\{0\}\subseteq \ker(L_\pi(\lambda))$ is trivially satisfied for all $\lambda\in\mathrm{APPT}_{m,n}$. Therefore, no constraints are enforced. In fact, when $\mathcal{K}= \emptyset$  and $U_\pi = \{0\}$ for all $\pi \in \widetilde{S}$, every $\lambda\in\mathrm{APPT}_{m,n}$ belongs in the face $F_{\mathcal{U}}$ so that we obtain the trivial face $F_{\mathcal{U}} = \mathrm{APPT}_{m,n}$. On the other hand, if $\dim(U_\pi) = m,$ we have that $U_\pi = \mathbb{R}^m$. This forces $L_\pi(\lambda) = 0$ so that all diagonal entries $\lambda_{\pi(p(i,i))} =0$ for all $1 \le i \le m$. But $\lambda \in \text{APPT}_{m,n}$ strictly has at most one eigenvalue being equal to zero \cite[Prop.~1]{johnston2014separability} (see also \cite[Prop.~7.3]{jivulescu2015positive}), thus $\lambda\notin \text{APPT}_{m,n}$. Therefore, if there exists $\pi \in \widetilde{S}$ such that $\dim(U_\pi) = m$, the face is the empty set, $F_{\mathcal{U}} = \emptyset.$ Thus, the facial characterization of proper faces reduces to the active permutations such that
\begin{equation}\label{eq:face_equivalence}
    \lambda \in F_{\mathcal{U}} \iff U_\pi \subseteq \ker(L_\pi(\lambda)), \quad \forall \pi \in \mathcal{K} \text{ and } U_\pi = \ker(L_\pi(\lambda_0)), \; \lambda_0 \in \mathrm{relint}(F_\mathcal{U}).
\end{equation}

Now consider any permutation $\pi \in \mathcal{K}$ such that each minimal kernel subspace $U_\pi$, is determined by a chosen basis $\mathcal{B}_\pi$. Then for $U_\pi \subseteq \ker(L_\pi(\lambda))$, we must have that \begin{equation} 
    L_\pi(\lambda) \mathbf{u}^{(\pi)} = 0 , \quad \forall \mathbf{u}^{(\pi)} \in \mathcal{B}_\pi, \; \forall \pi \in \mathcal{K}.
\end{equation}
This reduces to a system of linear equations or equivalently a vector with each component being $ [L_\pi(\lambda) \mathbf{u}^{(\pi)}]_i =0$ for $1\leq i\leq m.$ We can decompose the $i$-th component of the vector as 
\begin{align*}
    [L_\pi(\lambda) \mathbf{u}^{(\pi)}]_i &=\sum_{j=1}^{m}[L_\pi(\lambda)]_{(i,j)} u^{(\pi)}_j\\
    &= \underbrace{\sum_{j=1}^{i-1} [L_\pi(\lambda)]_{(i,j)} u^{(\pi)}_j}_{j<i}+   \underbrace{[L_\pi(\lambda)]_{(i,i)}u^{(\pi)}_i}_{j=i}  +\underbrace{\sum_{j=i+1}^m [L_\pi(\lambda)]_{(i,j)} u^{(\pi)}_j}_{j>i}.
\end{align*}
Since $L_\pi(\lambda)$ is symmetric, $[L_\pi(\lambda)]_{(i,j)} = [L_\pi(\lambda)]_{(j,i)}$. Thus, the $i$-th component is given by
\begin{align}\label{eq:expanded_sum}
[L_\pi(\lambda) \mathbf{u}^{(\pi)}]_i&= [L_\pi(\lambda)]_{(i,i)} u^{(\pi)}_i + \sum_{j=i+1}^m [L_\pi(\lambda)]_{(i,j)} u^{(\pi)}_j + \sum_{j=1}^{i-1} [L_\pi(\lambda)]_{(j,i)} u^{(\pi)}_j\nonumber \\
&=  2\lambda_{\pi(p(i,i))} u^{(\pi)}_i + \sum_{j=i+1}^m \left( \lambda_{\pi(p(i,j))} - \lambda_{\pi(q(i,j))} \right)u^{(\pi)}_j + \sum_{j=1}^{i-1} \left( \lambda_{\pi(p(j,i))} - \lambda_{\pi(q(j,i))} \right) u^{(\pi)}_j.
\end{align}
For all $1\leq i\le m$, we can express this system as the matrix equation $W_{\pi,\mathbf{u}^{(\pi)}}\lambda = 0$ where $W_{\pi,\mathbf{u}^{(\pi)}}$ is a sparse $m \times mn$ matrix such that the $i$-th row of the matrix equation is
\begin{equation}
    [W_{\pi, \mathbf{u}^{(\pi)}} \lambda]_i = \sum_{k=1}^{mn} [W_{\pi, \mathbf{u}^{(\pi)}}]_{(i,k)} \lambda_k = [L_\pi(\lambda) \mathbf{u}^{(\pi)}]_i.
\end{equation}
By inspecting the coefficients of the eigenvalues $\lambda_k$ in Eq.~\eqref{eq:expanded_sum}, we can derive the entries of $W_{\pi,\mathbf{u}^{(\pi)}}$ in terms of the row index $i$ and the column index $k$.  Observe that for index $j=i$, the coefficient of $\lambda_{\pi(p(i,i))}$ is $2u^{(\pi)}_i$. 
Thus, if $k = \pi(p(i,i))$, $[W_{\pi, \mathbf{u}^{(\pi)}}]_{(i,k)} = 2u^{(\pi)}_i$. Similarly, if there exists for indices $j>i$, $\lambda_{\pi(p(i,j))}$ has coefficient $u^{(\pi)}_j$, and the eigenvalue $\lambda_{\pi(q(i,j))}$ has coefficient $-u^{(\pi)}_j$. Thus, since the elements $[W_{\pi, \mathbf{u}^{(\pi)}}]_{(i,k)}$ are defined independent of index $j$,  
\begin{equation*}
  [W_{\pi,\mathbf{u}^{(\pi)}}]_{(i,k)} 
    = \begin{cases} 
      \phantom{-}u^{(\pi)}_j &\text{if } \exists \; j > i \text{ s.t. } k = \pi(p(i,j)), \\
      -u^{(\pi)}_j &\text{if }\exists \; j > i \text{ s.t. }  k = \pi(q(i,j)).
\end{cases}
\end{equation*} 
On the other hand, for indices $j<i$, $\lambda_{\pi(p(j,i))}$ has coefficient $u^{(\pi)}_j$, and $\lambda_{\pi(q(j,i))}$ has coefficient $-u^{(\pi)}_j$. Hence, 
\begin{equation*}
  [W_{\pi,\mathbf{u}^{(\pi)}}]_{(i,k)}
      = \begin{cases}
      \phantom{-}u^{(\pi)}_j &\text{if }\exists \; j < i \text{ s.t. }  k = \pi(p(j,i)), \\
-u^{(\pi)}_j &\text{if }\exists \; j < i \text{ s.t. }  k = \pi(q(j,i)).
\end{cases} 
\end{equation*} 
Notice that for both the cases where $j<i$ and $j>i$, the index maps $p$ and $q$ take $\min(i,j)$ as their first entry and $\max(i,j)$ as the second entry. Thus, altogether, we see that the matrix $W_{\pi,\mathbf{u}^{(\pi)}}$ has entries 
\begin{equation}
    [W_{\pi, \mathbf{u}^{(\pi)}}]_{(i,k)} = 
    \begin{cases} 
        \phantom{.}2u^{(\pi)}_i & \text{if } k = \pi(p(i,i)) \\ 
        \phantom{-}u^{(\pi)}_j & \text{if }\exists \; j\neq i \text{ such that } k = \pi(p(i\wedge j, i\vee j))  \\ 
        -u^{(\pi)}_j & \text{if }\exists \; j\neq i \text{ such that } k = \pi(q(i\wedge j, i\vee j))  \\ 
        \phantom{-;}0 & \text{otherwise}. 
    \end{cases}
\end{equation}
    This construction therefore allows that each permutation generates the system $W_{\pi,\mathbf{u}^{(\pi)}}\lambda = L_\pi(\lambda) \mathbf{u}^{(\pi)} = 0.$ Thus, for all active permutations $\pi\in \mathcal{K}$ and every corresponding basis vector $\mathbf{u}^{(\pi)} \in \mathcal{B}_\pi$, we can vertically concatenate the matrices $W_{\pi, \mathbf{u}^{(\pi)}}$, generating a larger system of linear equations such that for the $\left( m \sum_{\pi \in \mathcal{K}} \dim(U_\pi) \times mn \right)$-matrix $C$, 
\begin{equation*}
    C=  \begin{bmatrix} W_{\pi, \mathbf{u}^{(\pi)}} \end{bmatrix}_{\substack{\pi \in \mathcal{K} \\ \mathbf{u}^{(\pi)} \in \mathcal{B}_\pi}}\end{equation*}
where each  $W_{\pi, \mathbf{u}^{(\pi)}}$ is a row-block of $C$. This large matrix characterizes exactly all the minimal kernel constraints that uniquely define the proper face $F_{\mathcal{U}}$. 
\end{proof}

By the relation established for the general set of spectra $\mathrm{APPT}_{m,n}$ in Eq.~\eqref{eq:APPT_spectrahedron}, we can equally define an equivalent isomorphism for faces $F_\mathcal{U}\subseteq \mathrm{APPT}_{m,n}$ such that 
\begin{equation}\label{eq:face_isomophism}
F_\mathcal{U}  \cong \mathcal{L}\left(\mathrm{aff}(\Delta_{mn-1})\right) \cap \hat{F}_\mathcal{U}
\end{equation}
 where $\hat{F}_\mathcal{U} =\bigoplus_{\pi \in \widetilde{S}} \{A\in\mathcal{S}^m_+: U_\pi \subseteq\ker(A)\}$ is a face of the positive semidefinite cone associated with minimal kernel subspaces $U_\pi$ \cite[Lemma~4.5]{cynthiaspectrahedra}. In fact, this isomorphism is established because $\mathcal{L}$ is an injective affine mapping onto its image. Thus, by applying the pullback through its pre-image $\mathcal{L}^{-1}$, we have that 
 \begin{equation*}F_\mathcal{U} =\Delta_{mn-1} \cap \mathcal{L}^{-1}(\hat{F}_\mathcal{U}).\end{equation*}

\begin{proposition}
\label{prop:face_dimensions}
Let $F_{\mathcal{U}} \subseteq \mathrm{APPT}_{m,n}$ be a face defined by kernel subspaces $\mathcal{U} = \{U_\pi\}_{\pi \in \widetilde{S}}$, and let  $C $ be its associated block constraint matrix. Then 
\begin{equation}
    \dim(F_{\mathcal{U}}) = (mn - 1) - \mathrm{rank}(C)
\end{equation}  
where $\mathrm{rank}(C) \in\{0\} \cup \{m,m+1,\ldots,mn-1\}$, is the number of linear independent equations defining the face.
\end{proposition}

\begin{proof}  
By \cite[Prop.~2.14]{scheiderer2022extreme} and the isomorphism of faces as seen in Eq.~\eqref{eq:face_isomophism}, the face $F_{\mathcal{U}} \subseteq \mathrm{APPT}_{m,n}$ has equal dimension to its affine space 
$$\dim(F_{\mathcal{U}}) = \dim \left(\mathrm{aff}(\Delta_{mn-1}) \cap \mathcal{L}^{-1}\left(\bigoplus_{\pi \in \widetilde{S}} \{A\in\mathcal{S}^m_+: U_\pi \subseteq\ker(A)\}\right)\right).$$

To evaluate this, let $\hat{F}_\mathcal{U} = \mathcal{L}^{-1}\left(\bigoplus_{\pi \in \widetilde{S}} \{A\in\mathcal{S}^m_+: U_\pi \subseteq\ker(A)\}\right)$ be the convex cone such that the associated linear subspace $V \coloneq \{\lambda\in\RR^{mn}\mid U_\pi \subseteq \ker(L_\pi(\lambda)) \textrm{ for all } \pi\in \widetilde{S}\}$. Then, by Theorem \ref{thm:face_matrix}, this is equivalent to $C\cdot \lambda = 0$, where $ C=  \begin{bmatrix} W_{\pi, \mathbf{u}^{(\pi)}} \end{bmatrix}_{\substack{\pi \in \mathcal{K} \\ \mathbf{u}^{(\pi)} \in \mathcal{B}_\pi}}$ is the block constraint matrix generated by the active constraints $\pi \in \mathcal{K}$. Thus, the linear subspace $V$ is exactly the kernel of $C$, $V = \ker(C)$. 

Now let $\delta = \mathrm{rank}(C)$ denote the total number of linearly independent equations generated by the active constraints $\pi \in \mathcal{K}$. Then, we find that the dimension of the pre-image set $V$ is 
$$\dim(V) = \dim(\mathbb{R}^{mn}) - \mathrm{rank}(C) = mn - \delta.$$ Since the affine span of the non-empty face is defined exactly as the intersection of the $(mn-\delta)$-dimensional linear subspace $V$ passing through the origin and the affine hyperplane $\textrm{aff}(\Delta_{mn-1})$ defined by $\sum_i \lambda_i =1$, (which does not include the origin), the dimension of the affine span reduces the dimension of $V$ by $1.$ Therefore, 
  \begin{equation*}
        \dim(F_{\mathcal{U}}) = \dim(V) - 1 = (mn - 1) - \delta.
    \end{equation*}

To determine the possible range of values for $\delta$, we consider how the rank of the block matrix $C$ changes with respect to the active permutations in $\mathcal{K}$. 

First, notice that if $\mathcal{K} = \emptyset$, no constraints are imposed and therefore $\delta = \textrm{rank}(C) = 0$. This corresponds to the face of dimension $\dim(F_{\mathcal{U}})= mn-1$, the trivial face $\mathrm{APPT}_{m,n}$.
Now assume the face $F_{\mathcal{U}}$ is a proper face and there is at least one active permutation $\pi\in \mathcal{K}$ where $\dim(U_\pi)\geq1$. Then there exists at least one non-zero basis vector $\mathbf{u}^{(\pi)} \in \mathcal{B}_\pi$ which generates the $m \times mn$ block matrix $W_{\pi, \mathbf{u}^{(\pi)}}$ as a submatrix of $C$. Suppose there exists  $\mathbf{c} = (c_1,...,c_m) \in \RR^m$ such that for each $i$-th row $\sum_{i=1}^{m}c_i [W_{\pi, \mathbf{u}^{(\pi)}}]_i =0^T$.
Then if $c_i = 0$ for all $1\leq i\leq m$, then the rows $[W_{\pi, \mathbf{u}^{(\pi)}}]_i$ are linearly independent. Since the basis vector $\mathbf{u}^{(\pi)}$ is non-zero, there exists at least one row index $s$ such that $u^{(\pi)}_s\neq  0.$
By construction, at the column index $k = \pi(p(s,s))$, only that $s$-th row contains the non-zero entry $2u^{(\pi)}_s$. Thus, $c_s(2u^{(\pi)}_s) =0$ implying that $c_s =0$. 
Now consider any other row with index $t$ where $s\neq t$. At the column index $k = \pi(p(s\wedge t, s\vee t))$, only rows $s$ and $t$ contain non-zero entries $u^{(\pi)}_t$ and $u^{(\pi)}_s$, respectively. Thus, $\sum_{i=1}^{m}c_i [W_{\pi, \mathbf{u}^{(\pi)}}]_i = c_su^{(\pi)}_t + c_t u^{(\pi)}_s =0 $. Since it is established that $c_s = 0,$ and $u_s\neq 0$, it must be that $c_t = 0$. Thus, since $c_i = 0$ for all $1\le i\le m$, the rows of $W_{\pi, \mathbf{u}^{(\pi)}}$ are linearly independent and so $\textrm{rank}(W_{\pi, \mathbf{u}^{(\pi)}}) = m$.  
Therefore, the matrix $ C=  \begin{bmatrix} W_{\pi, \mathbf{u}^{(\pi)}} \end{bmatrix}_{\substack{\pi \in \mathcal{K} \\ \mathbf{u}^{(\pi)} \in \mathcal{B}_\pi}}$ must have total rank $\delta = \mathrm{rank}(C)\geq m.$

Finally since the minimum possible dimension of a non-empty proper face is $0$, we must have that the $\dim(F_{\mathcal{U}}) = (mn-1)-\delta\geq 0$, hence $\delta =\mathrm{rank}(C)\leq mn-1$. Thus, altogether, the rank of $C$ must have 
$$\delta \in \{0\} \cup \{m, m+1, \ldots, mn-1\},$$ with $\delta = 0$ corresponding to the full set $\mathrm{APPT}_{m,n}$.
\end{proof}

\begin{corollary}
\label{cor:face_relation}
Let $F_{\mathcal{U}}^{(m,m)} $ be a face of $\mathrm{APPT}_{m,m}$ defined by the kernel subspaces $\mathcal{U}$, and let $F_{\mathcal{U}}^{(m,n)} $ (where $m < n$) be the corresponding face in $\mathrm{APPT}_{m,n}$ defined by the same kernel subspaces and a fixed active set of $m^2$ coordinates. If, up to permutations of columns,  $C^{(m,n)} = [C^{(m,m)}\;\; 0]$, then 
\begin{equation*}
    \dim(F_{\mathcal{U}}^{(m,n)}) = \dim(F_{\mathcal{U}}^{(m,m)}) + m(n-m).
\end{equation*}
\end{corollary}

\begin{proof}
Recall that each block operator $L_\pi(\lambda)$ requires  exactly $m^2$ distinct components of the spectrum $\lambda$ due to the index functions defined in Eq.~\eqref{eq:index_functions}. Since both faces are characterized by the same kernel subspaces $\mathcal{U}$, the constraint matrix $C^{(m,n)}$ of the face $F_{\mathcal{U}}^{(m,n)} $ is identical to 
$C^{(m,m)}$ for its non-zero entries but also includes $mn-m^2$ columns of zeros. Due to this, $\mathrm{rank}(C^{(m,n)}) = \mathrm{rank}(C^{(m,m)})$ and it follows directly from Proposition \ref{prop:face_dimensions} that
\begin{align*}
    \dim(F_{\mathcal{U}}^{(m,n)}) 
        &= (mn - 1) - \mathrm{rank}(C^{(m,n)}) \\
        &= (mn - m^2) + (m^2 - 1) - \mathrm{rank}(C^{(m,m)}) \\
        &= m(n-m) + \dim(F_{\mathcal{U}}^{(m,m)}),
\end{align*}
and we are done.
\end{proof}

\begin{corollary}
\label{cor:maximal_face_single_line}
Suppose there exists an active permutation $\sigma\in\widetilde{S}$ such that for  $\lambda_0\in\relint(\mathrm{APPT}_{m,n})$ and a nonzero vector $\mathbf{u} \in\RR^m $, $\ker(L_\sigma(\lambda_0))=\mathrm{span}\{\mathbf{u}\} $ and $L_\pi(\lambda)>0$ for all permutations $\pi\in\widetilde{S}\setminus\{\sigma\}$. Then,
\begin{equation}\label{eq:max_face}
    F^{\max}=\left\{\lambda\in\mathrm{APPT}_{m,n}:L_\sigma(\lambda)\mathbf{u}=0 \right\},
\end{equation}
is a maximal proper face of $\mathrm{APPT}_{m,n}$.
\end{corollary}

\begin{proof}
The expression in Eq.~\eqref{eq:max_face} follows directly from the proof of Theorem \ref{thm:face_matrix}. To prove maximally, suppose that $E$ is a face satisfying 
\begin{equation}\label{eq:face_inclusion}
    F^{\max} \subseteq E \subseteq \mathrm{APPT}_{m,n}
\end{equation} and define the associated  collection of common-kernel subspaces for all $\pi\in \widetilde{S}$ such that
$$\mathcal{U} \coloneq \bigcap_{\lambda_0\in \relint(F^{\max})} \ker (L_\pi(\lambda_0))\qquad \textrm{ and } \quad \mathcal{V}\coloneq \bigcap_{\lambda_0\in \relint(E)} \ker (L_\pi(\lambda_0)).$$ Let $U_\pi\coloneq\ker (L_\pi(\lambda_0))$ for $\lambda_0\in \relint(F^{\max})$ and $V_\pi\coloneq\ker (L_\pi(\lambda_0))$ for $\lambda_0\in \relint(E)$. By the inclusion in Eq.~\eqref{eq:face_inclusion}, we know that $V_\pi\subseteq U_\pi $ for all $\pi\in\widetilde{S}$. By assumption, for the active permutation $\sigma\in\widetilde{S},$ $U_\sigma = \mathrm{span}\{\mathbf{u}\} $ and $U_{\pi^*} = \{0\} $ for all $\pi^*\in\widetilde{S}\setminus\{\sigma\}$. It follows that   $V_\sigma\subseteq \mathrm{span}\{\mathbf{u}\}$ and $V_{\pi^*} = \{0\}$ for $\pi^*\in\widetilde{S}\setminus\{\sigma\}$.
Since $\mathrm{span}\{\mathbb{u}\}$ is one-dimensional, either $V_\sigma= \mathrm{span}\{\mathbf{u}\} $ or $V_\sigma=\{0\}$. In the first case, since $V_\sigma = U_\sigma = \mathrm{span}\{\mathbf{u}\} $ and $V_{\pi^*} = U_{\pi^*} =\{0\} $ for $\pi^*\in\widetilde{S}\setminus\{\sigma\}$, it follows that for all $\pi\in\widetilde{S}$, $V_{\pi} = U_{\pi} $. Hence, the faces $F^{\max}$ and $E$ have the same common-kernel subspace, i.e., $\mathcal{U}=\mathcal{V}$, implying that $F^{\max} = E$. In the second case where $V_\sigma=\{0\}$, $V_\pi=\{0\}$ for all $\pi\in\widetilde{S}.$ Thus, the face $E$ has the trivial set as its associated common-kernel subspace implying  $E = \mathrm{APPT}_{m,n}$. Therefore, no proper face lies between $F^{\max}$ and $\mathrm{APPT}_{m,n}$. And hence, $F^{\max}$ is a maximal proper face. 
\end{proof}

Thus, a maximal face is obtained by allowing exactly one matrix constraint $L_\sigma(\lambda)$ (which is singular) to acquire a one-dimensional kernel while all other matrix constraints which are not congruent to $L_\sigma(\lambda)$ remain strictly positive definite. Hence, the corresponding kernel conditions provides exactly $m$ linearly independent equations, and no additional equations are generated locally. As such, we observe the following:

\begin{corollary}
\label{cor:maximal_faces}
The maximal proper faces $F^{\max} \subset\mathrm{APPT}_{m,n}$ have dimension $mn - m - 1$.
\end{corollary}

\begin{proof}
Suppose $\sigma\in\widetilde{S}$ is the active permutation and $\mathbf{u}$ is a nonzero vector associated with the maximal face $F^{\max}\subseteq \mathrm{APPT}_{m,n} $. Then, by Theorem~\ref{thm:faces_of_APPT} the corresponding matrix $W_{\sigma,\mathbf{u}}$ constitutes a row block of the matrix $C$. As shown in the proof of Proposition~\ref{prop:face_dimensions}, $m = \mathrm{rank}(W_{\sigma,\mathbf{u}} ).$ For all other inactive permutations $\pi\in\widetilde{S}\setminus\{\sigma\}$, their corresponding matrix blocks satisfy $L_\pi(\lambda)>0$ and their associated kernel subspaces with respect to $\relint(\mathrm{APPT}_{m,n})$ is trivial (see Eq.~\eqref{eq:face_equivalence}). Thus, they do not contribute to the rank of the total matrix $C_{F^{\max}}$ associated with the maximal face $F^{\max}$.  As such, 
$$\mathrm{rank}(C_{F^{\max}}) = \mathrm{rank}(W_{\sigma,\mathbf{u}} )= m.$$
Thus, it follows from Proposition~\ref{prop:face_dimensions} that 
$\dim(F^{\max} ) = (mn-1) - \mathrm{rank}(C_{F^{\max}}) = mn-m-1$.
  
\end{proof}

\begin{corollary}
\label{cor:maximal_faces_2xn}
Every maximal proper face $F^{\max}$ of $\mathrm{ASEP}_{2,n}$ is defined by the system
\begin{equation} 
    F^{\max} = \left\{ \lambda \in \mathrm{ASEP}_{2,n} \;\middle|\; \begin{aligned}
        \lambda_{\sigma(2n-2)} &= \alpha^2 \lambda_{\sigma(2n)}\\
        \lambda_{\sigma(2n-1)} - \lambda_{\sigma(1)} &= -2\alpha\lambda_{\sigma(2n)}
    \end{aligned} \right\},
\end{equation}
where $\alpha = \frac{u_1}{u_2} \in \RR$ and $u_2\neq0$ for a permutation $\sigma \in\widetilde{S}$ and a nonzero vector  $\mathbf{u}=(u_1, u_2)^T$.
\end{corollary}

\begin{proof}
Since the maximal face $F^{\max}$ is defined by a single permutation $\sigma \in \widetilde{S}$ and $\mathbf{u}=(u_1, u_2)^T$, Theorem \ref{thm:face_matrix} tells us that $F^{\max}$ is defined by the system of equations
$$C \lambda =  W_{\sigma, \mathbf{u}} \lambda = 0.$$ 
Notice that the $i-$th row of the vector $W_{\sigma, \mathbf{u}} \lambda$ reduces to 
\begin{align*}
 0 &=[W_{\sigma, \mathbf{u}} \lambda]_i =\sum_{k=1}^{2n} [W_{\sigma, \mathbf{u}} ]_{(i,k)}\lambda_k,  \\
 0 &=2\lambda_{\sigma(p(i,i))} u_i + \sum_{j=i+1}^m \left( \lambda_{\sigma(p(i,j))} - \lambda_{\sigma(q(i,j))} \right)u_j + \sum_{j=1}^{i-1} \left( \lambda_{\sigma(p(j,i))} - \lambda_{\sigma(q(j,i))} \right) u_j.
\end{align*}
Thus, we obtain the explicit system of equations
\begin{align}
        2u_1\lambda_{\sigma(2n)} + u_2(\lambda_{\sigma(2n-1)} - \lambda_{\sigma(1)}) &= 0, \label{eq:2x2_row1} \\
        u_1(\lambda_{\sigma(2n-1)} - \lambda_{\sigma(1)}) + 2u_2\lambda_{\pi(2n-2)} &= 0, \label{eq:2x2_row2}
\end{align}
which resolves to  
$$\lambda_{\sigma(2n-1)} - \lambda_{\sigma(1)} = -2\left(\frac{u_1}{u_2}\right)\lambda_{\sigma(2n)} \quad \textrm{ and } \quad \lambda_{\sigma(2n-2)} = \left(\frac{u_1}{u_2}\right)^2 \lambda_{\sigma(2n)}.$$
Taking $\alpha =\frac{u_1}{u_2},$ we obtain the desired result.
\end{proof}

\begin{figure}[ht]
    \centering
    \includegraphics[width=0.4\textwidth]{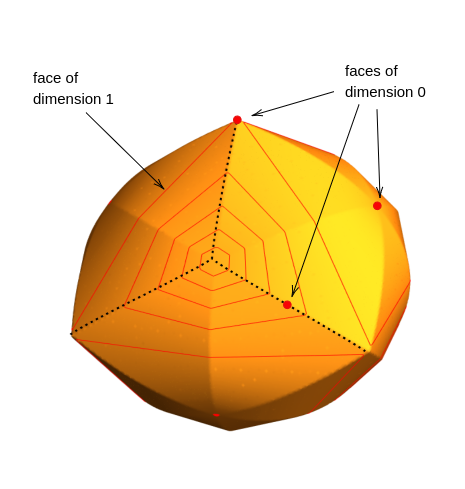}
        \caption{Some extreme points and  maximal faces of $\mathrm{ASEP}_{2,2}$ defined by $\alpha\in \RR$  (Corollary \ref{cor:maximal_faces_2xn}).}
    \label{fig:faces}
\end{figure}

\begin{theorem}
\label{thm:mxn_face_system}
Let $F^{\max}$ be a maximal proper face of $\mathrm{APPT}_{m,n}$ characterized by an active permutation $\sigma \in \widetilde{S}$ and a nonzero vector $\mathbf{u} \in \mathbb{R}^m$. Assuming without loss of generality that $u_m \neq 0$. Then $F^{\max}$ is defined by the system \begin{equation}
    2 \alpha_i \lambda_{\sigma(p(i,i))} + \sum_{j=1}^{i-1} \alpha_j \left(\lambda_{\sigma(p(j,i))} - \lambda_{\sigma(q(j,i))}\right) + \sum_{j=i+1}^m \alpha_j \left(\lambda_{\sigma(p(i,j))} - \lambda_{\sigma(q(i,j))}\right) = 0
\end{equation}
for  $1 \leq i\leq m$ and $\alpha_k = \frac{u_k}{u_m}\in \RR$ for $1\leq k \leq m-1$, with $\alpha_m = 1$.

\end{theorem}

\begin{proof}
This follows similarly as in the proof of Corollary \ref{cor:maximal_faces_2xn}. Each row of the system $C \lambda =  W_{\sigma, \mathbf{u}} \lambda = 0$ reduces to 
\begin{align*}
[W_{\sigma, \mathbf{u}} \lambda]_i
=2\lambda_{\sigma(p(i,i))} u_i +\sum_{j=1}^{i-1} \left( \lambda_{\sigma(p(j,i))} - \lambda_{\sigma(q(j,i))} \right) u_j+  \sum_{j=i+1}^m \left( \lambda_{\sigma(p(i,j))} - \lambda_{\sigma(q(i,j))} \right)u_j  =0.
    \end{align*}
Since $u_m\neq 0,$ we divide through by the factor so that the system becomes 
\begin{align*}
    &2 \left( \frac{u_i}{u_m}\right) \lambda_{\sigma(p(i,i))} + \sum_{j=1}^{i-1} \left( \frac{u_j}{u_m}\right)\left(\lambda_{\sigma(p(j,i))} - \lambda_{\sigma(q(j,i))}\right)+ \sum_{j=i+1}^{(m-1)} \left( \frac{u_j}{u_m}\right) \left(\lambda_{\sigma(p(i,j))} - \lambda_{\sigma(q(i,j))}\right) \\
    &+\left(\lambda_{\sigma(p(i,m))} - \lambda_{\sigma(q(i,m))}\right)= 0.
\end{align*}
Taking $\alpha_k = \frac{u_k}{u_m}$ for $1\leq k \leq m-1$, and $\alpha_m = 1$, the result follows directly.
\end{proof}


\begin{corollary}
\label{cor:exp-face}
Every face of the set of spectra of absolute PPT states, $\mathrm{APPT}_{m,n}$, is exposed.
\end{corollary}

\begin{proof}
This follows directly as $\text{APPT}_{m,n}$ is a spectrahedron.
\end{proof}

\begin{figure}[htbp]
    \centering
    \textbf{Schematic generation of faces of the $\mathrm{ASEP}_{2,n}$ spectrahedron} \\[1em]
    \resizebox{\textwidth}{!}{%
    \begin{tikzpicture}[
        node distance=1.5cm and 0.5cm,
        box/.style={draw, rectangle, minimum width=2.8cm, minimum height=1cm, align=center, font=\small},
        arrow/.style={-{Latex[length=2mm, width=2mm]}, thick}
    ]
        \node[box] (interior) { $\text{ASEP}_{2,n}$ \\ $d = 2n-1$ \\ $\delta = 0$ };
        
        \node[box, below=1.5cm of interior] (l1) { \textbf{Maximal Proper Face} \\ $d = 2n-3$ \\ $\delta = 2$ };
        
        \node[box, below left=2cm and 1cm of l1] (l2_left) { \textbf{Shared Diagonals} \\ $d = 2n-4$ \\ $\delta = 3$ };
        \node[box, below right=2cm and 1cm of l1] (l2_right) { \textbf{Disjoint/Off-Diag} \\ $d = 2n-5$ \\ $\delta = 4$ };
        
        \node[box, below left=2cm and 0.2cm of l2_left] (l3_1) { $d = 2n-5$ \\ $\delta = 4$ };
        \node[box, below right=2cm and 0.05cm of l2_left] (l3_2) { $d = 2n-6$ \\ $\delta = 5$ };
        \node[box, below left=2cm and 0.05cm of l2_right] (l3_3) { $d = 2n-6$ \\ $\delta = 5$ };
        \node[box, below right=2cm and 0.2cm of l2_right] (l3_4) { $d = 2n-7$ \\ $\delta = 6$ };

        \node[box, below=10cm of l1] (vertices) { \textbf{Extreme points} \\ $d = 0$ \\ $\delta = 2n-1$ };

        \draw[arrow] (interior) -- node[right, font=\footnotesize, align=center] { Add $\pi_1$ \\ $\Delta \delta = +2$} (l1);
        
        \draw[arrow] (l1) -- node[above left, font=\footnotesize, align=center] { Add $\pi_2$ \\ $\Delta \delta = +1$} (l2_left);
        \draw[arrow] (l1) -- node[above right, font=\footnotesize, align=center] { Add $\pi_2$ \\ $\Delta \delta = +2$} (l2_right);
        
        \draw[arrow] (l2_left) -- node[left, font=\footnotesize, align=center] { Add $\pi_3$ \\ $\Delta \delta = +1$} (l3_1);
        \draw[arrow] (l2_left) -- node[right, font=\footnotesize, align=center] { $\Delta \delta = +2$} (l3_2);
        \draw[arrow] (l2_right) -- node[left, font=\footnotesize, align=center] { Add $\pi_3$ \\ $\Delta \delta = +1$} (l3_3);
        \draw[arrow] (l2_right) -- node[right, font=\footnotesize, align=center] { $\Delta \delta = +2$} (l3_4);
        
        \draw[arrow, dashed] (l3_1) -- (vertices);
        \draw[arrow, dashed] (l3_2) -- (vertices);
        \draw[arrow, dashed] (l3_3) -- (vertices);
        \draw[arrow, dashed] (l3_4) -- (vertices);
    \end{tikzpicture}%
    }
    \caption{Schematic illustration of how faces of $\mathrm{ASEP}_{2,n}$ arise as active kernel constraints are added. Here $\delta=\mathrm{rank}(C)$ denotes the number of linearly independent equations defining the face, and $d=\dim(F_{\mathcal{U}})=(2n-1)-\delta$. Each added active permutation, $\pi_i$ may increase the rank by $1$ or $2$, thereby lowering the face dimension accordingly. Dashed arrows indicate the continuation of this process down to extreme points.}
\end{figure}
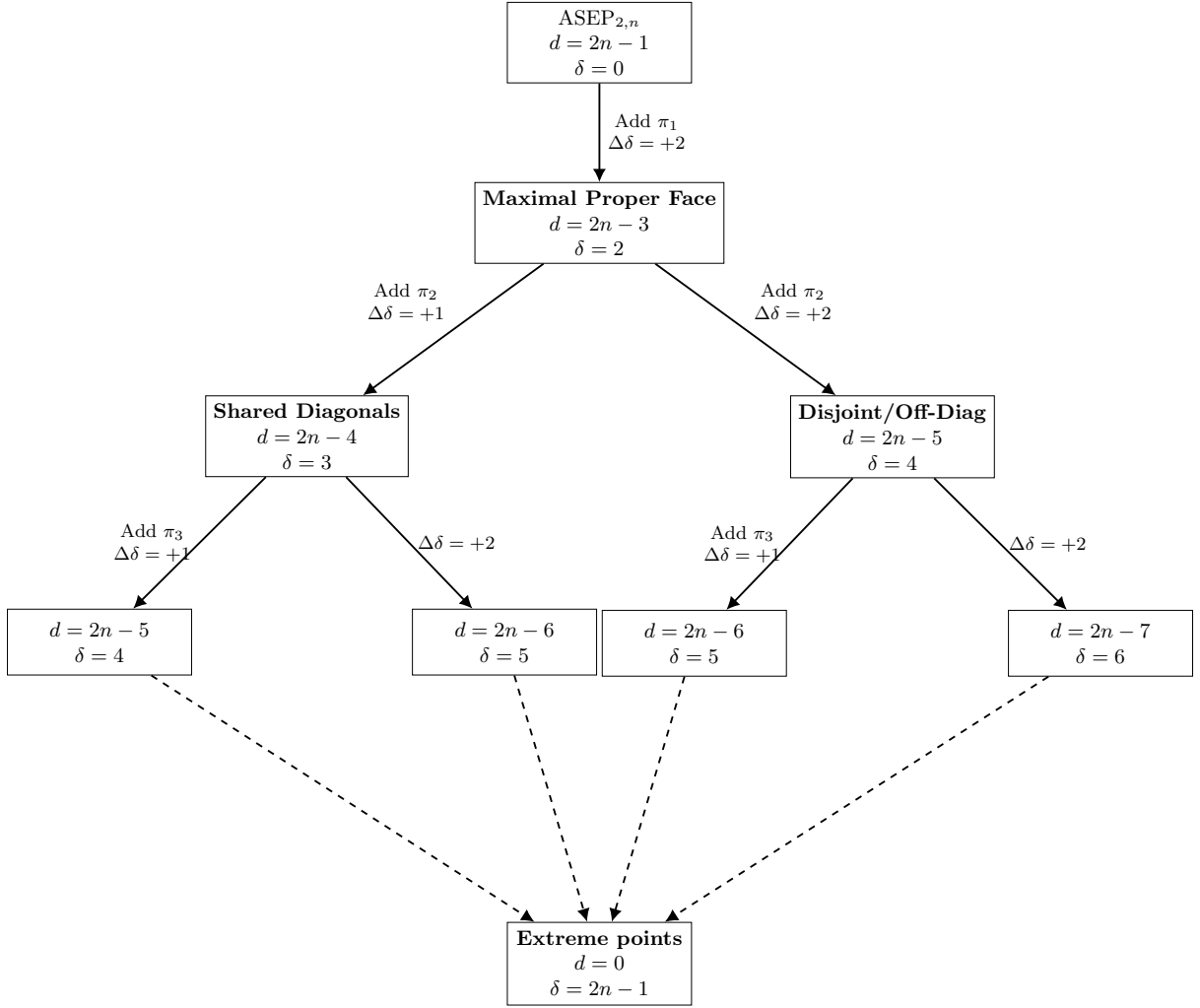

\subsection{Extreme points}
In this section, we provide a characterization of the boundary and extreme points of $\mathrm{APPT}_{m,n}$. 

\begin{corollary}\label{cor:boundary_characterization}
    Let $\lambda \in \mathrm{APPT}_{m,n}$ and let $C_\lambda$ be the kernel constraint matrix generated by its active kernel subspaces $U_\pi = \ker(L_\pi(\lambda))$. Then $\lambda$ lies on the topological boundary $\partial\mathrm{APPT}_{m,n}$ if and only if $ m \leq \mathrm{rank}(C_\lambda) $.
\end{corollary}

\begin{proof}
    Suppose $\lambda \in \partial\mathrm{APPT}_{m,n}$. By Theorem \ref{thm:boundary_characterization}, at least one of the constraint matrices is singular for the boundary and therefore satisfies $\det(L_{\pi_0}(\lambda) )= 0$ for some permutation $\pi_0 \in \widetilde{S}$. This implies that there exists a non-zero vector $\mathbf{u} \in U_{\pi_0}$ such that $L_{\pi_0}(\lambda) \mathbf{u} = 0 $. Thus, the corresponding kernel subspace is non-trivial such that $\dim(U_{\pi_0}) \ge 1$. As such, there always exists an active permutation so that $\mathcal{K}\neq \emptyset$ and $ m \leq \mathrm{rank}(C_\lambda).$
    
    Conversely, suppose $\mathrm{rank}(C_\lambda) \ge m$. Let $F_\mathcal{U}$ be the minimal face of $\mathrm{APPT}_{m,n}$ such that $\lambda\in \textrm{relint}(F_\mathcal{U})$. By Proposition \ref{prop:face_dimensions}, we know that 
    \begin{equation*}
        \dim(F_\mathcal{U}) = (mn - 1) - \mathrm{rank}(C_\lambda)  
        \leq (mn-1) -m   < mn-1. \end{equation*}
       Thus, $F_\mathcal{U}$ must be a proper face and since $\lambda\in \textrm{relint}(F_\mathcal{U})$, it implies that $\lambda \in \partial\mathrm{APPT}_{m,n}$. 
    \end{proof}

\begin{corollary}\label{cor:extreme_points}
For any spectrum $\lambda \in \mathrm{APPT}_{m,n}$, let $C_\lambda$ be the block constraint matrix generated by its active kernel subspaces $U_\pi = \ker(L_\pi(\lambda))$. Then $\lambda\in\mathrm{APPT}_{m,n}$ is an extreme point if and only if $ \mathrm{rank}(C_\lambda) = mn - 1.$
   
\end{corollary}

\begin{proof}
It follows directly from Proposition \ref{prop:face_dimensions}.
\end{proof}

The following last result of the present section stands apart from the flow of our investigation, in that it treats only the subset of $\mathrm{APPT}_{m,n}$ of probability vectors in non-increasing order, in other words the convex set 
\[
  \mathrm{APPT}_{m,n}^\downarrow 
  \coloneq \mathrm{APPT}_{m,n} \cap \Delta_{mn-1}^{\downarrow},
\]
with the ordered probability simplex $\Delta_{mn-1}^{\downarrow} \coloneq \{ \lambda \in \Delta_{mn-1}: \lambda_1 \ge \lambda_2 \ge \dots \ge \lambda_{mn} \}$. This is the set of Hildebrand's original characterization \cite{hildebrand2007positive}, and it allows us to relate the extreme points of $\mathrm{APPT}_{m,n}^\downarrow$ with those of $\mathrm{APPT}_{m,m}^\downarrow$. For this purpose, define the projection map $\tau:\RR^{mn}\rightarrow\RR^{m^2}$ acting as
\[
  \tau(\lambda) 
    = \left( \lambda_1,\ldots,\lambda_{\binom{m}{2}},
       \lambda_{mn+1-\binom{m+1}{2}},\ldots,\lambda_{mn} \right),
\]
\emph{i.e.}, it retains the first $\binom{m}{2}$ and the last $\binom{m+1}{2}$ coordinates. Evidently, $\tau$ maps $\Delta_{mn-1}^\downarrow$ to $\Delta_{m^2-1}^\downarrow$, up to scaling. 

\begin{theorem}
\label{thm:extreme_mn-ordered}
The projection $\tau$ maps $\mathrm{APPT}_{m,n}^\downarrow$ to $\mathrm{APPT}_{m,m}^\downarrow$ (up to scaling). 
In fact, for any extreme point $\lambda$ of $\mathrm{APPT}_{m,n}^\downarrow$, it holds necessarily that 
\begin{enumerate}
\item\label{item:ext_item_1} the rescaled $\lambda' \propto \tau(\lambda)$ is an extreme point of $\mathrm{APPT}_{m,m}^\downarrow$, and

\item\label{item:ext_item_2} for all $\binom{m}{2} < i \leq mn-\binom{m+1}{2}$, 
$\lambda_i \in \left\{ \lambda_{\binom{m}{2}}, \lambda_{mn+1 - \binom{m+1}{2}} \right\}$.
\end{enumerate}
\end{theorem}

\begin{proof}
The membership of $\lambda' = \frac{1}{r} \tau(\lambda)$ in $\mathrm{APPT}_{m,m}^\downarrow$ follows from Hildebrand's characterisation of $\mathrm{APPT}_{m,n}^\downarrow$: his complete set of matrix inequalities is constructed precisely from the largest $\binom{m}{2}$ and the smallest $\binom{m+1}{2}$ eigenvalues. 

Now, assume that $\lambda$ is an extreme point of $\mathrm{APPT}_{m,n}^\downarrow$, but by way of contradiction that property 1 or property 2 fails. 
We start with the latter:  

\medskip
\noindent\emph{Case 2:} If there is an index $\binom{m}{2} < i \leq mn-\binom{m+1}{2}$ with $\lambda_{\binom{m}{2}} > \lambda_i > \lambda_{mn+1 - \binom{m+1}{2}}$, then we can construct two distinct $\lambda_1,\lambda_2 \in \mathrm{APPT}_{m,n}^\downarrow$ such that $\lambda=q\lambda^{(1)}+(1-q)\lambda^{(2)}$, $0<q<1$.
For this purpose, choose the smallest such $i$, and also the smallest $j>i$ such that $\lambda_j < \lambda_i$, \emph{i.e.} $\binom{m}{2} < i < j \leq mn-\binom{m+1}{2}$ and 
\[
  \lambda_{\binom{m}{2}} = \lambda_{\binom{m}{2}+1} = \ldots = \lambda_{i-1} > \lambda_i = \lambda_{i+1} = \ldots =\lambda_{j-1} > \lambda_j.
\]
Now define two (non-normalized) tuples $\widetilde{\lambda}^{(1)}$ and $\widetilde{\lambda}^{(2)}$, as follows:
\begin{align*}
  \widetilde{\lambda}^{(1)}_t 
    \coloneq \begin{cases}
               \lambda_t & \text{if } t < i \text{ or } t \geq j, \\
               \lambda_{\binom{m}{2}} & \text{if } i \leq t < j,
             \end{cases} \quad\phantom{.}\quad
  \widetilde{\lambda}^{(2)}_t 
    \coloneq \begin{cases}
               \lambda_t & \text{if } t < i \text{ or } t \geq j, \\
               \lambda_{j} & \text{if } i \leq t < j.
             \end{cases}
\end{align*}
Clearly, both tuples are nonincreasing, and $\lambda = p\widetilde{\lambda}^{(1)} + (1-p)\widetilde{\lambda}^{(2)}$ for a suitable $0<p<1$. Normalising the vectors, so that $\widetilde{\lambda}^{(u)} = s^{(u)}\lambda^{(u)}$  with distinct $\lambda^{(u)}\in\Delta_{mn-1}^\downarrow$, and $s^{(u)} > 0$ ($u=1,2$), we observe that both  $\lambda^{(u)} \in\mathrm{APPT}_{m.n}$, because the Hildebrand conditions are inherited from $\lambda$ as they are homogeneous. Furthermore, 
\[
  \lambda = p\widetilde{\lambda}^{(1)} + (1-p)\widetilde{\lambda}^{(2)} 
    = p s^{(1)} \lambda^{(1)} + (1-p) s^{(2)}\lambda^{(2)} 
    =: q \lambda^{(1)} + (1-q) \lambda^{(2)}, 
\]
the latter necessarily a convex combination due to normalisation. This contradicts the extremality of $\lambda$ and so our hypothesis must have been false. 

\medskip
\noindent\emph{Case 1:} If $\lambda'$ should not be extremal, this means that $\lambda' = p\lambda^{(1)\prime}+(1-p)\lambda^{(2)\prime}$ with distinct $\lambda^{(1)\prime},\lambda^{(2)\prime}\in\mathrm{APPT}_{m,m}^\downarrow$ and $0<p<1$. Then we can construct two distinct $\lambda_1,\lambda_2 \in \mathrm{APPT}_{m,n}^\downarrow$ such that $\lambda=q\lambda^{(1)}+(1-q)\lambda^{(2)}$, $0<q<1$, as follows: to start, we may assume that property 2 of the claim holds for $\lambda$, so that there exists an $i$ with
\[
  \lambda_t = \begin{cases}
                \lambda_{\binom{m}{2}} & \text{if } \binom{m}{2} < t \leq i, \\
                \lambda_{mn+1-\binom{m+1}{2}} & \text{if } i < t \leq mn-\binom{m+1}{2}.
              \end{cases}
\]
Again, we can construct two (non-normalized) tuples $\widetilde{\lambda}^{(1)}$ and $\widetilde{\lambda}^{(2)}$, by letting
\begin{align*}
  \widetilde{\lambda}^{(1)}_t 
    \coloneq \begin{cases}
               r\lambda^{(1)\prime}_t & \text{if } t \leq \binom{m}{2}  \text{ or } t > mn-\binom{m+1}{2}, \\
               r\lambda^{(1)\prime}_{\binom{m}{2}} & \text{if } \binom{m}{2} < t \leq i, \\
                r\lambda^{(1)\prime}_{mn+1-\binom{m+1}{2}} & \text{if } i < t \leq mn-\binom{m+1}{2},
             \end{cases} \\
  \widetilde{\lambda}^{(2)}_t 
    \coloneq \begin{cases}
               r\lambda^{(2)\prime}_t & \text{if } t \leq \binom{m}{2}  \text{ or } t > mn-\binom{m+1}{2}, \\
               r\lambda^{(2)\prime}_{\binom{m}{2}} & \text{if } \binom{m}{2} < t \leq i, \\
                r\lambda^{(2)\prime}_{mn+1-\binom{m+1}{2}} & \text{if } i < t \leq mn-\binom{m+1}{2}.
             \end{cases}
\end{align*}
These definitions ensure that $\lambda = p\widetilde{\lambda}^{(1)} + (1-p)\widetilde{\lambda}^{(2)}$. Introducing normalisations, $\widetilde{\lambda}^{(u)} = s^{(u)}\lambda^{(u)}$ with $s^{(u)}>0$ ($i=1,2$), we get as before,
\[
  \lambda = p\widetilde{\lambda}^{(1)} + (1-p)\widetilde{\lambda}^{(2)} 
    = p s^{(1)} \lambda^{(1)} + (1-p) s^{(2)}\lambda^{(2)} 
    =: q \lambda^{(1)} + (1-q) \lambda^{(2)}, 
\]
the convex combination enforced by the normalisations. This shows that $\lambda$ is not extremal, in contradiction to our assumption, and so our hypothesis must have been false. 
\end{proof}

\medskip
The reverse direction in the above theorem seems to be true, too, at least it holds for $m=2$ \cite[Theorem.~10]{SongChen2025} and $m=3$ \cite[Theorem.~18]{SongChen2025} by direct arguments, but we don't have a proof in generality yet. 

Note, furthermore, that the relationship between the extreme points of $\mathrm{APPT}_{m,n}$ and those of $\mathrm{APPT}_{m,n}^\downarrow$ is not clear a priori. Certainly, any extreme point of $\mathrm{APPT}_{m,n}$ that happens to lie in $\Delta_{mn-1}^\downarrow$ is extremal in $\mathrm{APPT}_{m,n}^\downarrow$. However, in general an intersection can create new extreme points, such as in fact the uniform distribution $\mathbf {u} = \left(\frac1{mn},\ldots,\frac1{mn}\right)$. Inspection of the known solutions in low dimension however suggests that apart from the latter, all other extreme points of $\mathrm{APPT}_{m,n}^\downarrow$ are in fact extremal in $\mathrm{APPT}_{m,n}$.

\begin{corollary}
The set of extreme points $\ex(\mathrm{APPT}_{m,n})$ coincides with the set of exposed points  $\exp(\mathrm{APPT}_{m,n})$.
\end{corollary}

This follows from Corollary \ref{cor:exp-face} and might provide an avenue to disprove the conjectured equality of absolute PPT and absolute separability. In particular, if one can show that there exists at least one non-exposed extreme point of $\mathrm{ASEP}_{m,n}$, then this necessarily means that $\mathrm{ASEP}_{m,n} \subsetneq \mathrm{APPT}_{m,n}$.


\section{Maximum purity, minimum von Neumann entropy and volume}
\label{sec:purity-etc}
In this section, we shift our attention to the quantitative properties of the set of spectra of absolute PPT $\mathrm{APPT}_{m,n}$, and in some cases $\mathrm{ASEP}_{m,n}$. In particular,  by benchmarking the set of spectra against the separable ball and the inscribed polytope, we provide tight lower and upper bounds for the maximum purity and minimum von Neumann entropy of a quantum state with spectra in $\mathrm{APPT}_{m,n}$, respectively. We start by identifying this inscribed absolute PPT polytope and the maximal separable ball. 

\begin{definition}
Given $m,n\geq 2$ and letting $\Delta_{mn-1}$ denote the probability simplex on $mn$ points, we define the following two sets:
\begin{enumerate}
    \item The \emph{separable ball} $\mathrm{BALL}_{m,n}$: this is the set of spectra satisfying the strict purity bound, given by  
    $$\mathrm{BALL}_{m,n} \coloneq\left\{\lambda \in\Delta_{mn-1}\;\middle|\; \sum_{i=1}^{mn} \lambda_i^2 \le \frac{1}{mn-1} \right\}.$$
    \item The \emph{inscribed polytope} $\mathcal{P}_{m,n}$: this is a polyhedral subset derived via the Gershgorin circle theorem \cite{horn2012matrix, jivulescu2015positive} to our constraint block matrices $L_\pi(\lambda)$, defined globally as  
    \begin{equation}
    \label{eq:polytope_inequalities}
             \mathcal{P}_{m,n} \coloneq \left\{  \lambda \in \Delta_{mn-1} \;\middle|\; \forall \pi \in \widetilde{S} \quad 2\lambda_{\pi(p(i,i))} \ge \sum_{\substack{j=1 \\ j \neq i}}^m \left| \lambda_{\pi(p(i \wedge j, i\vee j))} - \lambda_{\pi(q(i \wedge j, i\vee j))} \right| \right\},
    \end{equation} 
    where $i\wedge j = \min(i,j) $ and $i\vee j = \max(i,j) $, and the functions $p(i,j)$ and $q(i,j)$ for $1 \le i \le j \le m$  are defined as in Eq.~\eqref{eq:index_functions}.       
\end{enumerate}
\end{definition}

Since the purity of a quantum state is invariant under all eigenbasis permutations, the separable $\mathrm{BALL}_{m,n}$ is well-defined without reference to the permutations. Indeed, its been proven that $\mathrm{BALL}_{m,n}\subset \mathrm{ASEP}_{m,n}$ (see \cite{gurvits2002largest}) which implies $\mathrm{BALL}_{m,n}\subset \mathrm{APPT}_{m,n}$. 

By assuming the decreasing order of eigenvalues, the global condition for the inscribed polytope above can be reduced to a simple linear constraint. Thus, the system of inequalities \eqref{eq:polytope_inequalities} generated by all permutations $\pi \in \widetilde{S}$ collapses to this simple constraint, recovering the criteria observed in \cite{jivulescu2015positive, XiongSze2026}, given by the following: 

\begin{theorem}
\label{thm:sufficient_condition}
Given a mixed state $\rho \in \M_m \otimes \M_n$ with a decreasingly ordered eigenvalue spectrum $\lambda_1 \ge \lambda_2 \ge \dots \ge \lambda_{mn}\geq 0$, if the following linear inequality holds:
\begin{equation}
    2\lambda_{mn} + \sum_{k=1}^{m-1} \lambda_{mn-k} \ge \sum_{k=1}^{m-1} \lambda_k \label{eq:sufficient_appt}
\end{equation}
then the spectrum lies within $\mathcal{P}_{m,n}$, and $\rho$ is absolutely PPT.
\end{theorem}

\begin{proof}
Assume $\rho$ has spectra $\lambda$ which is decreasingly ordered as $\lambda_1 \ge \lambda_2 \ge \dots \ge \lambda_{mn}$, and assume 
\begin{equation*}
    2\lambda_{mn} + \sum_{k=1}^{m-1} \lambda_{mn-k} \ge \sum_{k=1}^{m-1} \lambda_k.
\end{equation*}    
To prove $\mathcal{P}_{m,n} \subseteq \mathrm{APPT}_{m,n}$, we must show that for any arbitrary permutation $\pi \in \widetilde{S}$, the active constraint block $L_\pi(\lambda)$ is positive semidefinite. By the Gershgorin circle theorem \cite{horn2012matrix}, $L_\pi(\lambda)$ is positive semidefinite if it is diagonally dominant. This means 
for every row $1 \le i \le m$,
\begin{equation}\label{eq:gershgorin_goal}         2\lambda_{\pi(p(i,i))} \ge \sum_{\substack{j=1 \\ j \neq i}}^m \left| \lambda_{\pi(p(i\wedge j, i\vee j))} - \lambda_{\pi(q(i\wedge j, i\vee j))} \right|.
    \end{equation}
Since $\lambda$ is ordered decreasingly, $\lambda_{mn}$ is the minimum eigenvalue. Thus, the diagonal entry for any arbitrary permutation $\pi$ satisfies $ 2\lambda_{\pi(p(i,i))} \ge 2\lambda_{mn}.$
Similarly, the right-hand side of Eq.~\eqref{eq:gershgorin_goal} requires summing $m-1$ absolute differences of the spectrum. The maximum possible value for this sum occurs when the $m-1$  largest eigenvalues are paired against the $m-1$ smallest eigenvalues. And since the spectrum is ordered, this absolute maximum is exactly $\sum_{k=1}^{m-1} ( \lambda_k - \lambda_{mn-k} )$.  This implies that
\begin{equation}\label{eq:offdiagonal_bound}
        \sum_{k=1}^{m-1} \left( \lambda_k - \lambda_{mn-k} \right) \ge \sum_{\substack{j=1 \\ j \neq i}}^m \left| \lambda_{\pi(p(i\wedge j, i\vee j))} - \lambda_{\pi(q(i\wedge j, i\vee j))} \right|.
    \end{equation}
Thus, by transitivity, it follows that 
\begin{equation*}
        2\lambda_{\pi(p(i,i))} \ge 2\lambda_{mn} \ge \sum_{k=1}^{m-1} \left( \lambda_k - \lambda_{mn-k} \right) \ge \sum_{\substack{j=1 \\ j \neq i}}^m \left| \lambda_{\pi(p(i \wedge j, i\vee j))} - \lambda_{\pi(q(i\wedge j, i \vee j))} \right|.
\end{equation*} 
Thus, the single ordered premise enforces a strict diagonal dominance across every row of every possible constraint block matrix $L_\pi(\lambda)$, so $\lambda \in \mathrm{APPT}_{m,n}$.   
\end{proof}

Thus, the polytope $\mathcal{P}_{m,n}$ is geometrically equivalent to the strictest sufficient linear condition for the general matrix inequalities $L_\pi(\lambda) \geq 0$ ($\pi\in \widetilde{S}$). Although each set $\mathcal{P}_{m,n}$ and $\mathrm{BALL}_{m,n}$ is inscribed in the set of spectra $\textrm{APPT}_{m,n}$, neither set contains the other, but their union $\mathcal{P}_{m,n}\cup\mathrm{BALL}_{m,n}$ is a lower approximation for $\mathrm{APPT}_{m,n}$ \cite[Remark~7.4]{jivulescu2015positive}.

\subsection{Maximum purity}

\begin{theorem}
\label{thm:lower_purity_bound}
Let $\mathcal P_{m,n}\subseteq \mathrm{APPT}_{m,n}$ be the inscribed absolute PPT polytope with $2\leq m\leq n$. The maximum purity of a quantum state $\rho \in \mathcal{M}_m \otimes \mathcal{M}_n$ with spectrum $\lambda \in \mathcal P_{m,n}$ is given by 
\begin{equation}
    \max_{\lambda  \in \mathcal P_{m,n}} \sum_{i=1}^{mn}\lambda_i^2=\max\left\{\frac{mn+8}{(mn+2)^2},\frac{4mt + mn(m-1)^2}{(2t + m^2n-mn)^2}
    \right\},
\end{equation}
where $t = \left\lceil \frac{(m-1)n}{2} \right\rceil$.
\end{theorem}

\begin{proof}
Since both $\mathcal{P}_{m,n}$ and the purity function are permutation invariant with respect to the eigenvalues, it is enough to maximize over the ordered simplex  $\Delta_{mn-1}^{\downarrow} \coloneq \{ \lambda \in \Delta_{mn-1}: \lambda_1 \ge \lambda_2 \ge \dots \ge \lambda_{mn} \}$. Since the vertices of the ordered simplex are the vectors  $\mathbf{u}_i \coloneq (
    \underbrace{\dfrac{1}{i},\ldots,\dfrac{1}{i}}_{i},
    0,\ldots,0 )$ for $i =1, \ldots, mn$, any point $\lambda\in\Delta_{mn-1}^{\downarrow}$ can be written as $\lambda = f(\mathbf{x}) \coloneq\sum_{i=1}^{mn} x_i\mathbf{u}_i$ where $\mathbf{x} = (x_1,\ldots,x_{mn})\in\Delta_{mn-1}$. 
Each component of the vertices are $[\mathbf{u}_i]_k = \frac{1}{i}$ for $i\geq k$ and $[\mathbf{u}_i]_k = 0$ if $i<k,$ thus, the $k$-th component of $\lambda$ is given by  \begin{equation}\label{eq:inverse_map}
    \lambda_k = \sum_{i=1}^{mn} x_i [\mathbf{u}_i]_k =\sum_{i=k}^{mn} \frac{x_i}{i},\qquad \textrm{ for } k = 1,\ldots,mn.\end{equation}
It follows that  $$\lambda_k -\lambda_{k+1} =\sum_{i=k}^{mn} \frac{x_i}{i} -\sum_{i=k+1}^{mn} \frac{x_i}{i} = \frac{x_k}{k}.$$ Thus, by defining $x_k \coloneq k(\lambda_k - \lambda_{k+1})$ and  $ \lambda_{mn+1}\coloneq0$, we obtain the barycentric coordinates of $\lambda$ with respect to the vertices $\mathbf{u}_i$ of the ordered simplex. Hence, $f$ defines a bijection from the standard simplex in $x$-coordinates onto the ordered simplex.
Observe that since $\lambda_k \geq \lambda_{k+1}$ in the ordered simplex, each $x_k \geq 0$. Additionally, 
\begin{equation*}
    \sum_{k=1}^{mn}x_k = \sum_{k=1}^{mn}k(\lambda_k-\lambda_{k+1})= \sum_{k=1}^{mn}\lambda_k = 1,   
\end{equation*} 
and so $\mathbf{x}\in\Delta_{mn-1}.$ Define the functional $h(\lambda)\coloneq 2\lambda_{mn} + \sum_{k=1}^{m-1} \lambda_{mn-k}- \sum_{k=1}^{m-1} \lambda_k\geq 0$ as seen in Eq.~\eqref{eq:sufficient_appt} so that $\mathcal{P}_{m,n}\cap\Delta_{mn-1}^{\downarrow}  = \{\lambda \in\Delta_{mn-1}^{\downarrow} : h(\lambda) \geq 0\}$. Substituting the components of $\lambda$ from Eq.~\eqref{eq:inverse_map}, we have
\begin{align}
 0 \leq h(\lambda) &=2\lambda_{mn} + \sum_{j=m(n-1)+1}^{mn-1} \lambda_{j}- \sum_{k=1}^{m-1} \lambda_k \nonumber\\
 &= 2\left(\sum_{i=mn}^{mn} \frac{x_i}{i}\right) + \sum_{j=m(n-1)+1}^{mn-1} \left(\sum_{i=j}^{mn} \frac{x_i}{i}\right) - \sum_{k=1}^{m-1} \left(\sum_{i=k}^{mn} \frac{x_i}{i}\right)\nonumber \\
&= 2\frac{x_{mn}}{mn} + \sum_{i=m(n-1)+1}^{mn} \frac{x_i}{i} \left( \sum_{j=m(n-1)+1}^{\min(i, mn-1)} 1 \right) - \sum_{i=1}^{mn} \frac{x_i}{i} \left( \sum_{k=1}^{\min(i, m-1)} 1 \right)  \nonumber \\
 &= 2\frac{x_{mn}}{mn} + \sum_{i=m(n-1)+1}^{mn} \frac{x_i}{i} \Big( \min(i, mn-1) - m(n-1) \Big) - \sum_{i=1}^{mn} \frac{x_i}{i} \Big( \min(i, m-1) \Big)\nonumber \\
    &=\left(\frac{2}{mn}\right)x_{mn} + \sum_{i=m(n-1)+1}^{mn-1}\left(\frac{i - mn + 1} {i}\right)x_{i}  + \sum_{i=m}^{m(n-1)}\left(-\frac{m-1}{i}\right)x_{i}  + \sum_{i=1}^{m-1}(-1)x_i \nonumber \\
    &= \sum_{i =1}^{mn} h(\mathbf{u}_i) x_i.
\end{align}
Thus, 
$\mathcal{P}_{m,n}\cap\Delta_{mn-1}^{\downarrow} $  is exactly the set 
\begin{equation}
\label{eq:polytope_new}
\mathcal{P}_{m,n}\cap\Delta_{mn-1}^{\downarrow}  = f\left( \left\{\mathbf{x} \in \Delta_{mn-1} :\sum_{i=1}^{mn} h(\mathbf{u}_i) x_i \ge 0\right\}\right)
\end{equation} where  \begin{equation}
h(\mathbf{u}_i) = 
\begin{cases} 
\frac{2}{mn} & \text{if } i = mn, \\
0 & \text{if } i = mn-1, \\
\frac{i - mn + 1}{i} & \text{if } m(n-1)+1 \le i \le mn-2, \\
-\frac{m-1}{i} & \text{if } m \le i \le m(n-1),\\
-1 & \text{if } 1 \le i \le m-1.
\end{cases}\end{equation}

Notice that the coefficients $h(\mathbf{u}_i)\geq 0$ or $i = mn-1$ and $i = mn, $ implying that the vertices $\mathbf{u}_{mn}, \mathbf{u}_{mn-1}\in\mathcal{P}_{m,n}\cap\Delta_{mn-1}^{\downarrow}$ with $\mathbf{u}_{mn-1} $ lying exactly on the intersecting hyperplane.  On the other hand, $h(\mathbf{u}_i) < 0$, for $1\leq i\leq mn-2$ and therefore the associated ordered simplex vertices $\mathbf{u}_i\notin \mathcal{P}_{m,n}\cap\Delta_{mn-1}^{\downarrow}$. This means that the vertices of $\mathcal{P}_{m,n}\cap\Delta_{mn-1}^{\downarrow}$ are exactly 
$$  \mathbf{u}_{mn-1} = \left(\underbrace{\frac{1}{mn-1},\ldots,\frac{1}{mn-1}}_{mn-1},0 \right),\qquad \mathbf{u}_{mn} = \left(\frac{1}{mn}, \ldots,\frac{1}{mn}\right)$$ 
and new vertices generated via the intersection of the ordered simplex with the halfspace which we define by 
$$\mathbf{w}_{i}\coloneq[\mathbf{u}_{i},\mathbf{u}_{mn}]\cap \left\{\mathbf{\lambda}\in \Delta^\downarrow_{mn-1}:h(\lambda)=0\right\}, \qquad i=1,\ldots,mn-2.$$ 
Since purity is a convex function, its maximum over the polytope is attained at one of these vertices. The purity of the vertices $\mathbf{u}_{mn-1}$ and $\mathbf{u}_{mn}$ is computed to be \begin{equation}\label{eq:purity_mn&mn-1}
    \|\mathbf{u}_{mn-1}\|_2^2=\frac{1}{mn-1},\qquad \|\mathbf{u}_{mn}\|_2^2=\frac{1}{mn}\end{equation}
such that $\|\mathbf{u}_{mn}\|_2^2 <\|\mathbf{u}_{mn-1}\|_2^2 $ for all $m\leq n.$

For $1 \leq i\leq mn-2,$ the new vertices can written as 
$\mathbf{w}_{i} = \alpha_i\mathbf{u}_{i} + (1-\alpha_i)\mathbf{u}_{mn}, \textrm{ where } \alpha_i\in (0,1).$
However, since we require that $h(\mathbf{w}_{i}) = 0,$ we have that 
\begin{equation*}
    \alpha_i h(\mathbf{u}_{i}) + (1-\alpha_i)h(\mathbf{u}_{mn}) = 0,
    \text{ implying }
    \alpha_{i} = \frac{2}{2-(mn\cdot h(\mathbf{u}_{i}))}.
\end{equation*}
Thus, the new vertices take the form $\mathbf{w}_{i} =\left(
    \underbrace{a_i,\ldots,a_i}_{i},
    \underbrace{b_i,\ldots,b_i}_{mn-i}
    \right),$
where $$ a_i =\frac{\alpha_i}{i}+\frac{(1-\alpha_i)}{mn},\qquad b_i =\frac{(1-\alpha_i)}{mn}.$$ More precisely, we find these new vertices to be
\begin{equation}
\label{eq:polytope_vertices_wi}
    \mathbf{w}_{i} =\begin{cases} 
    a_i = \frac{mn-i+1}{(mn-i-1)(mn-2)+2mn-2}, \; b_i = \frac{mn-i-1}{(mn-i-1)(mn-2)+2mn-2}  & \text{ for }  m(n-1)+1 \le i \le mn-2, \\
a_i = \frac{m+1}{2i+mn(m-1)}, \; \qquad \qquad \quad  b_i = \frac{m-1}{2i+mn(m-1)} & \qquad m \le i \le m(n-1), \\
a_i = \frac{i+2}{i(2+mn)},\qquad \qquad \qquad\quad b_i =\frac{1}{(2+mn)} & \qquad 1 \le i \le m-1.
\end{cases}\end{equation}
Computing the purity within these three range of values, it follows that 
\begin{equation}
\|\mathbf{w}_{i}\|_2^2=
\begin{cases} \displaystyle\frac{1}{mn-1} -\frac{mn( mn -i-1)^2}{(mn-1)(( mn -i-1)(mn-2)+2mn-2)^2} & \text{ for }  m(n-1)+1 \le i \le mn-2, \\[1ex]
    \displaystyle\frac{4mi + mn(m-1)^2}{(2i + m^2n -mn)^2}& \qquad m \le i \le m(n-1), \\[1ex]
    \displaystyle\frac{mn+4+\frac{4}{i}}{(2+mn)^2} & \qquad 1 \le i \le m-1.
\end{cases}\end{equation}
Notice that for $1\leq i \leq m-1$, the maximum purity is attained at $i = 1$. Hence, \begin{equation}\label{eq:max_purity1}
    \max_{1\leq i \leq m-1} \|\mathbf{w}_{i}\|_2^2= \frac{mn+8}{(mn+2)^2}.\end{equation}
When $mn =4$, notice that $\|\mathbf{u}_{3}\|_2^2 =\|\mathbf{w}_{1}\|_2^2 = \frac{1}{3}.$ However, for $mn >4$, it follows that
\begin{align*}
    \max_{1\leq i \leq m-1} \|\mathbf{w}_{i}\|_2^2 - \|\mathbf{u}_{mn-1}\|_2^2 =\frac{mn +8}{(mn+2)^2} -\frac{1}{mn-1} =\frac{3(mn-4)}{(mn+2)^2 (mn-1)}\geq 0.
\end{align*}
 
Now for $m \leq i\leq mn-m$, since the purity $\|\mathbf{w}_{i}\|_2^2$ is a rational function, we can fix $m,n$ and consider $i$ as a continuous real variable to find the stationary point. Taking the first derivative of the purity with respect to $i$, we have $$ (\|\mathbf{w}_{i}\|_2^2)' = \frac{4m(n(m-1)-2i)}{(mn(m-1)+2i)^3}. $$ At the stationary point, $(\|\mathbf{w}_{i}\|_2^2)' = 0$ implying that $i = \frac{n(m-1)}{2}$. However, the purity function is increasing for $i < \frac{n(m-1)}{2}$  and decreasing for $i > \frac{n(m-1)}{2}$. Additionally, the stationary point $i \notin \mathbb{Z}$ whenever $n(m-1)$ is odd.  Thus, to account for the discrete structure of $i$, if $n(m-1)$, we can choose $r_0 = \frac{n(m-1)-1}{2}$ such that $r_0 = \left\lfloor \frac{(m-1)n}{2} \right\rfloor$ and $r_0+1 = \left\lceil \frac{(m-1)n}{2} \right\rceil$. Notice that 
$$\|\mathbf{w}_{r_0+1}\|_2^2 -\|\mathbf{w}_{r_0}\|_2^2 = \frac{4m}{(m^2n-n+1)^2 (m^2n-n-1)^2} >0.$$ Therefore, its integer maximum is attained at 
$$ \left\lceil \frac{(m-1)n}{2} \right\rceil.$$
For all $m(n-1)+1 \le i \le mn-2$, it is evident that the  purity $ \|\mathbf{w}_{i}\|_2^2 \leq \frac{1}{mn-1} =  \|\mathbf{u}_{mn-1}\|_2^2$. Therefore, no vertex in this third range maximizes the purity of the set. 

Thus, altogether, neither $\mathbf{u}_{mn}, \mathbf{u}_{mn-1} $ nor $\mathbf{w}_{i} $ for $m(n-1)+1 \le i \le mn-2$ maximizes the purity of the set. Therefore, the maximum purity of the polytope is attained at either $\mathbf{w}_{1}$ or $\mathbf{w}_{\left\lceil \frac{(m-1)n}{2} \right\rceil} $. Thus, 
\begin{align*}
    \max_{\lambda\in \mathcal{P}_{m,n}}\sum_{i=1}^{mn}\lambda_i^2&= \max \{\|\mathbf{w}_{1}\|_2^2, \|\mathbf{w}_{\left\lceil \frac{(m-1)n}{2} \right\rceil}\|_2^2 \} \\
    &= \max\left\{\frac{mn+8}{(mn+2)^2},\frac{4mt + mn(m-1)^2}{(2t + m^2n -mn)^2}
    \right\}\end{align*}
where $ t =\left\lceil \frac{(m-1)n}{2}\right\rceil$ as desired.

\end{proof}

More precisely, we observe the following:
\begin{corollary}\label{cor:precise_polytope_max_purity}
    Let $\mathcal P_{m,n}\subseteq \mathrm{APPT}_{m,n}$ be the inscribed absolute PPT polytope with $2\leq m\leq n,\; n>2$. The maximum achievable purity of a quantum state $\rho \in \mathcal{M}_m \otimes \mathcal{M}_n$ with spectrum $\lambda \in \mathcal P_{m,n}$ is given by 
\begin{equation}
    \max_{\lambda\in\mathcal P_{m,n}}
    \sum_{i=1}^{mn}\lambda_i^2=
    \begin{cases}
    \displaystyle
    \frac{mn+8}{(mn+2)^2},&\textrm{if }  n<n^*, \\[2ex]
    \displaystyle
    \frac{m}{(m^2-1)n},&\quad n\geq n^* \text{ and } \big(m \text{ is odd or } n \text{ is even}\big), \\[2ex]
    \displaystyle
    \frac{mn(m^2-1)+2m}{
    \big((m^2-1)n+1\big)^2},
    &\quad n\geq n^*,\ m \text{ is even and } n \text{ is odd }
    \end{cases}\end{equation} where $$n^* = \frac{2(m^2-2) + 2\sqrt{(m^2-1)(m^2-4)}}{m}.$$
Furthermore, the spectrum of the maximal purity state with $\lambda_1 \geq\cdots\geq \lambda_{mn}\geq 0$ is given by 
\begin{equation}\label{eq:polyope_max_purity_spectra}
    \begin{cases} 
      \lambda_1 =\frac{3}{mn+2} , \quad \lambda_j = \frac{1}{mn+2} &\text{for } j = 2,\ldots, mn,\  \qquad \qquad \qquad \qquad \quad \;\text{ if } n < n^* \\
      \lambda_i =a , \quad \lambda_j =b &\text{for } i = 1,\ldots, t \ \text{ and } \ j = t + 1, \ldots, mn, \ \text{ if } n \geq n^*
    \end{cases}
\end{equation}
where the block size $t = \left\lceil \frac{(m-1)n}{2} \right\rceil$, and the eigenvalues $a$ and $b$ evaluate exactly to
\begin{equation}
    a = \frac{m+1}{2t + mn(m-1)} \quad \text{and} \quad b = \frac{m-1}{2t + mn(m-1)}.
\end{equation}

\end{corollary}

\begin{proof}
    Following the proof of Theorem \ref{thm:lower_purity_bound}, the maximal purity of a state with spectrum $\lambda\in \mathcal{P}_{m,n}$ is given by 
    \begin{equation}
    \max_{\lambda  \in \mathcal P_{m,n}} \sum_{i=1}^{mn}\lambda_i^2=\max\left\{\frac{mn+8}{(mn+2)^2},\frac{4mt + mn(m-1)^2}{(2t + m^2n -mn)^2}
    \right\}\end{equation}  where $ t =
    \left\lceil \frac{(m-1)n}{2}\right\rceil$ and is attained at the vertices 
    \begin{align}\label{eq:polytope_max_purity_spectrum}
     \mathbf{w}_1 &=  \left( \frac{3}{(2+mn)},\frac{1}{(2+mn)},\ldots,  \frac{1}{(2+mn)} \right)\\   
     \mathbf{w}_t &= \left(
    \underbrace{\frac{m+1}{2t + mn(m-1)} ,\ldots,\frac{m+1}{2t + mn(m-1)} }_{t},
    \underbrace{\frac{m-1}{2t + mn(m-1)} ,\ldots,\frac{m-1}{2t + mn(m-1)} )}_{mn-t}
    \right)
    \end{align} respectively. In light of this, the second purity function $\|\mathbf{w_t}\|_2^2$ depends strictly on the parity of $m$ and $n$ by definition of the value of $t$. As such, let us consider the different parity cases of $m$ and $n$. 

\emph{Case 1:} Suppose $mn>4$ where $m$ is odd, or  $n$ is even. Here, $(m-1)n$ is strictly even implying that $t = \frac{(m-1)n}{2} \in \mathbb{Z}^+.$ By substitution, we have
\begin{equation}
    \|\mathbf{w}_{\frac{n(m-1)}{2}}\|_2^2 = \frac{2mn(m-1) + mn(m-1)^2}{\left[ n(m-1) + mn(m-1) \right]^2} = \frac{m}{n(m^2-1)}.
\end{equation}
Therefore, in this case, the max purity is given by           $$\max_{\lambda\in \mathcal{P}_{m,n}}\sum_{i=1}^{mn}\lambda_i^2 = \max\left\{\frac{mn+8}{(mn+2)^2},\frac{m}{(m^2  -1)^2}\right\}.$$ To determine the point of transition between these values, we equate both such that 
\begin{align}\label{eq:quadratic_eq}
    \frac{mn+8}{(mn+2)^2}&=\frac{m}{(m^2  -1)^2}\nonumber\\
    n(m^2-1)(mn+8) &= m(m^2n^2 + 4mn + 4)\nonumber\\
      mn^2 - 4(m^2-2)n + 4m &= 0
\end{align}
The roots of this equation include 
$$ n_\pm = \frac{2(m^2-2)\pm 2\sqrt{(m^2-1)(m^2-4)} }{m}. $$ However, we only require the root $n_+$ since $n_-$ is not relevant for $n>2.$ Since the quadratic equation in Eq.~\eqref{eq:quadratic_eq} has a positive leading coefficient with $m\geq2,$ we observe that for $n^*\coloneq n_+ $
\begin{align*}
    \frac{mn+8}{(mn+2)^2}<\frac{m}{(m^2  -1)^2}\qquad &\textrm{ for } n<n^*\\
    \frac{mn+8}{(mn+2)^2}\geq\frac{m}{(m^2  -1)^2}\qquad &\textrm{ for } n\geq n^*.
\end{align*}

\emph{Case 2:} Suppose $m$ is even and $n$ is odd. Then, $(m-1)n$ is odd, implying that $t = \left\lceil \frac{(m-1)n}{2} \right\rceil = \frac{(m-1)n + 1}{2}\in \mathbb{Z}^+$. It follows that 
\begin{align*}
     \|\mathbf{w}_{\frac{(m-1)n + 1}{2}} \|_2^2  =\frac{2mn(m-1) + 2m + mn(m-1)^2}{\left[ n(m-1) + 1 + mn(m-1) \right]^2} 
= \frac{mn(m^2-1) + 2m}{(n(m^2-1) + 1)^2}.    \end{align*}
Similarly, the max purity in this parity case is given by           $$\max_{\lambda\in \mathcal{P}_{m,n}}\sum_{i=1}^{mn}\lambda_i^2 = \max\left\{\frac{mn+8}{(mn+2)^2},\frac{2mn(m-1) + 2m + mn(m-1)^2}{\left[ n(m-1) + 1 + mn(m-1) \right]^2}\right\}.$$
In fact, one can check that the same transition value $n^*$ allows for 
\begin{align*}
    \frac{mn+8}{(mn+2)^2}\leq \frac{2mn(m-1) + 2m + mn(m-1)^2}{\left[ n(m-1) + 1 + mn(m-1) \right]^2} \qquad &\textrm{ for } n\geq n^*\\
    \frac{mn+8}{(mn+2)^2}> \frac{2mn(m-1) + 2m + mn(m-1)^2}{\left[ n(m-1) + 1 + mn(m-1) \right]^2} \qquad &\textrm{ for } n< n^*.
\end{align*}
Thus, combining the two cases, the max purity of the polytope is given by 
\begin{equation}
    \max_{\lambda\in\mathcal P_{m,n}}
    \sum_{i=1}^{mn}\lambda_i^2=
    \begin{cases}
    \displaystyle
    \frac{mn+8}{(mn+2)^2},&\textrm{if }  n<n^*, \\[2ex]
    \displaystyle
    \frac{m}{(m^2-1)n},&\quad n\geq n^* \text{ and } m \text{ is odd or } n \text{ is even}, \\[2ex]
    \displaystyle
    \frac{mn(m^2-1)+2m}{
    \big((m^2-1)n+1\big)^2},
    &\quad n\geq n^*,\ m \text{ is even and } n \text{ is odd }
    \end{cases}\end{equation}
    which correspond to the spectra as seen in Eq.~\eqref{eq:polytope_max_purity_spectrum} 
with $ t = \left\lceil \frac{(m-1)n}{2}\right\rceil$, as desired. 
\end{proof}

\begin{remark}
    For the polytope $\mathcal{P}_{2,2}$, the maximal purity of $\frac{1}{3}$ is attained by a state with spectrum $\lambda = \frac{1}{3}(1,1,1,0)$ or $\lambda = \frac{1}{6}(3,1,1,1) $, coinciding with the maximal purity of the separable ball $\mathrm{BALL}_{2,2}.$  
\end{remark}
We can leverage this inscribed polytope to derive a tighter lower bound for the maximal purity of $\mathrm{APPT}_{m,n}$ as follows:

\begin{proposition}
For every bipartite system $\M_m \otimes \M_n$ where $mn>4$,
\begin{equation}
        \max_{\lambda\in \mathrm{APPT}_{m,n}}
    \sum_{i=1}^{mn}\lambda_i^2\geq    \max_{\lambda\in\mathcal P_{m,n}}
    \sum_{i=1}^{mn}\lambda_i^2> \max_{\lambda\in\mathrm{BALL}_{m,n}}
    \sum_{i=1}^{mn}\lambda_i^2
\end{equation}
\end{proposition}

\begin{proof} 
Since $\mathcal{P}_{m,n} \subseteq \mathrm{APPT}_{m,n},$ the first inequality follows directly. By definition, 
$\max_{\lambda\in\mathrm{BALL}_{m,n}}
    \sum_{i=1}^{mn}\lambda_i^2 = \frac{1}{mn-1}
$ and by Theorem \ref{thm:lower_purity_bound}, we know that $  \max_{\lambda \in \mathcal{P}_{m,n}} \sum_{i=1}^{mn}\lambda_i^2 \geq \frac{mn+8}{(mn+2)^2}. $ Observe that 
\begin{align*}
    \frac{mn+8}{(mn+2)^2} - \frac{1}{mn-1} = \frac{3(mn - 4)}{(mn+2)^2(mn-1)} >0
\end{align*}
since the denominator $(mn+2)^2(mn-1)> 0$ and the numerator $3(mn-4)> 0$ for all $mn>4$. Therefore 
$$ \max_{\lambda \in \mathcal{P}_{m,n}} \sum_{i=1}^{mn}\lambda_i^2 > \frac{1}{mn-1} = \max_{\lambda\in\mathrm{BALL}_{m,n}}
    \sum_{i=1}^{mn}\lambda_i^2. $$
\end{proof}

\begin{conjecture}\label{conj:max_purity} Let $\mathcal P_{m,n}\subseteq \mathrm{APPT}_{m,n}$ be the inscribed absolute PPT polytope with $2\leq m\leq n,\; n>2$. Then 
\begin{equation*}
      \max_{\lambda\in \mathrm{APPT}_{m,n}}
    \sum_{i=1}^{mn}\lambda_i^2=  \max_{\lambda\in\mathcal P_{m,n}}
    \sum_{i=1}^{mn}\lambda_i^2
\end{equation*}
and occurs at the spectra given by Eq.~\eqref{eq:polyope_max_purity_spectra}.
\end{conjecture}

To support the conjecture, we numerically compute the approximate maximal purity of $\mathrm{APPT}_{m,n}$ for $m=2,3,4$, using a multistart numerical optimization with Python's SLSQP optimizer (see Figure \ref{fig:max-purity-phase_transitions}). By starting from different random points inside the set (while avoiding the vertices of the polytope) we recover the maximal purity for the polytope such that no feasible $\mathrm{APPT}_{m,n}$ spectrum of greater purity was found. Here, we use the minimal number of linear matrix inequalities introduced by  \cite{hildebrand2007positive} (see also \cite{XiongSze2026}) for our computations. 
However, since the number of matrix inequalities defining the absolute PPT set grows exponentially with the dimension, a numerical verification becomes increasingly complex for higher dimensions.

\begin{figure}[ht]
    \centering
    \includegraphics[width=1\linewidth]{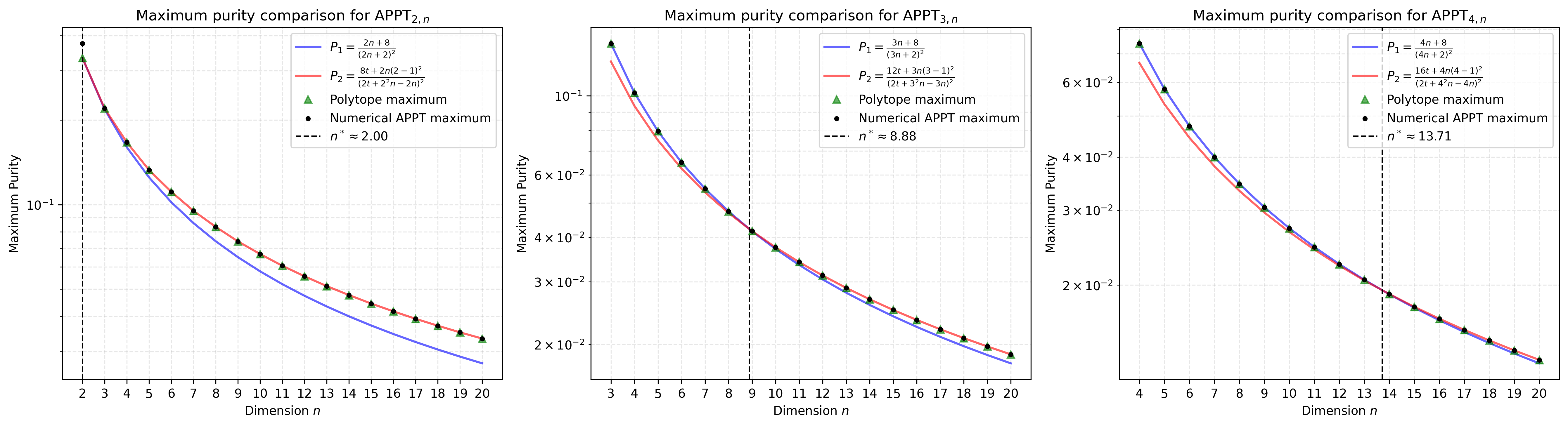}
    \caption{Plots of estimated maximum purity of $\mathrm{APPT}_{m,n}$ compared to exact maximum purity of the inscribed polytope $\mathcal{P}_{m,n}$ for $m =2,3,4.$  The maximum purity of both sets coincide and  as indicated in Corollary \ref{cor:precise_polytope_max_purity}, this value transitions from $P_1 = \frac{mn+8}{(mn+2)^2}$ to $P_2 = \frac{4mt +mn(m-1)^2}{(2t+m^2n-mn)^2}$ after the transition value $n^*$.}
    \label{fig:max-purity-phase_transitions}
\end{figure}

Our numerical results therefore suggest that although  $\mathcal{P}_{m,n}\subseteq \mathrm{APPT}_{m,n}$, the two sets appear to have the same maximal Euclidean radius from the maximally mixed spectrum (at least for $m=2,3,4$, with $mn=4$ as an exception). Indeed, notice that for $\mathbf{u}_{mn} = \left(\frac{1}{mn},\ldots,\frac{1}{mn} \right)$ we have the Euclidean distance
\begin{align*}
   \|\lambda-\mathbf{u}_{mn} \|_2^2 &= \sum_{i=1}^{mn} \left(\lambda_i-\frac{1}{mn}\right)^2\\
   &= \sum_{i=1}^{mn} \left(\lambda_i^2-\frac{2}{mn}\lambda_i + \frac{1}{(mn)^2}\right).
\end{align*}
Since $\sum_{i=1}^{mn}\lambda_i=1$ and $\sum_{i=1}^{mn}\frac{1}{(mn)^2} = \frac{mn}{(mn)^2} = \frac{1}{mn}  $, 
we have
\begin{equation}\label{eq:ball_eqn}
   \sum_{i=1}^{mn} \lambda_i^2 = \frac{1}{mn} + \|\lambda-\mathbf{u}_{mn} \|_2^2. 
\end{equation}

Thus, the evidence for equal maximal purity of $\mathcal{P}_{m,n}$ and $ \mathrm{APPT}_{m,n}$ is equivalent to the two sets having the same largest Euclidean radius around $\mathbf{u}_{mn}.$
Geometrically, this suggests that although $\mathrm{APPT}_{m,n}$ has a curved protruded boundary beyond the polytope boundary, these protrusions do not exceed the Euclidean sphere centered at $\mathbf{u}_{mn}$ with radius reached by the purity maximizers of the polytope.
%
%
The maximum purity for $\mathrm{ASEP}_{2,n}$ has also been discussed in \cite[Cor.~11]{SongChen2025} and \cite{phi2025maximum}.

\begin{corollary}[{Song/Chen~\cite{SongChen2025}}]
\label{cor:Chen_2x2_entropy}
In $\mathrm{ASEP}_{2,2}$, the maximal purity of $\frac{3}{8}$ is attained precisely by the states with eigenvalues $\frac{1}{8+4\sqrt{2}} (3 + 2\sqrt{2}, 3 + 2\sqrt{2}, 1, 1)$.
\end{corollary}

\begin{remark}
The set $\mathrm{ASEP}_{2,2}$ is exceptional among the qubit-qudit system as its maximum purity of $\frac{3}{8}$ does not coincide with the maximum purity of the inscribed polytope $\mathcal{P}_{2,2}$. Recall the exact condition for $\mathrm{ASEP}_{2,n}$ is $\lambda_1-\lambda_{2n-1}\leq2\sqrt{\lambda_{2n-2}\lambda_{2n}} $. The  eigenvalues $\lambda_{2n-2}$ and $\lambda_{2n}$ dictate how large the difference $\lambda_1-\lambda_{2n-1} $ can be. Thus, for $n=2$, even for small $\lambda_{2n}=\lambda_4$, the eigenvalue $\lambda_{2n-2} = \lambda_2$ is the second largest and can therefore keep the product $\lambda_{2n-2}\lambda_{2n}$ sufficiently large to still satisfy the boundary condition. Meaning that, the simultaneous increase of the two largest eigenvalues increases the purity on the boundary. Thus, allowing for the maximum to be attained at eigenvalues $\frac{1}{8+4\sqrt{2}} (3 + 2\sqrt{2}, 3 + 2\sqrt{2}, 1, 1).$ For $n>2$, both $\lambda_{2n-2}$ and $\lambda_{2n}$ are among the smallest eigenvalues and can no longer sustain increasing the purity while maintaining the boundary condition. 

In contrast, the polytope only requires $\lambda_1-\lambda_{2n-1}\leq 2\lambda_{2n}$, using only the smallest eigenvalue $\lambda_{4}$ and ignoring the compensation provided by the larger $\lambda_2$ since $\sqrt{\lambda_{2}\lambda_{4}} \geq\lambda_{4}$.
\end{remark}

We also note that applications of the maximum purity of absolute PPT states have been made for absolutely maximally entangled states and their related quantum marginals \cite{rico2026threshold}.

\subsection{Minimum von Neumann entropy}
In the same way, we find the minimal von Neumann entropy of the polytope $\mathcal{P}_{m,n}$ as follows:

\begin{theorem}\label{thm:polytope_vonNeumann}
    Let $\mathcal P_{m,n}\subseteq \mathrm{APPT}_{m,n}$ be the inscribed absolute PPT polytope with $2\leq m\leq n$.  For $m \leq k\leq mn-m,$
 define $$E(k) = \log_2(mn(m-1)+2k)-\frac{k(m+1)\log_2(m+1)+(mn-k)(m-1)\log_2(m-1)}{(mn(m-1)+2k)},$$ its continuous stationary point by 
$k^* \coloneq \frac{mn(m-1)}{4}\left((m+1)\ln\left(\frac{m+1}{m-1}\right)-2\right) $ and \\
$\zeta_{m,n}  \coloneq\min\{mn-m, \max\{m,k^*\}\}$.

Then for any $t\in \arg\min_{k\in \{\lfloor\zeta_{m,n}\rfloor, \lceil\zeta_{m,n}\rceil \}} E(k)$, the minimal von Neumann entropy of a quantum state $\rho \in \mathcal{M}_m \otimes \mathcal{M}_n$ with spectrum $\lambda \in \mathcal P_{m,n}$ is given by 
    \begin{equation}\label{eq:min_entropy_values}
    \min_{\lambda  \in \mathcal P_{m,n}}S(\lambda)= \min\left\{\log_2(mn-1), \log_2(mn+2)-\frac{3}{mn+2}\log_2 3, E(t)
    \right\}, 
    \end{equation}   
    where  
    $S(\lambda)=-\sum_{i=1}^{mn}\lambda_i\log_2\lambda_i$. 
    Furthermore, the minimum is attained at the spectra
    \begin{equation}
    \label{eq:min_entropy_vertices}
    \lambda = 
    \begin{cases} 
     \lambda_i= \frac{1}{mn-1},\qquad \lambda_{mn} = 0 & \text{for } 1 \leq i \leq mn-1,   \\
      \lambda_1= \frac{3}{2+mn},\qquad \lambda_j = \frac{1}{2+mn} & \text{for } 2 \leq j \leq mn, \\
      \lambda_i= \frac{m+1}{2t+mn(m-1)},\; \lambda_j = \frac{m-1}{2t+mn(m-1)} & \text{for } 1 \leq i \leq t, \; t+1\leq j\leq mn, \\
   \end{cases}
   \end{equation} 
   corresponding to the entropy values in Eq. \eqref{eq:min_entropy_values} respectively.
\end{theorem}

\begin{proof}
Since both $\mathcal{P}_{m,n}$ and the von Neumann entropy are permutation invariant with respect to the eigenvalues, it is sufficient to minimize over the ordered simplex  $\Delta_{mn-1}^{\downarrow} \coloneq \{ \lambda \in \Delta_{mn-1}: \lambda_1 \ge \lambda_2 \ge \dots \ge \lambda_{mn} \}$. Recall from proof of Theorem \ref{thm:lower_purity_bound} that  the vertices of  $\mathcal{P}_{m,n}\cap\Delta_{mn-1}^{\downarrow}$ are 
$$  \mathbf{u}_{mn-1} = \left(\underbrace{\frac{1}{mn-1},\ldots,\frac{1}{mn-1}}_{mn-1},0 \right),\quad \mathbf{u}_{mn} = \left(\frac{1}{mn}, \ldots,\frac{1}{mn}\right) $$ 
and new vertices generated via the intersection defined by 
$$\mathbf{w}_{i}\coloneq[\mathbf{u}_{i},\mathbf{u}_{mn}]\cap \left\{\mathbf{\lambda}\in \Delta^\downarrow_{mn-1}:h(\lambda)=0\right\}, \qquad i=1,\ldots,mn-2$$ 
where $h(\lambda) =2\lambda_{mn} + \sum_{j=m(n-1)+1}^{mn-1} \lambda_{j}- \sum_{k=1}^{m-1} \lambda_k.  $ Since the entropy is a concave function, its minimum over the polytope must occur at a vertex. The entropies of $\mathbf{u}_{mn-1} $ and $\mathbf{u}_{mn} $ are given by 
\begin{equation*}
   S(\mathbf{u}_{mn-1} ) = \log_2(mn-1)\qquad  \textrm{and } \qquad S(\mathbf{u}_{mn} ) = \log_2(mn).
\end{equation*}
Since $mn>mn-1$, $ S(\mathbf{u}_{mn} )>S(\mathbf{u}_{mn-1} )$, the vertex $\mathbf{u}_{mn} $ cannot minimize the entropy. 

Now we consider the entropies of three different vertices whenever $1\leq i \leq mn-2$. As seen from Eq. \ref{eq:polytope_vertices_wi}, these vertices take the form $$\mathbf{w}_i=\left(\underbrace{a_i,\ldots,a_i}_{i},\underbrace{b_i,\ldots,b_i}_{mn-i}\right), $$ with
\begin{equation}
    \begin{cases} 
    a_i = \frac{mn-i+1}{(mn-i-1)(mn-2)+2mn-2}, \; b_i = \frac{mn-i-1}{(mn-i-1)(mn-2)+2mn-2}  & \text{ for }  m(n-1)+1 \le i \le mn-2 \\
a_i = \frac{m+1}{2i+mn(m-1)}, \; \qquad \qquad \quad  b_i = \frac{m-1}{2i+mn(m-1)} & \qquad m \le i \le m(n-1) \\
a_i = \frac{i+2}{i(2+mn)},\qquad \qquad \qquad\quad b_i =\frac{1}{(2+mn)} & \qquad 1 \le i \le m-1.
\end{cases}\end{equation}
For ease of computation, we occasionally use the natural log via the relation $\log_2 i = \frac{\ln i}{\ln 2}$ so that the entropy is $S(\lambda ) = -\frac{1}{\ln 2}\sum_i \lambda_i\ln \lambda_i $. 

For $1 \le i \le m-1,  $, we find the entropy at $\mathbf{w}_{i}$ to be 
\begin{align}
    S(\mathbf{w}_{i} ) &=-ia_i\log_2a_i-(mn-i)b_i\log_2b_i\nonumber\\
    &=-i\left(\frac{i+2}{i(2+mn)}   \right)\log_2\left(\frac{i+2}{i(2+mn)}   \right)-(mn-i)\left(\frac{1}{(2+mn)}   \right)\log_2 \left(\frac{1}{(2+mn)}   \right)\nonumber\\
    &=\log_2(mn+2) -\frac{i+2}{mn+2}
    \log_2\left(\frac{i+2}{i}\right)
\end{align}
To determine the minimum in this range, let $\phi(i)\coloneq (i+2)\ln\left(\frac{i+2}{i}\right)$ and fix $m,n$.  Then its derivative $\phi(i)' =\ln\left(1+\frac{2}{i}\right)-\frac{2}{i}<0 $ since $\ln (1+t)<t$ for any $t>0$. As such, $\phi(i)$ is a strictly decreasing function. Consequently, $S(\mathbf{w}_{i} )$ is a strictly increasing function and its minimum is attained at the smallest value, $i= 1.$ Thus, 
\begin{equation}
    \min_{1 \le i \le m-1 } S(\mathbf{w}_{i} ) = S(\mathbf{w}_{1} ) =\log_2(mn+2) -\frac{3}{mn+2}
    \log_2 3.
\end{equation}

For $m \le i \le m(n-1) $, the entropy is given by 
\begin{align*}
    S(\mathbf{w}_{i} )   &=-ia_i\log_2a_i-(mn-i)b_i\log_2b_i\nonumber\\
    &=-i \left(\frac{m+1}{2i+mn(m-1)}\right)\log_2 \left(\frac{m+1}{2i+mn(m-1)}\right)-(mn-i)\left(\frac{m-1}{2i+mn(m-1)}\right)\\
    &\qquad\log_2 \left(\frac{m-1}{2i+mn(m-1)}\right)\nonumber\\
    &=\log_2(mn(m-1)+2i)-\frac{i(m+1)\log_2(m+1)+(mn-i)(m-1)\log_2(m-1)
    }{(mn(m-1)+2i)}\nonumber\\
    &\coloneq E(i)
\end{align*}
Let $A \coloneq m+1, \; B\coloneq m-1,\; C_i\coloneq iA\ln A + (mn-i)B\ln B $ and $D_i =(mn(m-1)+2i) $. Then the entropy is given $E(i) =\frac{1}{\ln 2} \left(\ln D_i - \frac{C_i}{D_i}\right)$.
By fixing $m,n$ and taking $i$ as a real variable, we find the derivative to be 
\begin{align*}
    E(i)' &=   \frac{1}{\ln 2} \left(\frac{2}{D_i} - \frac{C_i' D_i-2C_i}{D_i^2}\right)\nonumber\\
          &= \frac{(2-C'_i)D_i+2C_i}{D_i^2\ln2}
\end{align*}
Since $C_i' = A\ln A - B\ln B$, by substitution we have the derivative to be, 
\begin{equation*}
    E(i)' = \frac{4i-mn(m-1)\left[(m+1)\ln\left(\frac{m+1}{m-1}\right)-2\right]}{\left(mn(m-1)+2i\right)^2\ln 2}.
\end{equation*}
At the optimal value of $i$, we take $ E(i)' = 0. $ Since the denominator is positive, we find this optimum to be 
$$i^* = \frac{mn(m-1)}{4}\left[(m+1)\ln\left(\frac{m+1}{m-1}\right)-2\right].  $$
Since the numerator of $E(i)' $ is linear in $i$, it follows that $E(i)'<0 $ for $i<i^*$ and $E(i)'>0 $ for $i>i^*$ and therefore implying that, $E(i)$ is strictly decreasing before $i^*$ and increasing afterwards. However, since the minimizer must lie in the range $[m,mn-m]$, we can define 
$$\zeta_{m,n}  \coloneq\begin{cases}
    m &\,\text{for } i^*<m,\\
    i^* &\; m \leq i^* \leq mn-m,\\
    mn-m &\;  i^* > mn-m.
\end{cases} $$
Additionally, we also require that the minimizer must be an integer. Therefore, if $\zeta_{m,n} \in \mathbb{Z}$, it follows directly that the minimizer $t=\zeta_{m,n}.$ On the other hand, if $\zeta_{m,n} \notin \mathbb{Z}$, it cannot be a feasible index for the vertex. However, since $E(i)$ is strictly decreasing before $\zeta_{m,n}$ and strictly increasing after, only the two adjacent integers $\lfloor\zeta_{m,n}\rfloor$ and $\lceil\zeta_{m,n}\rceil$ can minimize $E(i)$. Therefore, the minimizer here is $t = \arg\min_{i\in \{\lfloor\zeta_{m,n}\rfloor, \lceil\zeta_{m,n}\rceil \}} E(i)$. Thus, for $m\leq i \leq mn-m $, 
\begin{equation}
    \min_{m\leq i \leq mn-m}  S(\mathbf{w}_{i} ) = \log_2(mn(m-1)+2t)-\frac{t(m+1)\log_2(m+1)+(mn-t)(m-1)\log_2(m-1)
    }{(mn(m-1)+2t)}
\end{equation}
where $t \in \arg\min_{i\in \{\lfloor\zeta_{m,n}\rfloor, \lceil\zeta_{m,n}\rceil \}} E(i) $.

Now let $g(x) =-\log_2 x$. Since $g(\lambda_i)$ is convex and $\sum_i\lambda_i= 1$, by Jensen's inequality
\begin{align*}
  S(\lambda) =\sum_{i=1}^{mn}\lambda_i g(\lambda_i) \geq g\left(\sum_{i=1}^{mn}\lambda_i \lambda_i \right)  = -\log_2 \left(\sum_{i=1}^{mn}\lambda_i^2  \right)
\end{align*}
Then for $m(n-1)+1 \le i \le mn-2  $, we have seen from the proof of Theorem \ref{thm:lower_purity_bound} that 
 $ \|\mathbf{w}_i\|_2^2< \frac{1}{mn-1}$. Since $g(x)$ is decreasing,
 $$-\log_2 \left(\|\mathbf{w}_i\|_2^2  \right) > -\log_2\left(\frac{1}{mn-1}\right).$$
Thus it follows that, $S(\mathbf{w}_i) \geq -\log_2 \left(\|\mathbf{w}_i\|_2^2  \right) > \log_2\left(mn-1\right) = S(\mathbf{u}_{mn-1}).$
Therefore, the vertex $\mathbf{w}_i$ for $m(n-1)+1 \le i \le mn-2  $ cannot minimize the von Neumann entropy. Combining all these range of values of $i$, the minimum von Neumann entropy may be attained at any of the vertices 
$\mathbf{u}_{mn-1},\; \mathbf{w}_1 \text{ and } \mathbf{w}_t$ with 
$S(\mathbf{u}_{mn-1} ) = \log_2(mn-1) $, $ S(\mathbf{w}_1 )= \log_2(mn+2)-\frac{3}{mn+2}\log_2 3  $ and $S(\mathbf{w}_t )=  E(t)$ respectively. Therefore, 
\begin{equation}
    \min_{\lambda  \in \mathcal P_{m,n}}S(\lambda)= \min\left\{\log_2(mn-1), \log_2(mn+2)-\frac{3}{mn+2}\log_2 3, E(t)
    \right\}.
\end{equation}
\end{proof}

\begin{example}
    Consider the polytope $\mathcal{P}_{m,n}$ such that $m=2$ and $n =2.$ By Theorem \ref{thm:polytope_vonNeumann}, we can find that the range of values for $k$ is  $[mn-m,m]=[2,2]$. Thus,  $k^*=3\ln3-2\approx 1.295837$ and $\zeta_{2,2}=2$. Since  $k\in\mathbb{Z}$, $t =2$ and $E(2) = 3-\frac{2\log_2 3}{4}$. Therefore, the minimal von Neumann entropy over the polytope $\mathcal{P}_{2,2}$ is given by
    \begin{align*}
        \min_{\lambda  \in \mathcal P_{2,2}}S(\lambda) &= \min\left\{\log_2 3, \log_2 6-\frac{1}{2}\log_2 3, 3-\frac{1}{2}\log_2 3
    \right\} \\
    &\approx \min\{1.585,1.792, 1.811 \}\\
    &=1.585.
    \end{align*}
    Thus, $\min_{\lambda  \in \mathcal P_{2,2}}S(\lambda) = \log_2 3$ and occurs at the spectrum $\lambda = \frac{1}{3}\left(1,1,1,0\right)$. 
\end{example}

\noindent This minimal von Neumann entropy and its associated spectra of $\mathcal{P}_{2,2}$ coincides with those of $\mathrm{ASEP}_{2,2}$ (see \cite[Cor.~11]{SongChen2025}). In general, it follows that:

\begin{proposition}
For every bipartite system $\M_m \otimes \M_n$ where $2\leq m\leq n$,
\begin{equation}
        \min_{\lambda\in \mathrm{APPT}_{m,n}}
    S(\lambda)\leq    \min_{\lambda\in\mathcal P_{m,n}}
    S(\lambda) \leq\min_{\lambda\in\mathrm{BALL}_{m,n}}
    S(\lambda)
\end{equation}
\end{proposition}

 \begin{figure}[ht]
    \centering
    \includegraphics[width=0.8\linewidth]{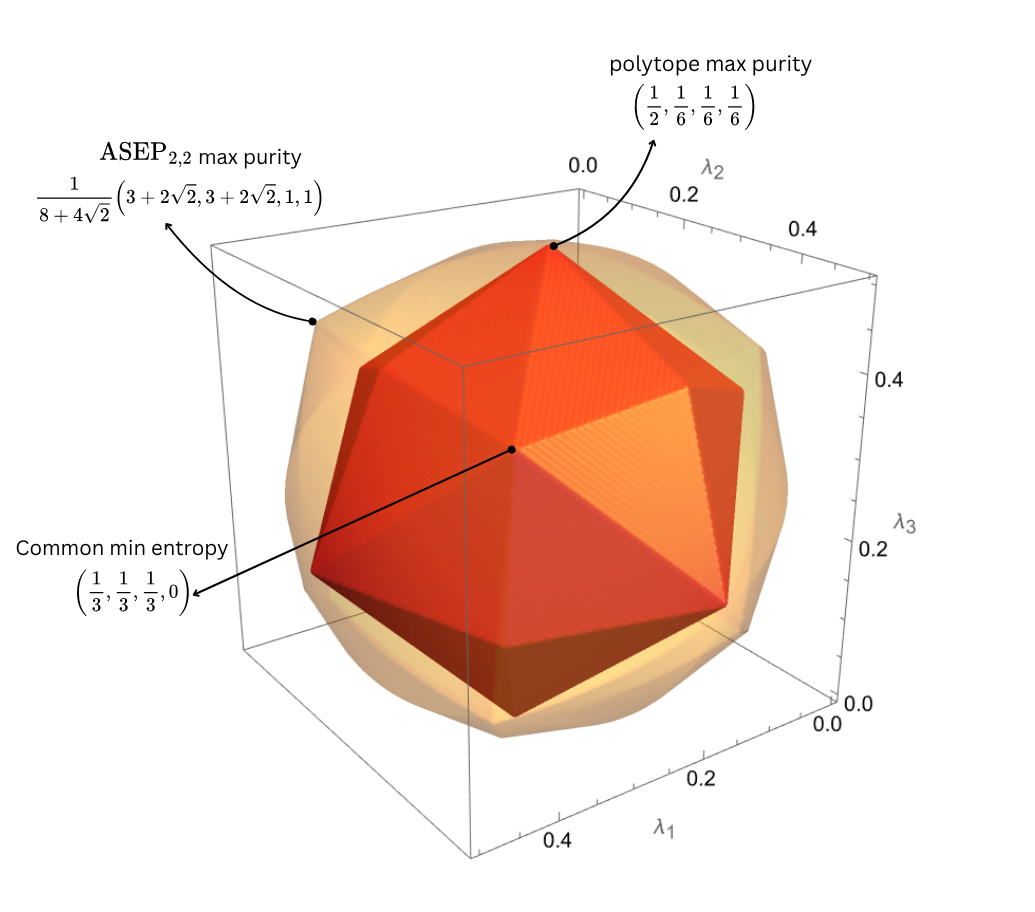}
    \caption{The polytope $\mathcal{P}_{2,2}$ inscribed in $\text{ASEP}_{2,2}$ with their respective purity maximizers and entropy minimizers.}
    \label{fig:placeholder}
\end{figure}

To properly distinguish the entropy of the polytope with $\mathrm{APPT}_{m,n}$, we consider $m=2,3,4$ and numerically approximate the minimal von Neumann entropy of  $\mathrm{APPT}_{m,n}$  for $n\leq 20$ as seen in Figure \ref{fig:min-entropy-phase_transitions}. In general, there exists several points of transition of the minimal entropy between the continuous branch functions, $S_1 = \log_2 (mn-1), \; S_2= \log_2(mn+2) - \frac{3}{mn+2}\log_2 3$ and the discrete value function $E(t)$ defined in Eq.~\eqref{eq:min_entropy_values}. These natural transition values are obtained by equating the continuous functions and finding the smallest subsystem of dimension $n$ at which the minimum occurs. However, in Figure \ref{fig:min-entropy-phase_transitions}, we only indicate the final discrete transition value $n^*_S$ at which the eventual minimizing branch becomes dominant as $n$ increases. 

\begin{figure}[ht]
    \centering
    \includegraphics[width=1\linewidth]{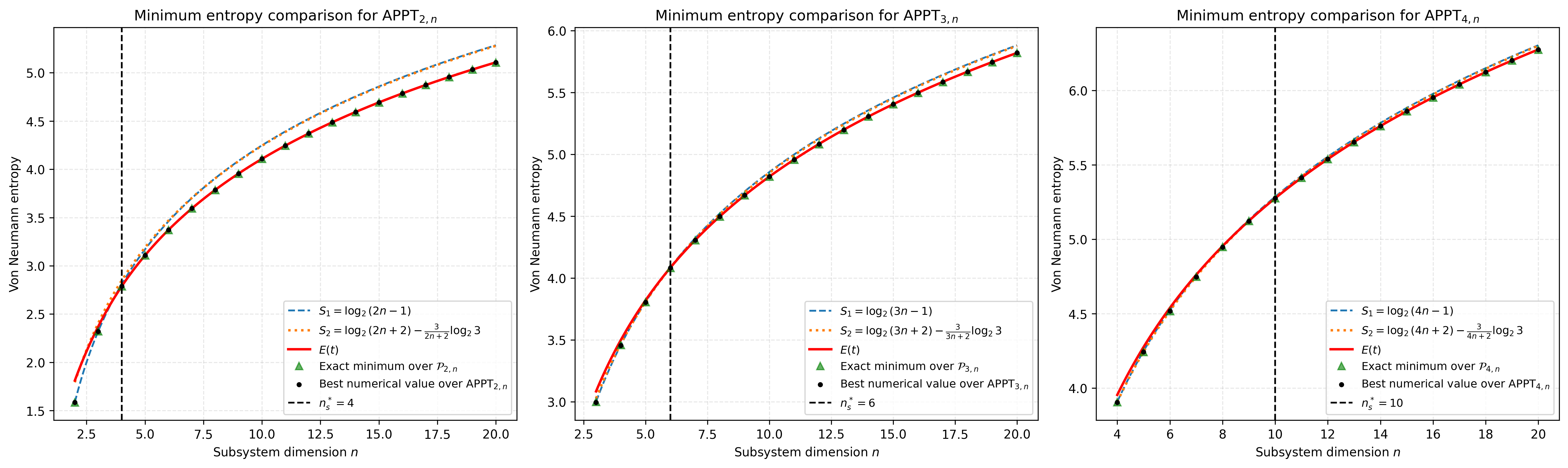}
    \caption{The minimal von Neumann entropy of $\mathrm{APPT}_{m,n}$ and $\mathcal{P}_{m,n}$ where $m=2,3,4$. For each $m$, the minimum of $\mathrm{APPT}_{m,n}$ eventually coincides with the polytope after the discrete transition value $n^*_S$.} 
    \label{fig:min-entropy-phase_transitions}
\end{figure}

\noindent In spite of this, we find that the best numerical  minimal entropy of $\mathrm{APPT}_{m,n}$ may be strictly less than that of the polytope for some dimensions. To see this further, define the entropy gap $\Delta S_{\min} \coloneq \min_{\lambda\in \mathcal{P}_{m,n}} S(\lambda)-   \min_{\lambda\in\mathrm{APPT}_{m,n}}S(\lambda) $ so that up to numerical error, $\Delta S_{\min} \geq 0 $ since $\mathcal{P}_{m,n}\subseteq \mathrm{APPT}_{m,n}$. If $\Delta S_{\min} =0$, then the entropies of the two sets coincide and if $\Delta S_{\min} >0$, the entropy of $\mathrm{APPT}_{m,n}$ is strictly less than $\mathcal{P}_{m,n}$. This is demonstrated in Figure \ref{fig:entropy-gap}. For example, notice that for $m=2$, the numerical gap vanishes for the tested $n\leq 20$ except $n=3$. For $m=4$, $\Delta S_{\min} >0$ for certan low-dimensional $n$ and collapses to zero after the transition value $n^*_S =10.$ This suggests that although the minimal von Neumann entropy of $\mathrm{APPT}_{m,n}$ may be lower than the polytope in some lower dimensional subsystems of $n$, it eventually coincides with that of the polytope after the final discrete transition value $n^*_S$. 

\begin{figure}[ht]
    \centering
    \includegraphics[width=1\linewidth]{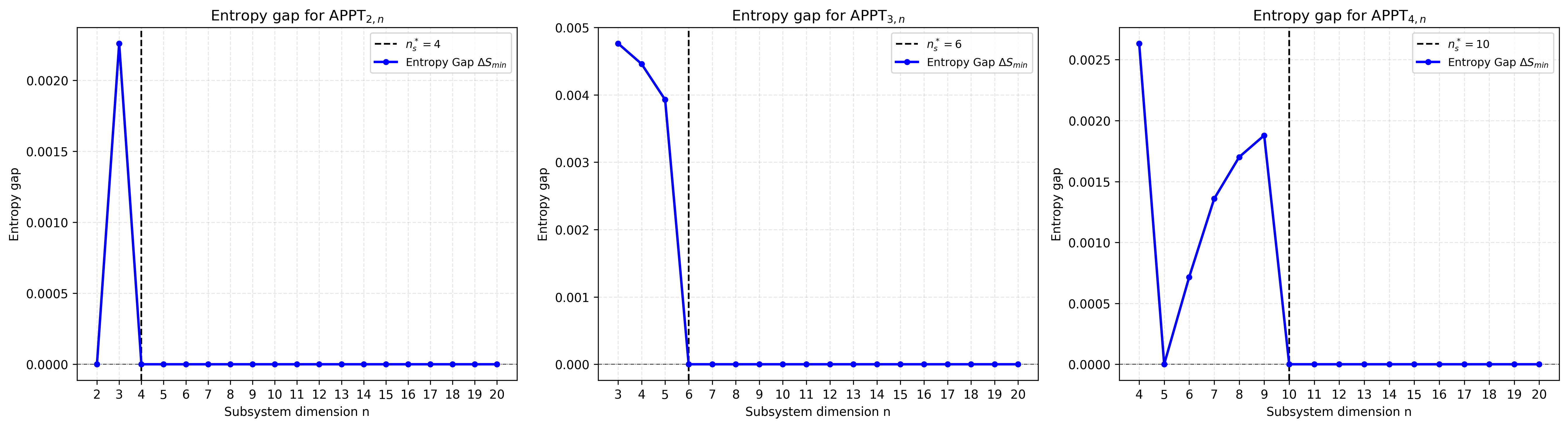}
    \caption{The minimal von Neumann entropy gap of $\mathrm{APPT}_{m,n}$ and $\mathcal{P}_{m,n}$.} 
    \label{fig:entropy-gap}
\end{figure}

\noindent Geometrically, the entropy gap measures the extent to which the curved boundary of $\mathrm{APPT}_{m,n}$ provides entropy-minimizing spectra beyond the polytope. For low dimensions where $n<n^*_S,$ the numerical minimizers sometimes lie on this curved boundary outside the polytope producing a positive gap. And for $n \geq n^*_S,$ the observed minimizers return to the vertices of the polytope, and consequently the gap vanishes in the tested range. 



\subsection{Volume analysis}
In this section, we investigate the spectral relative volume of the absolute PPT set in relation to the separable ball $\mathrm{APPT}_{m,n}$ and the inscribed polytope $\mathcal{P}_{m,n}$.

\begin{theorem}[{Lasserre~\cite{Lasserre2015}}] 
\label{thm:Lasserre}
Let $\Delta_{N-1} \subset \mathbb{R}^N$ be the probability simplex. Let $\mathbf{a} = (a_1, \dots, a_N) \in \mathbf{S}^{N-1}$, the unit sphere, $a_0 \coloneq 0$, and assume that $a_i \neq a_j$ for any pair $(i, j)$ with $i \neq j$. Then the $(N-1)$-dimensional volume of the simplex section $\Theta(\mathbf{a}, t) \coloneq \{\mathbf{x} \in \Delta_{N-1} : \mathbf{a}^T \mathbf{x} \le t\}$ is given by
\begin{equation}
\mathrm{vol}(\Theta(\mathbf{a}, t)) = \frac{\sqrt{N}}{(N-1)!} \sum_{i=1}^N \frac{(t - a_i)^{N-1}_+}{\prod_{j \neq i} (a_j - a_i)}
\end{equation}
where $(z)_+ \coloneq \max\{0, z\}$. 

\end{theorem}

 As a consequence of the above theorem, we can define the $(N-1)$-dimensional relative volume of the simplex section to be $$\mathrm{vol}_{\mathrm{rel}}(\Theta(\mathbf{a}, t))\coloneq\frac{\mathrm{vol}(\Theta(\mathbf{a}, t))  }{\mathrm{vol}(\Delta_{N-1})   }.
$$ Thus, since $\mathrm{vol}(\Delta_{N-1}) =\frac{\sqrt{N}}{(N-1)!}  \ $, it follows that  whenever $a_i \neq a_j$,
With this, we find the relative volume of the polytope to be given as follows:
$$\mathrm{vol}_{\mathrm{rel}}(\Theta(\mathbf{a}, t))= \sum_{i=1}^N \frac{(t - a_i)^{N-1}_+}{\prod_{j \neq i} (a_j - a_i)}$$ for any such simplex section.

\begin{proposition}[Volume of the inscribed polytope]
\label{prop:poly_volume}
Let $\mathcal{P}_{m,n} \subset \Delta_{mn-1}$ be the inscribed absolute PPT polytope with $2\leq m\leq n$.  Then the relative volume of $\mathcal{P}_{m,n}$  is given by

\begin{equation}
\mathrm{vol}_{\mathrm{rel}}(\mathcal{P}_{m,n}) =\left(\frac{2}{mn+2}\right)^{m-1} \left(\prod_{i=m}^{mn-m}\frac{2i}{mn(m-1)+2i}\right) \left( \prod_{j=1}^{m-2}\frac{2(mn-j-1)}{mnj + 2(mn-j-1)}\right).
\end{equation}
Equivalently, this expression simplifies to 
\begin{equation}
\mathrm{vol}_{\mathrm{rel}}(\mathcal{P}_{m,n}) =\frac{4^{m-1} \Gamma(mn)}{ (m^2n^2 - 4)^{m-1} \Gamma(m)} \cdot \frac{\Gamma\left(\frac{mn(m-1) + 2m}{2}\right)}{\Gamma\left(\frac{mn(m+1) - 2m + 2}{2}\right)} \cdot \frac{\Gamma\left(\frac{2mn-2}{mn-2}\right)}{\Gamma\left(m + \frac{mn}{mn-2}\right)}
\end{equation}
where $\Gamma(\cdot)$ denotes the Gamma function satisfying $\Gamma(z)=(z-1)!$ for $z\in \mathbb{Z}_{\geq 0}$.
\end{proposition}

\begin{proof}
Following the proof of Theorem \ref{thm:lower_purity_bound}, we identified that for the ordered simplex  $\Delta_{mn-1}^\downarrow = \{\lambda\in \Delta_{mn-1}: \lambda_1\geq \cdots \lambda_{mn}\geq 0\}$ and $f(\mathbf{x}) = \sum_{i=1}^{mn}x_i\mathbf{u}_i$ where $\mathbf{x}\in \Delta_{mn-1}$ and
$\mathbf{u}_i \coloneq\dfrac{1}{i} (
    \underbrace{1,\ldots,1}_{i},
    0,\ldots,0 )$, the ordered section of the polytope given by 
    $\mathcal{P}_{m,n}\cap\Delta_{mn-1}^{\downarrow} $  is exactly the set 
$\mathcal{P}_{m,n}\cap\Delta_{mn-1}^{\downarrow}  = f\left( Y\right)$
 where $Y = \left\{\mathbf{x} \in \Delta_{mn-1} :\sum_{i=1}^{mn} h(\mathbf{u}_i) x_i \ge 0\right\}   $ and \begin{equation}
h(\mathbf{u}_i) = 
\begin{cases} 
\frac{2}{mn} & \text{if } i = mn \\
0 & \text{if } i = mn-1\\
\frac{i - mn + 1}{i} & \text{if } m(n-1)+1 \le i \le mn-2 \\
-\frac{m-1}{i} & \text{if } m \le i \le m(n-1) \\
-1 & \text{if } 1 \le i \le m-1.
\end{cases}\end{equation}
Following Theorem \ref{thm:Lasserre}, even though the coefficients $a_i = -h(\mathbf{u}_i)$ are not pairwise distinct in some instances as required, Lassere's identical weights condition applies (see \cite[Sec.~2.1]{Lasserre2015}). Here, notice that all identical coefficients occur at instances where $h(\mathbf{u}_i)< 0$. Since $t=0$, the numerator terms $(t-a_i)_+= (h(\mathbf{u}_i))_+ = 0$ for all $h(\mathbf{u}_i)\leq 0.$ Thus, for each identical-weight corrected term associated with an identical coefficient equals zero.  The unique strictly positive coefficient that contributes to the sum is exactly $h(\mathbf{u}_{mn}) = \frac{2}{mn}.$ Thus, the relative volume of the simplex section $Y$ is given by 

\begin{align*}
\mathrm{vol}_{\mathrm{rel}}(Y) &= \frac{(\frac{2}{mn})^{mn-1}}{\prod_{j=1}^{mn-1} (\frac{2}{mn} - h(\mathbf{u}_j))} 
\end{align*}
 However, since $h(\mathbf{u}_{mn-1}) =0,$ the relative volume reduces to 
 $$ \mathrm{vol}_{\mathrm{rel}}(Y) =\frac{(\frac{2}{mn})^{mn-2}}{\prod_{j=1}^{mn-2} (\frac{2}{mn} - h(\mathbf{u}_j))}. $$
Since the map $f$ is affine bijective from the probability simplex $\Delta_{mn-1}$ to the ordered simplex $\Delta_{mn-1}^\downarrow$, and $f(Y) =\mathcal{P}_{m,n}\cap\Delta_{mn-1}^{\downarrow}   $, the relative volume ratio is preserved such that 
$$\mathrm{vol}_{\mathrm{rel}}(Y) = \frac{\mathrm{vol}(Y)  }{\mathrm{vol}(\Delta_{mn-1} )}= \frac{\mathrm{vol}(\mathcal{P}_{m,n}\cap\Delta_{mn-1}^{\downarrow}   )}{\mathrm{vol}(\Delta_{mn-1}^\downarrow )}.$$
Additionally, since the polytope is permutation invariant and the full probability simplex can be partitioned into $(mn)!$ congruent ordered sections, up to measure zero boundaries, it follows that 
$$ \mathrm{vol}(\Delta_{mn-1}) = (mn)!\cdot \mathrm{vol}(\Delta_{mn-1}^\downarrow) 
\quad \textrm{ and } \quad
\mathrm{vol}(\mathcal{P}_{m,n}) = (mn)!\cdot    \mathrm{vol}(\mathcal{P}_{m,n}\cap\Delta_{mn-1}^{\downarrow}).$$ 
Thus, the relative volume of each ordered section remains equal to the relative volume of the polytope as 
$$\frac{\mathrm{vol}(\mathcal{P}_{m,n}\cap\Delta_{mn-1}^{\downarrow}   )}{\mathrm{vol}(\Delta_{mn-1}^\downarrow) }   = \frac{\mathrm{vol}(\mathcal{P}_{m,n})}{\mathrm{vol}(\Delta_{mn-1})} = \mathrm{vol}_{\mathrm{rel}}(\mathcal{P}_{m,n}). $$
The relative volume of the polytope is therefore given by 
\begin{align*}
    \mathrm{vol}_{\mathrm{rel}}(\mathcal{P}_{m,n}) &=\frac{(\frac{2}{mn})^{mn-2}}{\prod_{j=1}^{mn-2} (\frac{2}{mn} - h(\mathbf{u}_j))}\\
    &= \prod_{j=1}^{mn-2}\frac{\frac{2}{mn}}{ (\frac{2}{mn} - h(\mathbf{u}_j)) }\\
    &=\prod_{j=1}^{m-1}\left(\frac{\frac{2}{mn}}{ (\frac{2}{mn} +1) }\right) \cdot \prod_{j=m}^{mn-m}\left(\frac{\frac{2}{mn}}{ (\frac{2}{mn} + \frac{(m-1)}{j}) }\right)\cdot \prod_{j=mn-m+1}^{mn-2}\left(\frac{\frac{2}{mn}}{ (\frac{2}{mn} + \frac{(mn-j-1)}{j}) }\right)\\
    &= \left(\frac{2}{ mn+2 }\right)^{m-1}\cdot \prod_{j=m}^{mn-m}\left(\frac{2j}{ mn(m-1)+2j }\right)\cdot \prod_{k=1}^{m-2}\left(\frac{2(mn-k-1)}{ (mnk +2(mn-k-1)) }\right)
\end{align*}
where we take $k = mn-j-1$ so that since $mn-m+1\leq j\leq mn-2$, it follows that $m-2\geq k \geq 1 $.

To simplify this expression even further, we write the products in Gamma-function form.  Using the Gamma identity $\Gamma(z+1) = z\Gamma(z)$ where $\Gamma(z) = (z-1)!$ and $z\in \mathbb{Z}^+$, we can iteratively show the identities 
\begin{equation}\label{eq:gamma_idn}
    \prod_{j=A}^{B} (j+S) = \frac{\Gamma(B+1+S)}{\Gamma(A+S)} \quad \textrm{and} \quad \prod_{j=A}^{B}j = \frac{(B)!}{(A-1)!}= \frac{\Gamma(B+1)}{\Gamma(A)}.\end{equation}
With this, let $\eta = \frac{mn(m-1)}{2}$ and $ \gamma = \frac{2mn-2}{mn-2}$. Then the relative volume of the polytope simplifies to 
\begin{align*}
    \mathrm{vol}_{\mathrm{rel}}(\mathcal{P}_{m,n})&= \left(\frac{2}{ mn+2 }\right)^{m-1}\cdot \left(  \frac{\prod_{j=m}^{mn-m}  (j)}{ \prod_{j=m}^{mn-m}  (\eta +j) }\right)\cdot \left( \frac{\prod_{k=1}^{m-2} [  2(mn-k-1)]}{ \prod_{k=1}^{m-2}  ( mn-2)(k+\gamma) }\right)\\
    &= \left(\frac{2}{ mn+2 }\right)^{m-1}\cdot \frac{\Gamma(mn-m+1)\Gamma(m+\eta)}{\Gamma(m)\Gamma(mn-m+1+\eta)}\cdot \frac{(2^{m-2})\prod_{j=mn-m+1}^{mn-2}(j)  }{(mn-2)^{m-2} \prod_{k=1}^{m-2}(k+\gamma) }\\
    &= \left(\frac{2}{ mn+2 }\right)^{m-1}\cdot \frac{\Gamma(mn-m+1)\Gamma(m+\eta)}{\Gamma(m)\Gamma(mn-m+1+\eta)} \cdot \left( \frac{2}{mn-2}\right)^{m-2}\cdot \frac{\Gamma(mn-1)\Gamma(1+\gamma)}{\Gamma(mn-m+1)\Gamma(\gamma +m-1)}.
\end{align*}
By substituting $\Gamma(1+\gamma) =  \gamma\Gamma(\gamma)$, $\eta = \frac{mn(m-1)}{2}$ and $ \gamma = \frac{2mn-2}{mn-2}$, the relative volume equation simplifies to 
\begin{align*}
      \mathrm{vol}_{\mathrm{rel}}(\mathcal{P}_{m,n})&= \frac{4^{m-1}(mn-1)\Gamma(mn-1)}{ (mn+2)^{m-1}  (mn-2)^{m-1} \Gamma(m) }\cdot \frac{\Gamma\left(\frac{mn(m-1) + 2m}{2}\right)}{\Gamma\left(\frac{mn(m+1) - 2m + 2}{2}\right)} \cdot \frac{\Gamma\left(\frac{2mn-2}{mn-2}\right)}{\Gamma\left(m + \frac{mn}{mn-2}\right)} \\
      &= \frac{4^{m-1} \Gamma(mn)}{ ((mn)^2 - 4)^{m-1} \Gamma(m)} \cdot \frac{\Gamma\left(\frac{mn(m-1) + 2m}{2}\right)}{\Gamma\left(\frac{mn(m+1) - 2m + 2}{2}\right)} \cdot \frac{\Gamma\left(\frac{2mn-2}{mn-2}\right)}{\Gamma\left(m + \frac{mn}{mn-2}\right)},
\end{align*}
as desired. 

    \end{proof}

\begin{proposition}[{Volume of separable ball, \.Zyczkowski~\emph{et~al.}~\cite{zyczkowski1998volume}}]
\label{prop:ball_volume}
Let $\mathrm{BALL}_{m,n} \subset \Delta_{mn-1}$ be the separable ball with $2\leq m\leq n$.  Then the relative volume of $\mathrm{BALL}_{m,n}$  is 
\begin{equation}
\mathrm{vol}_{\mathrm{rel}}(\mathrm{BALL}_{m,n}) = \frac{\pi^{(mn-1)/2}  (mn-1)!}{\Gamma(\frac{mn-1}{2} + 1) (mn)^{mn/2}  (mn-1)^{(mn-1)/2}}.
\end{equation}

\end{proposition}

\begin{proof}Let $\mathbf{u}= (\frac{1}{mn}, \ldots, \frac{1}{mn})$ be the spectrum of the maximally mixed state. As seen from
Eq.~\eqref{eq:ball_eqn}, for every $\lambda\in \Delta_{mn-1}$, the purity of a state can be written in terms of the Euclidean distance as 
$$ \sum_{i=1}^{mn} \lambda_i^2 = \frac{1}{mn} + \|\lambda-\mathbf{u}_{mn} \|_2^2.$$ This implies that by definition, for every $\lambda\in \mathrm{BALL}_{m,n}$, 
\begin{align*}
     \|\lambda-\mathbf{u}_{mn} \|_2^2  \leq \frac{1}{mn(mn-1)}.
\end{align*}
Thus, $\mathrm{BALL}_{m,n}$ is the $(mn-1)$-Euclidean ball centered at $\mathbf{u}$ with radius $ \frac{1}{\sqrt{mn(mn-1)}}$. Therefore, the Euclidean volume of this separable ball is \begin{equation}
       \mathrm{vol}(\mathrm{BALL}_{m,n}) = \frac{\pi^{(mn-1)/2}   }{\Gamma(\frac{mn-1}{2}+1)} \left( \frac{1}{\sqrt{mn(mn-1)}}\right)^{mn-1}. 
\end{equation}
Dividing by the Euclidean volume of the  probability simplex $\text{vol}(\Delta_{mn-1})= \frac{\sqrt{mn}}{(mn-1)!}$, the relative volume of the separable ball is therefore given by
\begin{equation}
       \mathrm{vol}_{\mathrm{rel}}(\mathrm{BALL}_{m,n})   = \frac{\mathrm{vol}(\mathrm{BALL}_{m,n})}{\mathrm{vol}(\Delta_{mn-1})} = \frac{\pi^{(mn-1)/2}  (mn-1)! }{\Gamma(\frac{mn-1}{2} + 1)  (mn)^{mn/2}  (mn-1)^{(mn-1)/2}},
\end{equation}  
which concludes the proof.
\end{proof}

\begin{theorem}
Let $\mathrm{vol}_{\mathrm{rel}}(\mathrm{APPT}_{m,n})$ denote the relative volume of $\mathrm{APPT}_{m,n}$. Then 
\begin{equation}
\mathrm{vol}_{\mathrm{rel}}(\mathrm{APPT}_{m,n}) \ge \max \{ \mathrm{vol}_{\mathrm{rel}}(\mathrm{BALL}_{m,n}), \; \mathrm{vol}_{\mathrm{rel}}(\mathcal{P}_{m,n})  \}.
\end{equation}
\end{theorem}

This follows directly as $\mathrm{BALL}_{m,n}$ and $\mathcal{P}_{m,n}$ are subsets of $\mathrm{APPT}_{m,n}$.  We implement a spectral Monte Carlo method to numerically approximate the relative volume of $\mathrm{APPT}_{m,n}$ following ideas of similar type discussed by Fok and Crevier \cite{fok1989volume} by adapting to the probability simplex. We observe in Figure \ref{fig:relative_volume} that the relative volume of $\mathrm{APPT}_{m,n}$ for $m=2,3,4$ starts of within close range to the relative volume of the separable ball. As $n$ increases,  the volume follows along a similar decay rate as the inscribed polytope. 

\begin{figure}[ht]
    \centering
    \includegraphics[width=1\linewidth]{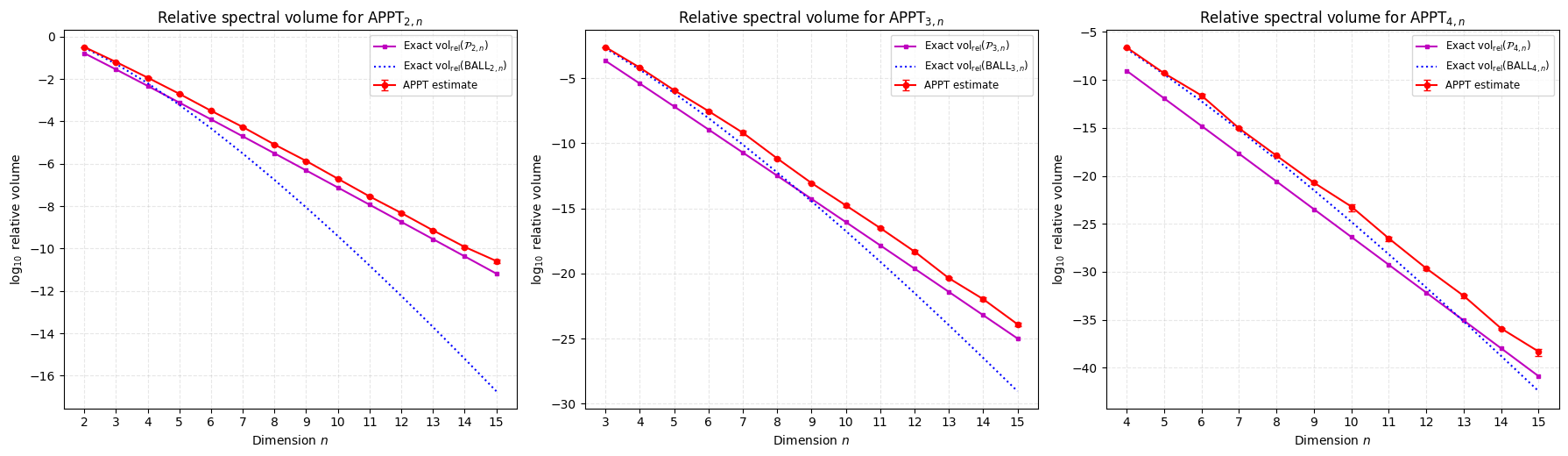}
    \caption{The estimated relative volume of $\text{APPT}_{m,n}$ for $m=2,3,4$ (red) in comparison with the exact polytope relative volume (purple) and exact separable ball volume (dotted blue).} 
    \label{fig:relative_volume}
\end{figure}

Our two-qubit relative spectral volume estimate of approximately $0.32618$ agrees with   the spectral volume of $ 0.32723006$ discussed in \cite{SongChen2025} within reasonable Monte Carlo error. 

\begin{figure}[ht]
    \centering
    \includegraphics[width=1\linewidth]{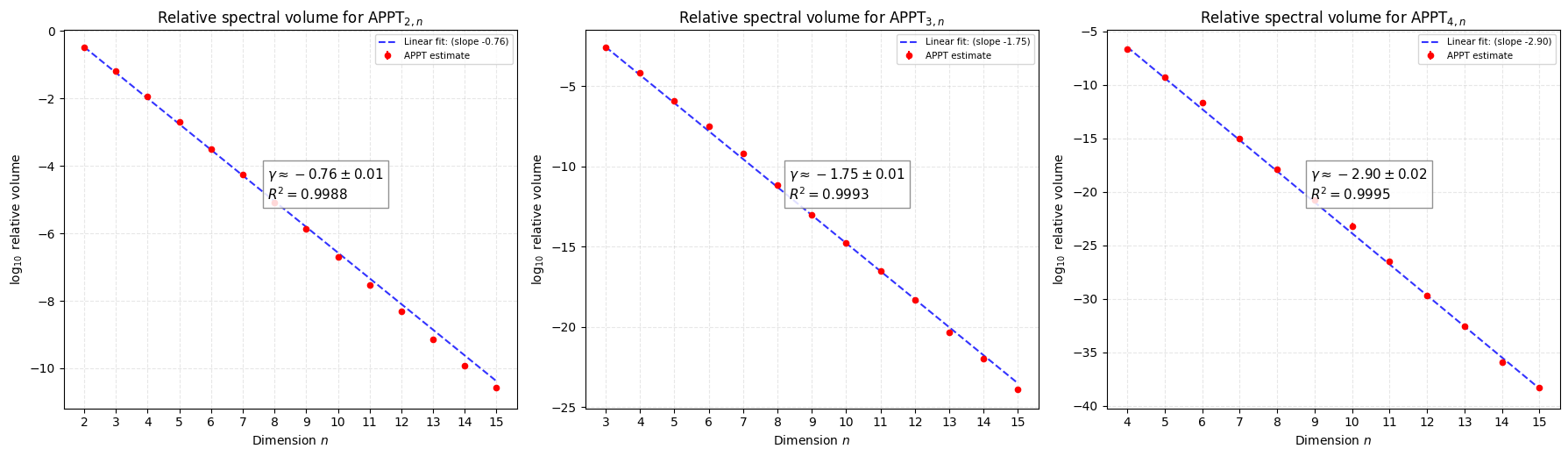}
    \caption{ Log-linear regression of the relative volume for $\mathrm{APPT}_{m,n}$ for $m=2,3,4$ whenever $n\leq 15$ respectively.}
    \label{fig:log-linear_volume_regression}
\end{figure}

\noindent The estimation reveals a strict log-linear decay of the relative spectral volume of $\mathrm{APPT}_{m,n}$ for $m=2,3,4$ and $ n\leq15$ as seen in Figure \ref{fig:log-linear_volume_regression}. To determine the asymptotic decay rate of the relative volume as $n$ increases, we perform a linear regression analysis to the log-volume estimates such that we fit 
\begin{equation}
    \log_{10} \mathrm{vol}_{\mathrm{rel}}(\text{APPT}_{m,n}) \approx   C_m + \gamma_m n, \text{ i.e. } \mathrm{vol}_{\mathrm{rel}}(\text{APPT}_{m,n}) \approx 10^{C_m}10^{\gamma_m \cdot n},
\end{equation} 
where $\gamma_m$ is the decay slope and $C_m<0$. 
This fit suggests that the relative volume of $\mathrm{APPT}_{m,n}$ appears to decay an exponentially in $n$ and much faster as $m$ increases. This indicates that for every unit increase in the local dimension $n$, the relative volume of $\mathrm{APPT}_{m,n}$ decreases by an approximate factor of $10^{-\gamma_m}$. Indeed, notice the shrink factor is 
$$\frac{\mathrm{vol}_{\mathrm{rel}}(\text{APPT}_{m,n}) }{\mathrm{vol}_{\mathrm{rel}}(\text{APPT}_{m,n+1}) } \approx  \frac{10^{C_m}10^{\gamma_m \cdot n} }{10^{C_m}10^{\gamma_m \cdot n+1} } = 10^{-\gamma_m}.$$
For instance, for $m=2$, we find the decay slope of $\mathrm{APPT}_{2,n}$ to be $\gamma_2\approx -0.7625$ which implies the approximate shrink factor per unit $n$ is $10^{0.7625} \approx 5.79$. Meaning that  $\mathrm{vol}_{\mathrm{rel}}(\text{APPT}_{2,n+1}) \approx \frac{1}{5.79}\mathrm{vol}_{\mathrm{rel}}(\text{APPT}_{2,n}) $ for all $n$ within our tested range. As seen in Table
 \ref{tab:volume_decay_comparison}, the relative volume of the polytope appears to capture the same rate of exponential decay as $\mathrm{APPT}_{m,n}$ on the log-linear scale as $n$ increases. This asymptotic behaviour suggests that although $\mathrm{APPT}_{m,n}$ may be larger than the polytope by a multiplicative factor, the exponential rate of change remains unaffected over the tested range. In comparison, the inscribed ball predicts a much faster decay rate typical of smooth convex bodies. 


\begin{table}[ht]
\centering
\begin{tabular}{c| c c c}

\hline
$m$
& $\mathrm{APPT}_{m,n}$ ($\gamma_m$)
& $\mathcal{P}_{m,n}$ 
& $\mathrm{BALL}_{m,n}$  \\
\hline
$2$ & $-0.7625$ & $-0.8020$ & $-1.2655$ \\
$3$ & $-1.7454$ & $-1.7793$ & $-2.2143$ \\
$4$ & $-2.9033$ & $-2.8975$ & $-3.2628$ \\
\hline
\end{tabular}
\caption{Approximate decay slopes for the relative volumes of
$\mathrm{APPT}_{m,n}$, the inscribed polytope $\mathcal P_{m,n}$, and the
separable ball, $\mathrm{BALL}_{m,n}$ when $m=2,3,4$ and $n\leq 15$.}
\label{tab:volume_decay_comparison}
\end{table}


\section{Conclusion}
\label{sec:conclusion}
The present investigation of the structure of the sets of absolutely separable and absolutely PPT states has revealed facets of the intricate geometry of their associated convex sets of spectra, as subsets of the appropriate probability simplex. We have introduced a permutation-symmetric reformulation of the absolute PPT criterion which allows us to identify the set of absolute PPT spectra as a spectrahedron. This representation removes the need to impose an a priori ordering of the spectrum of states and
makes the full permutation symmetry of $\mathrm{APPT}_{m,n}$ explicit. The spectrahedral description provides a unified approach to the geometry of the set of spectra of absolute separable and absolute PPT states. In particular, we express the topological boundary of $\mathrm{APPT}_{m,n}$ as the feasible part of the union of determinantal hypersurfaces and determine the number of irreducible determinant polynomials describing it. We have also developed a kernel-based description of the facial structure of the set, allowing us to characterize its boundary, faces and extreme points via a linear constraint matrix. With this, one can determine all possible facial dimensions of $\mathrm{APPT}_{m,n}$ and identify them by the rank of a single unified linear matrix constraint. Since the extreme points of a spectrahedron are all exposed extreme points, it provides a possible geometric route towards distinguishing absolute separability from absolute PPT whenever $m>2$. Since $\mathrm{ASEP}_{m,n} \subseteq \mathrm{APPT}_{m,n}$, the existence of an extreme point of $\mathrm{ASEP}_{m,n} $ which is not exposed in the set would imply that  $\mathrm{ASEP}_{m,n} \subsetneq \mathrm{APPT}_{m,n}$. Certainly, identifying an extreme point of $\mathrm{APPT}_{m,n} $ which is not absolute separable would immediately separate the two sets. 

We have also provided rigorous bounds on the maximal purity and minimal von Neumann entropy for absolute PPT by way of an inscribed polytope. 
Our numerical results indicate that these bounds are tight in several of the tested dimensions. 
Thus, another matter of interest is to resolve Conjecture \ref{conj:max_purity} in the positive, and to show that the linear condition in Eq.~\eqref{eq:sufficient_appt} is sufficient for absolute separability: because then the maximum purity over all three sets must coincide. Further, we have obtained an exact value for the volume of the polytope giving us a rigorous lower bound on the volume of absolute PPT and conditional on the containment  of the polytope in $\mathrm{ASEP}_{m,n}$, on absolute separability. Our spectral Monte Carlo estimates for $m=2,3,4$ show that the relative volume of $\mathrm{APPT}_{m,n} $ decreases sharply with $n$. Over the tested range, the fitted 
decay rates of $\mathrm{APPT}_{m,n} $ and the polytope are `close', whereas the
separable ball relative volume decays substantially faster. This suggests that, perhaps, the polytope captures a significant part of the spectral volume of $\mathrm{APPT}_{m,n} $ than the separable ball as the total dimension increases. Thus, our geometric descriptions might provide a viable road towards resolving the absolute separability problem.


\section*{Acknowledgments}
The authors thank Anna Sanpera, Karol {\.Z}yczkowski, Albert Rico, Jofre Abellanet-Vidal, Ilaria Svampa and Sayantan Chakraborty for insightful discussions; furthermore Marmaduke Temple and William Cooper for sound advice regarding the separability-entanglement frontier on various occasions.
JA, NBTK and AW are supported by the Spanish MICIN (project PID2022-141283NB-I00) with the support of FEDER funds and by the Alexander von Humboldt Foundation.
JA and AW were or are furthermore supported by the Spanish MICIN with funding from European Union NextGenerationEU (PRTR-C17.I1) and the Generalitat de Catalunya. 
NBTK was or is furthermore supported by ESA (EISI project 2021-01250-ESA) and by ICREA Academia.
AW was or is furthermore supported by the European Commission QuantERA grant ExTRaQT (Spanish MICIN project PCI2022-132965), by the Spanish MTDFP through the QUANTUM ENIA project: Quantum Spain and funded by the European Union NextGenerationEU within the framework of the ``Digital Spain 2026 Agenda''.



\printbibliography[heading=bibintoc]

@book{boyd2004convex,
  title={Convex Optimization},
  author={Boyd, Stephen and Vandenberghe, Lieven},
  year={2004},
  publisher={Cambridge University Press}
}

@article{ganguly2014witness,
  title={Witness of mixed separable states useful for entanglement creation},
  author={Ganguly, Nirman and Chatterjee, Jyotishman and Majumdar, Archan S.},
  journal={Physical Review A},
  volume={89},
  number={5},
  year={2014},
eid  = {052304},
  numpages = {5},
  month = {05},
  publisher = {American Physical Society},
  doi = {10.1103/PhysRevA.89.052304}
}

@thesis{netzerspectrahedra,
    author = {Netzer, Tim},
    title  = {Spectrahedra and Their Shadows},
    year   = {2012},
    school = {Universit{\"a}t Leipzig, Fakult{\"a}t f{\"u}r Mathematik und Informatik},
  type = {Habilitationsschrift},
    url    = {https://www.mathematik.uni-leipzig.de/preprints/p1207.0070.pdf}
}

@inproceedings{cynthiaspectrahedra,
  title={The geometry of spectrahedra},
  author={Vinzant, Cynthia},
  journal={Sum of Squares: Theory and Applications},
  booktitle={Proceedings of Symposia in Applied Mathematics},
  volume={77},
  pages={11--35},
  year={2020},
 publisher={Providence, RI, USA: American Mathematical Society}
}

@article{gurvits2002largest,
  title={Largest separable balls around the maximally mixed bipartite quantum state},
  author={Gurvits, Leonid and Barnum, Howard N.},
  journal={Physical Review A},
volume = {66},
  issue = {6},
  eid = {062311},
  numpages = {7},
  year = {2002},
  month = {12},
  publisher = {American Physical Society},
  doi = {10.1103/PhysRevA.66.062311}
}

@article{halder2021characterizing,
  title = {Characterizing the boundary of the set of absolutely separable states and their generation via noisy environments},
  author = {Halder, Saronath and Mal, Shiladitya and Sen(De), Aditi},
  journal = {Physical Review A},
  volume = {103},
  issue = {5},
  eid = {052431},
  numpages = {9},
  year = {2021},
  month = {05},
  publisher = {American Physical Society},
  doi = {10.1103/PhysRevA.103.052431}
}

@article{hildebrand2007positive,
author = {Hildebrand, Roland},
year = {2007},
month = {11},
title = {Positive partial transpose from spectra},
volume = {76},
number={5},
journal = {Physical Review A},
doi = {10.1103/PhysRevA.76.052325}
}

@article{ramana1995some,
  title={Some geometric results in semidefinite programming},
  author={Ramana, Motakuri and Goldman, Alan J.},
  journal={Journal of Global Optimization},
  volume={7},
  number={1},
  pages={33--50},
  year={1995},
  doi = {10.1007/BF01100204}
}

@book{rockafellar2015convex,
  title={Convex analysis},
  author={Rockafellar, Ralph Tyrrell},
  volume={28},
  year={1997},
  publisher={Princeton university press}
}

@inproceedings{gurvits2003classical,
author = {Gurvits, Leonid},
title = {{Classical deterministic complexity of Edmonds' Problem and quantum entanglement}},
year = {2003},
isbn = {1581136749},
publisher = {Association for Computing Machinery},
address = {New York, NY, USA},
doi = {10.1145/780542.780545},
booktitle = {Proceedings of the Thirty-Fifth Annual ACM Symposium on Theory of Computing},
pages = {10--19},
numpages = {10},
%location = {San Diego, CA, USA},
series = {STOC '03}
}

@article{horodecki2009quantum,
  title = {Quantum entanglement},
  author = {Horodecki, Ryszard and Horodecki, Pawe\l{} and Horodecki, Micha\l{} and Horodecki, Karol},
  journal = {Reviews of Modern Physics},
  volume = {81},
  number = {2},
  pages = {865--942},
  numpages = {0},
  year = {2009},
  month = {06},
  publisher = {American Physical Society},
  doi = {10.1103/RevModPhys.81.865}
}

@article{Horodecki1996Separability,
    author = {Horodecki, Micha{\l} and Horodecki, Pawe{\l} and Horodecki, Ryszard},
    title = {Separability of mixed states: necessary and sufficient conditions},
    journal = {Physical Letters A},
    volume = {223},
    number = {1--2},
    pages = {1--8},
    year = {1996},
    doi = {10.1016/S0375-9601(96)00706-2 }
}

@article{johnston2018inverse,
  title={The inverse eigenvalue problem for entanglement witnesses},
  author={Johnston, Nathaniel and Patterson, Everett},
  journal={Linear Algebra and its Applications},
  volume={550},
  pages={1--27},
  year={2018},
  publisher={Elsevier},
  doi={10.1016/j.laa.2018.03.043}
}

@article{SongChen2025,
  author    = {Song, Zhiwei and Chen, Lin},
  title     = {Extreme points of sets of absolutely separable and positive partial transpose states},
  journal   = {Physical Review A},
  volume    = {112},
  number    = {2},
  eid     = {022409},
  numpages ={20},
  year      = {2025},
  month ={08},
  publisher = {American Physical Society},
  doi = {10.1103/7pnp-qgrq}
}

@article{Lasserre2015,
  author    = {Lasserre, Jean-Bernard},
  title     = {Volume of slices and sections of the simplex in closed form},
  journal   = {Optimization Letters},
  volume    = {9},
  number    = {7},
  pages     = {1263--1269},
  year      = {2015},
  publisher = {Springer},
  doi       = {10.1007/s11590-015-0898-z}
}

@article{XiongSze2026,
  author    = {Xiong, L. and Sze, N.-S.},
  title     = {Spectral criteria for absolute positive partial transpose ({PPT}) in qudit-qudit state systems},
  journal   = {Journal of Mathematical Physics},
  volume    = {67},
  number    = {1},
  year      = {2026},
  month = {01},
   eid = {012201},
  publisher = {AIP Publishing},
  doi = {10.1063/5.0273779}
}

@article{guhne2009entanglement,
  title={Entanglement detection},
  author={G\"{u}hne, Otfried and T\'{o}th, G\'{e}za},
  journal={Physics Reports},
  volume={474},
  number={1-6},
  pages={1--75},
  year={2009},
  publisher={Elsevier},
  doi={10.1016/j.physrep.2009.02.004 },
  month={04}
}

@article{fawzi2021set,
  title={The set of separable states has no finite semidefinite representation except in dimension 3$\times$2},
  author={Fawzi, Hamza},
  journal={Communications in Mathematical Physics},
  volume={386},
  number={3},
  pages={1319--1335},
  year={2021},
  publisher={Springer},
 doi ={ 10.1007/s00220-021-04163-2}
}

@article{dieudonne1948generalisation,
  title={Sur une g{\'e}n{\'e}ralisation du groupe orthogonal {\`a} quatre variables},
  author={Dieudonn{\'e}, Jean},
  journal={Archiv der Mathematik},
  volume={1},
  number={4},
  pages={282--287},
  year={1948},
  publisher={Springer}, 
  doi ={10.1007/bf02038756}
}

@book{hartshorne1977,
  title={Algebraic Geometry},
  author={Hartshorne, Robin},
  journal={Algebraic Geometry},
  volume={52},
  year={1977},
  series={Graduate Texts in Mathematics},
  publisher={Springer},
  doi ={10.1007/978-1-4757-3849-0 }
}

@article{scheiderer2022extreme,
  title={Extreme points of Gram spectrahedra of binary forms},
  author={Scheiderer, Claus},
  journal={Discrete \& Computational Geometry},
  volume={67},
  number={4},
  pages={1174--1190},
  year={2022},
  publisher={Springer}, 
  doi ={10.1007/s00454-022-00385-w}
}

@article{jivulescu2015positive,
  title={Positive reduction from spectra},
  author={Jivulescu, Maria Anastasia and Lupa, Nicolae and Nechita, Ion and Reeb, David},
  journal={Linear Algebra and its Applications},
  volume={469},
  pages={276--304},
  year={2015},
  publisher={Elsevier},
  doi ={10.1016/j.laa.2014.11.031}
}

@book{horn2012matrix,
  title={Matrix analysis},
  author={Horn, Roger A and Johnson, Charles R},
  year={1985},
  publisher={Cambridge University Press}, 
  doi ={10.1017/CBO9780511810817}
}

@book{marshall1979inequalities,
   title = {Inequalities: Theory of Majorization and Its Applications},
  doi = {10.1007/978-0-387-68276-1},
  journal = {Springer Series in Statistics},
  publisher = {Springer New York},
  author = {Marshall,  Albert W. and Olkin,  Ingram and Arnold,  Barry C.},
  year = {2011}
  }

@article{phi2025maximum,
  title={On the maximum purity of absolutely separable bipartite states},
  author={Dung, Hoang Phi and Khoi, Vu The},
  journal={preprint arXiv:2510.19508},
  doi={10.48550/arXiv.2510.19508 },
  year={2025}
}

@article{zyczkowski1998volume,
  title={Volume of the set of separable states},
  author={\.{Z}yczkowski, Karol and Horodecki, Pawe{\l} and Sanpera, Anna and Lewenstein, Maciej},
  journal={Physical Review A},
  volume={58},
  number={2},
  pages = {883--892},
  year={1998},
  month = {08},
  publisher = {American Physical Society},
  doi = {10.1103/PhysRevA.58.883}
}

@article{fok1989volume,
  title={{Volume estimation by Monte Carlo methods}},
  author={Fok, Danny SK and Crevier, Daniel},
  journal={Journal of Statistical Computation and Simulation},
  volume={31},
  number={4},
  pages={223--235},
  year={1989},
  publisher={Taylor \& Francis}, 
  doi ={10.1080/00949658908811145}
}

@article{Clauser1969,
  title = {Proposed Experiment to Test Local Hidden-Variable Theories},
  author = {Clauser, John F. and Horne, Michael A. and Shimony, Abner and Holt, Richard A.},
  journal = {Physical Review Letters},
  volume = {23},
  issue = {15},
  pages = {880--884},
  year = {1969},
  month = {10},
  publisher = {American Physical Society},
  doi = {10.1103/PhysRevLett.23.880}
}

@article{ekert1991quantum,
  title = {Quantum cryptography based on {B}ell's theorem},
  author = {Ekert, Artur K.},
  journal = {Physical Review Letters},
  volume = {67},
  number = {6},
  pages = {661--663},
  year = {1991},
  month = {08},
  publisher = {American Physical Society},
  doi = {10.1103/PhysRevLett.67.661}
}

@article{bennett1993teleporting,
  title = {Teleporting an unknown quantum state via dual classical and {E}instein-{P}odolsky-{R}osen channels},
  author = {Bennett, Charles H. and Brassard, Gilles and Cr\'epeau, Claude and Jozsa, Richard and Peres, Asher and Wootters, William K.},
  journal = {Physical Review Letters},
  volume = {70},
  number = {13},
  pages = {1895--1899},
  numpages = {0},
  year = {1993},
  month = {03},
  publisher = {American Physical Society},
  doi = {10.1103/PhysRevLett.70.1895},
}

@article{schrodinger1935discussion,
  title = {Discussion of probability relations between separated systems},
  volume = {31},
 % ISSN = {1469-8064},
  DOI = {10.1017/s0305004100013554},
  number = {4},
  journal = {Mathematical Proceedings of the Cambridge Philosophical Society},
  publisher = {Cambridge University Press},
  author = {Schr\"{o}dinger, Erwin},
  year = {1935},
  month = {10},
  pages = {555--563}
}

@article{einstein1935can,
  title = {Can quantum-mechanical description of physical reality be considered complete?},
  author = {Einstein, Albert and Podolsky, Boris and Rosen, Nathan},
  journal = {Physical Review},
  volume = {47},
  number = {10},
  pages = {777--780},
  numpages = {0},
  year = {1935},
  publisher = {American Physical Society},
  month={05},
  doi = {10.1103/PhysRev.47.777},
}

@article{werner1989quantum,
  title = {Quantum states with {E}instein-{P}odolsky-{R}osen correlations admitting a hidden-variable model},
  author = {Werner, Reinhard F.},
  journal = {Physical Review A},
  volume = {40},
  number = {8},
  pages = {4277--4281},
  numpages = {0},
  year = {1989},
  month = {10},
  publisher = {American Physical Society},
  doi = {10.1103/PhysRevA.40.4277}
}

@article{gharbian2010,
  title={{Strong NP-hardness of the quantum separability problem}},
  author={Gharibian, Sevag},
  journal={Quantum Information and Computation},
  volume={10},
  pages={343--360},
  year={2010}
}

@article{asherperes1996,
  title = {Separability criterion for density matrices},
  author = {Peres, Asher},
  journal = {Physical Review Letters},
  volume = {77},
  number = {8},
  pages = {1413--1415},
  year = {1996},
  month = {08},
  publisher = {American Physical Society},
  doi = {10.1103/PhysRevLett.77.1413}
}

@article{kus2001geometry,
  title={Geometry of entangled states},
  author={Ku\'{s}, Marek and \.{Z}yczkowski, Karol},
  journal={Physical Review A},
  volume = {63},
  number = {3},
  eid = {032307},
  numpages = {13},
  year = {2001},
  month = {02},
  publisher = {American Physical Society},
  doi = {10.1103/PhysRevA.63.032307}
}

@misc{knill2003,
  title={Separability from spectrum},
  author={Knill, Emmanuel},
  note = {Open problems in quantum information theory},
  url= {http://qig.itp.uni-hannover.de/qiproblems/15},
  year={2003}
}

@article{johnston2014separability,
  title={Is absolute separability determined by the partial transpose?},
  author={Arunachalam, Srinivasan and Johnston, Nathaniel and Russo, Vincent},
   journal={Quantum Information and Computation},
  volume = {15},
  number = {7--8},
  pages ={0694--0720},
  year={2015}
}

@article{johnston2013sepfromspectra,
  title={Separability from spectrum for qubit-qudit states},
  author={Johnston, Nathaniel},
  journal={Physical Review A},
  volume={88},
  number={6},
  eid={062330},
  numpages={5},
  year={2013},
  month = {12},
  publisher ={American Physical Society},
   doi = {10.1103/PhysRevA.88.062330}
}

@article{abellanet2025sufficient,
  title={Sufficient criteria for absolute separability in arbitrary dimensions via linear map inverses},
  author={Abellanet-Vidal, Jofre and M\"{u}ller-Rigat, Guillem and Rajchel-Mieldzio\'{c}, Grzegorz and Sanpera, Anna},
  journal={Reports on Progress in Physics},
  volume={88},
  number={10},
  eid={107601},
  year={2025},
  doi ={10.1088/1361-6633/ae0cfa },
  publisher={IOP Publishing}
}

@article{kondra2026fundamental,
  title={Fundamental limitations on entanglement extraction from purity},
  author={Kondra, Tulja Varun and Hita, Pedro Barrios and Neumann, Justus and Kampermann, Hermann and Bru{\ss}, Dagmar},
  journal={preprint arXiv:2605.29197},
  doi={10.48550/arXiv.2605.29197 },
  year={2026}
}

@article{wang2026extreme,
  title={{Extreme points of absolutely PPT states with exactly three distinct eigenvalues}},
  author={Wang, Nalan and Chen, Lin and Song, Zhiwei},
  journal={preprint arXiv:2603.20717 },
    primaryClass={quant-ph},
    doi={10.48550/arXiv.2603.20717 }, 
  year={2026}
}

@article{cao2004determinant,
  title={Determinant preserving transformations on symmetric matrix spaces},
  author={Cao, Chongguang and Tang, Xiaomin},
  journal={The Electronic Journal of Linear Algebra},
  volume={11},
  pages={205--211},
  year={2004},
  publisher = {University of Wyoming Libraries},
  month ={01},
  doi ={10.13001/1081-3810.1133 }
}

@book{barvinok:convexity,
    author = {Alexander Barvinok},
    title = {A Course in Convexity},
    publisher = {American Mathematical Society},
    year = {2002},
    series = {Graduate Studies in Mathematics},
    volume = {54}
}

@book{BCR:RAG,
    author = {Jacek Bochnak and Michel Coste and Marie-Fran{\c{c}}oise Roy},
    title = {Real Algebraic Geometry },
    publisher = {Springer Verlag},
    year = {1998},
    series = {Ergebnisse der Mathematik und ihrer Grenzgebiete; Modern Surveys in Mathematics},
    volume = {36},
    doi = {10.1007/978-3-662-03718-8}
}

@book{tent2012course,
  title={A course in model theory},
  author={Tent, Katrin and Ziegler, Martin},
  number={40},
  year={2012},
  publisher={Cambridge University Press}
}

@book{hardy1952inequalities,
  title={Inequalities},
  author={Hardy, Godfrey Harold and Littlewood, John Edensor and P{\'o}lya, George},
  year={1952},
  publisher={Cambridge University Press}
}

@article{rico2026threshold,
  title={Threshold entanglement sharing: quantum states with absolutely separable marginals},
  author={Rico, Albert and Abellanet-Vidal, Jofre and Kothakonda, Naga Bhavya Teja and Sanpera, Anna and Munn{\'e}, Gerard Angl{\`e}s},
  journal={ preprint arXiv:2604.13169},
  primaryClass={quant-ph},
  doi = {10.48550/arXiv.2604.13169 },
  year={2026}
}

@article{Verstraete,
  title = {Maximally entangled mixed states of two qubits},
  author = {Verstraete, Frank and Audenaert, Koenraad and De Moor, Bart},
  journal = {Physical Review A},
  volume = {64},
  issue = {1},
  pages = {012316},
  numpages = {6},
  year = {2001},
  month = {Jun},
  publisher = {American Physical Society},
  doi = {10.1103/PhysRevA.64.012316}
}


\newpage
\appendix
\section{Proofs of certain results in Subsection \ref{subsec:boundary}}\label{app:Appendix_A1}

In this section, we  
prove Lemmas \ref{lem:polynomial_orbit} and \ref{lem:signed-permutation}, and Theorem \ref{thm:kappa-irreducible-components}, which we restate for the reader's convenience. 

\medskip
\noindent
\textbf{Lemma \ref{lem:polynomial_orbit}.}\ 
\emph{The set of determinant polynomials 
$P = \{ \det(L_\pi(\lambda))\}_{\pi \in S_{mn}}$ of $\mathrm{APPT}_{m,n}$ is the group orbit under $S_{mn}$ given by
\begin{equation*}
  P = \mathrm{Orb}_{S_{mn}}(\det(L_{\mathrm{id}}(\lambda))),
\end{equation*}
where $\det(L_{\mathrm{id}}(\lambda))$ is the determinant polynomial under the identity permutation.} 

\begin{proof}
Recall that the left action of the symmetric group $S_{mn} $ on $\mathbb{R}^{mn}$ induces a family of ring automorphisms $\phi_{\pi} :\mathbb{R}[\lambda]\to \mathbb{R}[\lambda]$  for each permutation $\pi \in S_{mn}$ defined by $ \phi_{\pi}(h(\lambda)) = h(\pi^{-1} \cdot \lambda)$ where $h(\lambda) = \det(\mathcal{L}(\lambda))$.

First, we prove that $\phi_{\pi}$ is indeed a ring automorphism. Consider arbitrary polynomials $h, g \in \mathbb{R}[\lambda]$ and permutations $\pi \in S_{mn}$. Notice that $\phi_{\pi}$ by definition preserves the addition and multiplication axioms as \begin{equation*}
    \phi_\pi(h + g)(\lambda) = (h + g)(\pi^{-1} \cdot \lambda) = h(\pi^{-1} \cdot \lambda) + g(\pi^{-1} \cdot \lambda) = \phi_\pi(h)(\lambda) + \phi_\pi(g)(\lambda)
\end{equation*}
and 
\begin{equation*}
    \phi_\pi(hg)(\lambda) = (hg)(\pi^{-1} \cdot \lambda) = h(\pi^{-1} \cdot \lambda) g(\pi^{-1} \cdot \lambda) = \phi_\pi(h)(\lambda) \phi_\pi(g)(\lambda).
\end{equation*}
Additionally, each permutation $\pi\in S_{mn}$ admits a unique inverse $\pi^{-1} \in S_{mn}$. Thus, by applying the corresponding mapping of the inverse, it follows that
\begin{equation*}
    \phi_{\pi^{-1}}(\phi_\pi(h(\lambda))) = h\big((\pi^{-1})^{-1}\cdot (\pi^{-1} \cdot \lambda)\big) = h(\pi \pi^{-1} \cdot \lambda) = h(id \cdot \lambda) = h(\lambda).
\end{equation*} 
Since $\phi_\pi$ satisfies all axioms, it defines a ring automorphism. 

Now we proceed to show that $\phi_\pi$ commutes with the determinant to generate the group orbit. Consider the  linear matrix constraint $L_{\mathrm{id}}(\lambda)$  under the identity permutation. By definition as seen in Eq.~\eqref{eq:base_matrix}, each entry term is a linear homogeneous polynomial in  $\mathbb{R}[\lambda]$ for all permutations $\pi$. Since $\phi_\pi$ is an automorphism, it acts linearly over the field $\mathbb{R}$ such that by extending it to the matrix ring, $\phi_\pi$ acts on the entries of $L_{\mathrm{id}}$ as
\begin{equation}\begin{split}
\label{eq:matrix_ring}
    \phi_\pi \left( [L_{\mathrm{id}}(\lambda)]_{i,j} \right)
    &= \begin{cases} 
\phi_\pi \left(2\lambda_{p(i,i)}\right)& \text{if } i = j,\\
\phi_\pi \left(\lambda_{p(i,j)} - \lambda_{q(i,j)}\right)& \text{if } i \neq j,
    \end{cases} \\
    &=  \begin{cases} 
2\lambda_{\pi(p(i,i))} & \text{if } i = j, \\
\lambda_{\pi(p(i,j))} - \lambda_{\pi(q(i,j))} & \text{if } i \neq j,
    \end{cases} \\
    &= [L_{\pi}(\lambda)]_{i,j}
\end{split}\end{equation}
This way, it follows that $L_\pi(\lambda) = \phi_\pi \left( L_{\mathrm{id}}(\lambda) \right)$. Now, the determinant $h_{\mathrm{id}}(\lambda) = \det(L_{\mathrm{id}}(\lambda))$ is a homogeneous polynomial of degree $m$ constructed via the addition and multiplication of the matrix entries of $L_{\mathrm{id}}(\lambda)$. Thus, since $\phi_\pi$ is a ring automorphism, we must have that $\det(\phi_\pi(L_{\mathrm{id}}(\lambda))) = \phi_\pi(\det(L_{\mathrm{id}}(\lambda)))$. Therefore, it follows that for any arbitrary permutation $\pi\in S_{mn}$, 
\begin{equation}\label{eq:determinant_permute}
    h_\pi(\lambda) = \det(L_\pi(\lambda)) = \det(\phi_\pi(L_{\mathrm{id}}(\lambda))) = \phi_\pi(h_{\mathrm{id}}(\lambda)).
\end{equation}
By definition, $\mathrm{Orb}_{S_{mn}}(h_{\mathrm{id}}) = \{ \phi_\pi(h_{\mathrm{id}}(\lambda)) \mid \pi \in S_{mn} \}.$ Therefore, it follows from Eq.~\eqref{eq:determinant_permute} that 
\begin{equation*}
   P = \{ h_\pi(\lambda) \mid \pi \in S_{mn} \} = \mathrm{Orb}_{S_{mn}}(h_{\mathrm{id}}),
\end{equation*}
concluding the proof. 
\end{proof}

\begin{lemma}
\label{lem:induced_linear_map}
Suppose that $\det(L_{\pi}(\lambda)) = \det(L_{\mathrm{id}}(\lambda))$ for some permutation $\pi \in S_{mn}$, then the map  $T_\pi: \mathrm{Sym}_m(\RR)\longrightarrow\mathrm{Sym}_m(\RR)$ given by
$$T_\pi(L_{\mathrm{id}}(\lambda))
    =L_\pi(\lambda),$$
is well defined, linear and invertible over the space of real symmetric matrices.
\end{lemma}

\begin{proof} 
By construction, both $L_{\mathrm{id}}:
\mathbb{R}^{mn}\longrightarrow\mathrm{Sym}_m(\mathbb{R})$ and $ L_\pi:
    \mathbb R^{mn}\longrightarrow\mathrm{Sym}_m(\mathbb{R})
$ are surjective linear maps. Indeed, the diagonal entries can be chosen independently using
the variables $\lambda_{p(i,i)}$, and similarly,  the off-diagonal entries can be
chosen independently using the differences
   $ \lambda_{p(i,j)}-\lambda_{q(i,j)}$.

Suppose $ L_{\mathrm{id}}(\lambda) =  L_{\mathrm{id}}(\mu) $ for $\lambda, \mu\in \RR^{mn}$. By linearity, $L_{\mathrm{id}}(\lambda-\mu) = 0$ implying that $\lambda-\mu\in \ker(L_{\mathrm{id}})$. Thus, for well-definedness, it is enough to show that $\lambda-\mu\in \ker(L_\pi )$, i.e.,  $\ker(L_{\mathrm{id}})\subseteq \ker(L_\pi ) $. 
Let $\mathbf{z}\in\ker(L_{\mathrm{id}})$ and set $\mu=\lambda+t\mathbf{z} $ for $t\in\RR$. Then since $ L_{\mathrm{id}}(\lambda) =  L_{\mathrm{id}}(\lambda+t\mathbf{z}) $, it follows that $\det(L_{\mathrm{id}}(\lambda)) = \det(L_{\mathrm{id}}(\lambda +t\mathbf{z}))$. By the assumption that $\det(L_{\pi}) = \det(L_{\mathrm{id}})$, we have the equivalence $\det L_\pi(\lambda+t\mathbf{z})= \det L_\pi(\lambda)$. Since $L_\pi$ is surjective, the matrix ranges over all of the symmetric space $\mathrm{Sym}_m(\mathbb{R}) $. Therefore, 
$$ \det(X +t L_\pi)= \det X$$
 for all $X \in\mathrm{Sym}_m(\mathbb{R})$ and $t\in\RR$. Taking $X = s\1_m$ and $t=1$, we have that $\det(s\1_m + L_\pi)= s^m $ for all $s \in \RR.$ Since $L_\pi$ is real symmetric, we can assume that $\mathbf{z}=\spec(L_\pi)$. This implies that $s^m = \prod_{i=1}^{m}(s+z_i)$. Therefore, we must have that for each eigenvalue of $L_\pi$, $z_1=z_2=\cdots=z_m$. Thus, $L_\pi = 0$ and consequently $\mathbf{z} \in\ker( L_\pi )$, implying $\lambda-\mu\in \ker( L_\pi)$ and $  L_\pi(\lambda) =L_\pi(\mu)$. Thus, 
 $\ker(L_{\mathrm{id}})\subseteq \ker(L_\pi )$.
In addition, since both \(L_{\mathrm{id}}\) and \(L_\pi\) are surjective maps onto the
same finite-dimensional space, their kernels have the same dimension. Thus, $\ker(L_\pi )\subseteq \ker(L_{\mathrm{id}})$, which results in  $\ker (L_{\mathrm{id}})=\ker (L_\pi)$.

Linearity of $T_\pi$ follows directly from the linearity of
$L_{\mathrm{id}}$ and $L_\pi$. Finally, $T_\pi$ is surjective since for every $Y\in\mathrm{Sym}_m(\mathbb R)$, we can always find  a  vector $\lambda$ such that $ L_\pi(\lambda)=Y$. Then $Y=T_\pi\left( L_{\mathrm{id}}(\lambda)\right).$ Thus, since $T_\pi$ is a surjective linear map from
$\mathrm{Sym}_m(\mathbb R)$ to itself, it is invertible.
\end{proof}

\noindent
\textbf{Lemma \ref{lem:signed-permutation}.}\ 
\emph{
As polynomials, $\det(L_{\pi}(\lambda)) = \det(L_{id}(\lambda))$ for some permutation $\pi \in S_{mn}$, if and only if there exists a signed permutation matrix $A$ such that $L_{\pi}(\lambda) = A L_{id}(\lambda) A^T$.}

\begin{proof}
By Lemma~\ref{lem:induced_linear_map}, the permutation $\pi \in  S_{mn} $ induces an
invertible linear map satisfying $T_\pi(L_{\mathrm{id}}(\lambda))= L_\pi(\lambda)
 $. Moreover, since $\det( L_\pi(\lambda))=\det (L_{\mathrm{id}}(\lambda))$,
and \(L_{\mathrm{id}}\) is surjective, one has $\det( T_\pi(X))=\det( X)$ for every $X\in\mathrm{Sym}_m(\mathbb R)$. Thus, by the theorem of linear determinant preservers on the space of real symmetric matrices \cite{cao2004determinant,dieudonne1948generalisation}, there exist $\alpha\in\mathbb R\setminus\{0\}$ and an invertible matrix $B$ such that
 $ T_\pi(X)=\alpha BXB^T$ for every $X\in\operatorname{Sym}_m(\mathbb R).$ As such, 
\begin{equation}
\label{eq:Lpi_alpha_congruence}
    L_\pi(\lambda)=\alpha B L_{\mathrm{id}}(\lambda)B^T.
\end{equation}
The $k$-th diagonal entry of the r.h.s. of Eq.~\eqref{eq:Lpi_alpha_congruence} has the quadratic form 
\begin{equation} \label{eq:matrix_diag}
    [\alpha B L_{id}(\lambda) B^T]_{k,k} = 2\alpha\sum_{i=1}^m x_{k,i}^2 \lambda_{p(i,i)} + 2\alpha\sum_{1\leq i<j\leq m} x_{k,i} x_{k,j} \left(\lambda_{p(i,j)} - \lambda_{q(i,j)}\right),
\end{equation}
where $(x_{k,1}, \dots, x_{k,m})$ denotes the $k$-th row vector of $B$. On the other hand, $[L_\pi(\lambda)]_{k,k} =2\lambda_{\pi(p(k,k))}.$ Therefore, 
\begin{equation}
\label{eq:positive_parity_M}
    2\lambda_{\pi(p(k,k))} = \sum_{i=1}^m 2x_{k,i}^2 \lambda_{p(i,i)} +  \sum_{1\leq i<j\leq m} x_{k,i} x_{k,j} \left(\lambda_{p(i,j)} - \lambda_{q(i,j)}\right).
\end{equation}
Since the coordinates  $\{\lambda_k\}_{k=1}^{mn}$ are linearly independent, their coefficients on both sides of  Eq.~\eqref{eq:positive_parity_M} must be equal. Suppose $ \pi(p(k,k))\notin\{p(1,1),\ldots,p(m,m)\}, $ then the coefficient of each diagonal variable
$\lambda_{p(i,i)}$ on the left-hand side must equal zero. Comparing the coefficients of these variables gives $\alpha x_{k,i}^2=0$, hence $x_{k,i} = 0$ for all $i$ as $\alpha\neq 0$. This means that the $k$-th row of $B$ is the vector $(0, \dots, 0)$. This would imply that Eq.~\eqref{eq:positive_parity_M} reduces to the polynomial identity $2\lambda_{\pi(p(k,k))} =0$. This creates a contradiction as $2\neq 0$. Thus, for linear independence, the unique coordinate $2\lambda_{\pi(p(k,k))} $ must equal precisely one diagonal variable on the right-hand side of Eq.~\eqref{eq:positive_parity_M}. Therefore, there must exist a unique index $l \in \{1, \ldots, m\}$ such that $\pi(p(k,k)) = p(l,l)$ so their coefficients are equal: $ \alpha x_{k,l}^2 = 1 $.
Thus, for all $i\neq l$, the linear independence of the coordinates reduces the other diagonal coefficients to zero, so that $\alpha x_{k,i}^2 = 0$ and hence $x_{k,i} = 0$.
In particular, $\alpha>0$, and every $k$-th row of $B$ contains exactly one non-zero entry, $x_{k,l}=\pm\frac1{\sqrt{\alpha}}, $ for all $k\in\{1,\ldots, m\}$. And since $B$ is an invertible matrix, no two rows share a non-zero entry in the same column. Set $A:=\sqrt{\alpha}\,B.$ Then every nonzero entry of \(A\) is equal to \(1\) or \(-1\), and every
row and column contains exactly one nonzero entry. Thus, $A$ must be a signed permutation matrix such that $L_\pi(\lambda)=\alpha B L_{\mathrm{id}}(\lambda)B^T= A L_{\mathrm{id}}(\lambda)A^T.$

Conversely, let $A = DC_\pi$ be a signed permutation matrix such that $C_\pi$ is the standard permutation matrix associated with  a permutation $\pi\in S_{mn}$ and $D = \diag(\epsilon_1, \ldots, \epsilon_m)$ with $ \epsilon_i \in \{ -1,1\}$ for all $i= 1,\ldots, m$. Then $C_\pi$ simultaneously rearranges the rows and columns of $L_{\mathrm{id}}(\lambda)$ by permuting the diagonal and off-diagonal pairs of the matrix. By conjugating the matrix $L_{\mathrm{id}}(\lambda)$ with $C_\pi$ each entry gains a coefficient, $\epsilon_i\epsilon_j$. When $\epsilon_i\epsilon_j<0 $, the off-diagonal pair $\lambda_{p(i,j)}-\lambda_{q(i,j)}$ is swapped by $\lambda_{q(i,j)}-\lambda_{p(i,j)} $. Thus there exists a permutation $\pi^*\in S_{mn}$ satisfying 
$L_{\pi^*}(\lambda)=A L_{\mathrm{id}}(\lambda)A^T.$ Since permutation matrices are orthogonal, $\det(A) = \pm 1$, and consequently 
$$\det(L_{\pi^*}(\lambda))  = \det(A)^2 \det(L_{\mathrm{id}}(\lambda)) = \det(L_{\mathrm{id}}(\lambda)),$$
which concludes the proof.
\end{proof}

\noindent 
\textbf{Theorem \ref{thm:kappa-irreducible-components}.}\ 
\emph{Let $\partial \mathrm{APPT}_{m,n}$ be  the topological boundary of the set of absolute $\mathrm{PPT}$ spectra and $V_{mn}=\{\lambda\in \RR^{mn}\mid \sum_{i=1}^{mn}\lambda_i=1 \}$ the affine hyperplane of normalized spectra. Then there exists permutations $\pi_1,\ldots,\pi_{\kappa_{m,n}}\in \widetilde{S}$ such that 
$$ \partial\mathrm{APPT}_{m,n}=\mathrm{APPT}_{m,n}\cap\left(\bigcup_{i=1}^{\kappa_{m,n}}
    \mathcal{Z}_{V_{mn}} \left(
        \det L_{\pi_i}(\lambda)
    \right)\right),$$ 
where 
$$\kappa_{m,n} = \frac{(mn)!}{(mn-m^2)!\cdot 2^{m-1} \cdot m!}.$$  
Furthermore, each algebraic set $\mathcal{Z}_{V_{mn}} \left( \det L_{\pi_i} \right) $ taken with respect to the affine space $V_{mn}$ is a distinct irreducible hypersurface of dimension $(mn-2)$.}


\begin{proof}
Define $h_{\mathrm{id}}(\lambda)\coloneq\det(L_{\mathrm{id}}(\lambda))$. By Lemma \ref{lem:polynomial_orbit}, the set of  determinant polynomials is the orbit $ \mathrm{Orb}_{S_{mn}}(h_{\mathrm{id}}(\lambda))$. Thus, the number of distinct determinant polynomials is the order of this group orbit: $\kappa_{m,n}=|\mathrm{Orb}_{S_{mn}}(h_{\mathrm{id}}(\lambda))|$. Let the stabilizer set of $h_{\mathrm{id}}$ be defined by $$\mathrm{Stab}_{S_{mn}}(h_{\mathrm{id}}) = \{ \pi \in S_{mn} \mid \pi \cdot h_{\mathrm{id}} =h_{\mathrm{id}}  \}.$$ Then by the orbit-stabilizer theorem, 
\begin{equation}\label{eq:stabilizer_orbit}
    \kappa_{m,n} = \frac{|S_{mn}|}{|\mathrm{Stab}_{S_{mn}}(h_{\mathrm{id}})|} = \frac{(mn)!}{|\mathrm{Stab}_{S_{mn}}(h_{\mathrm{id}})|}.
\end{equation}Thus, it remains to identify the stabilizer of $h_{\mathrm{id}} $ and find how many distinct permutations satisfy the stabilizer condition. 

Define the total index set $I \coloneq \{1,\ldots,mn\}$ for any spectrum $\lambda\in \text{APPT}_{m,n}$ and consider the constraint matrix $L_{\mathrm{id}}(\lambda)$ defined by the identity permutation. Since $L_{\mathrm{id}}(\lambda)$ is an $m \times m$ symmetric matrix, it consists of $m$ diagonal entries and $\frac{m(m-1)}{2}$ distinct off-diagonal pairs which altogether constitute an active index $J \subset I$. More precisely,  $J=\left\{p(i,i): 1\leq i\leq m\right\}\cup\left\{p(i,j),q(i,j):1\leq i<j\leq m\right\}.$
Then, the cardinality of the active set is $|J | = m^2$, leaving $mn-m^2$ eigenvalues which do not appear in the matrix constraint. These inactive eigenvalues correspond to the index set $I\setminus J$ so that $|I\setminus J| = mn-m^2.$ Since each of the eigenvalues with index in $I\setminus J$ do not appear in the matrix, any permutation acting exclusively on these indices leave the matrix and subsequently, its determinant, invariant. Thus, there exists a subgroup $H_{I\setminus J}$ of the stabilizer group $ \mathrm{Stab}_{S_{mn}}(h_{\mathrm{id}})$ such that 
\begin{equation*}
    H_{I\setminus J}= \{\pi \in S_{mn} : \pi(k) = k \quad \forall\; k \in J \}.
    \end{equation*}
Since any $\pi\in H_{I\setminus J}$ only permutes the $mn-m^2$ indices within $I\setminus J$, the subgroup $ H_{I\setminus J}$ is canonically isomorphic to the symmetric group $ S_{(mn-m^2)}$ and therefore has order $|H_{I\setminus J}| = (mn-m^2)!$.

Next, we consider the permutations acting solely on the active index set $J$ that preserve the stabilizer condition.  By Lemma~\ref{lem:signed-permutation}, for a permutation
$\pi\in S_{mn}$, $\pi\in\mathrm{Stab}_{S_{mn}}(h_{\mathrm{id}})$ if and only if there exists a signed permutation matrix $A$, independent
of $\lambda$, such that $ L_\pi(\lambda)=A L_{\mathrm{id}}(\lambda)A^T.$ Here, $A = D\Pi$ is such that $\Pi$ is the standard permutation matrix and $D = \diag(\epsilon_1, \ldots, \epsilon_m)$ with $ \epsilon_i \in \{ +1,-1\}$ for all $i= 1,\ldots, m$.
Since the indices of $J$ are constrained to the entries of the $m\times m$ matrix $L_{\mathrm{id}}(\lambda),$ any valid permutation must preserve the structure of the matrix.
In particular, every stabilizing permutation preserves the active set
$J$. Indeed, notice that the right-hand side
   $ A L_{\mathrm{id}}(\lambda)A^T$
depends only on the coordinates indexed by $J$ whereas $L_\pi(\lambda)$ depends on the coordinates indexed by $\pi(J)$. Thus $\pi(J) = J$. As such, the only valid permutations $\pi$ acting on $J$ are exactly those that simultaneously rearrange the rows and columns of $L_{\mathrm{id}}(\lambda)$ (here, $A = I\cdot \Pi$ ) or simultaneously flip the sign of specific rows and their corresponding columns (here, $A = D\cdot I$ ). 

First, we consider the case where $A = I\cdot \Pi_\tau $ where $\tau\in S_{m}$ represents the bijection $\tau: \{1,\ldots,m\} \to \{1,\ldots,m\}$ that simultaneously rearranges  the rows and columns of $L_{\mathrm{id}}(\lambda)$ and $\Pi_{\tau}$ denotes its associated matrix permutation. By conjugating the matrix $L_{\mathrm{id}}(\lambda)$ with $\Pi_{\tau}$, we obtain a new matrix $L^*(\lambda)$ with rearranged rows and columns such that $L^*(\lambda) = \Pi_{\tau} L_{\mathrm{id}}(\lambda) \Pi_{\tau}^T.$ Since permutation matrices are orthogonal, $\det(\Pi_{\tau}) = \pm 1$ and so it follows that \begin{equation}\label{eq:tau_determinant}
    \det(L^*(\lambda)) = \det(\Pi_\tau) \det(L_{\mathrm{id}}(\lambda)) \det(\Pi_\tau^T) = (\pm 1)^2 \det(L_{\mathrm{id}}(\lambda)) = h_{\mathrm{id}}(\lambda).
\end{equation}More importantly, rearranging the rows and columns repositions the eigenvalues associated with indices in $J$ within the matrix $L^*(\lambda)$. By doing so, we can always find a specific permutation $\pi_\tau \in S_{mn}$ induced by $\tau$ such that for all $k \in I$
\begin{equation*}
    \pi_\tau(k) = \begin{cases} 
      p((\tau(i)\wedge \tau(j)), (\tau(i)\vee\tau(j) ))& \text{if } k = p(i,j) \text{ for some } i,j \le m, \\
      q((\tau(i)\wedge \tau(j)), (\tau(i)\vee\tau(j) )) & \text{if } k = q(i,j) \text{ for some } i \neq j \le m, \\
      k & \text{if } k \notin J.
   \end{cases}
\end{equation*}
Thus, it can be seen that for this index permutation $\pi_\tau(k)$, we have $L^*(\lambda) =\Pi_{\tau} L_{\mathrm{id}}(\lambda) \Pi_{\tau}^T  =L_{\pi_\tau}(\lambda) $ whenever the signed permutation matrix is $A = I\cdot \Pi_\tau$ and it follows directly from Eq.~\eqref{eq:tau_determinant} that the stabilizer condition is preserved for all $\pi_\tau\in S_{mn}$.  Hence, $\pi_\tau\in\mathrm{Stab}_{S_{mn}}(h_{\mathrm{id}})  $. Since the set of all bijections $\tau\in S_m$ acting on the $m$ rows and columns of the matrix $L_{\mathrm{id}}(\lambda)$ generates a corresponding set of induced permutations $\pi_\tau\in S_{mn}$, these permutations form a subgroup $H_{J_{\tau}}$ of the stabilizer group $\mathrm{Stab}_{S_{mn}}(h_{\mathrm{id}})$. This subgroup is canonically isomorphic to the symmetric group $S_m$ and therefore $|H_{J_{\tau}}|= m!$.

On the other hand, consider the case when $A= D$
where  $D = \mathrm{diag}(\epsilon_1, \epsilon_2, \dots, \epsilon_m)$ and $\epsilon_i \in \{1, -1\}$. Then $L^*(\lambda) = D L_{\mathrm{id}}(\lambda) D^T$. Since $D$ is orthogonal and $D^T = D$, $\det(D) = \prod_{i=1}^m \epsilon_i = \pm 1$, so that the determinant polynomial $h_{\mathrm{id}}(\lambda)$ is preserved as
\begin{equation}\label{eq:D_determinant}
    \det(L^*(\lambda)) = \det(D) \det(L_{\mathrm{id}}(\lambda)) \det(D) = (\pm 1)^2 \det(L_{\mathrm{id}}(\lambda)) = h_{\mathrm{id}}(\lambda).
\end{equation}
Notice that the diagonal entries remain unchanged as $[D L_{\mathrm{id}} D]_{i,i} = \epsilon_i^2 (2\lambda_{p(i,i)}) = 2\lambda_{p(i,i)}$ while the off-diagonal entries may change as follows whenever $i<j$:
\begin{equation*}
 [D L_{id} D]_{i,j} = \epsilon_i \epsilon_j \left( \lambda_{p(i,j)} - \lambda_{q(i,j)} \right) =
    \begin{cases}
        \lambda_{p(i,j)} - \lambda_{q(i,j)}  & \text{ if } \epsilon_i \epsilon_j = 1,\\
        \lambda_{q(i,j)} - \lambda_{p(i,j)} & \text{ if } \epsilon_i \epsilon_j = -1.
    \end{cases}
\end{equation*}
Since the indices of $J$ are constrained to the coordinates above, there exists an index permutation $\pi_D\in S_{mn}$ such that for $\epsilon_i \epsilon_j = -1$, the off-diagonal eigenvalue pair $\lambda_{p(i,j)} , \lambda_{q(i,j)}$ swap positions. Therefore, this permutation induced by the matrix $D$ is defined for all indices $k\in I$ such that 
\begin{equation*}
    \pi_D(k) = \begin{cases} 
      q(i,j) & \text{if } k = p(i,j) \text{ and } \epsilon_i \epsilon_j = -1, \\
      p(i,j) & \text{if } k = q(i,j) \text{ and } \epsilon_i \epsilon_j = -1, \\
      k & \text{otherwise}.
   \end{cases}
\end{equation*} Under this permutation, it follows that $L^*(\lambda) = D L_{\mathrm{id}}(\lambda) D^T = L_{\pi_D}(\lambda)$ where the associated signed permutation matrix is exactly $A = D\cdot I$. By Eq.~\eqref{eq:D_determinant}, the stabilizer condition holds true. Hence, $\pi_D \in\mathrm{Stab}_{S_{mn}}(h_{\mathrm{id}})  $.   Since each diagonal entry $\epsilon_i$ of $D$ can either be $+1$ or $-1$, and there are exactly $m$ of them, there are $2^m$ possible diagonal matrices of the form $D$. However, notice that trivially $D L_{\mathrm{id}} D = (-D) L_{\mathrm{id}} (-D)$. Therefore, it follows that the total number of distinct, sign-swapping induced permutations, including the identity is $2^{m-1}.$ Again, these permutations form a subgroup $H_{J_{D}}$ which isomorphic to $\mathbb{Z}_{2}^{m-1}$ so that $|H_{J_{D}}|= 2^{m-1}.$

Now since the subgroups $H_{J_{\tau}}$ and $H_{J_{D}}$ are uniquely defined via their respective signed permutation matrices, $H_{J_{\tau}}\cap H_{J_{D}} = \{\mathrm{id}\}$. Indeed, every element of $H_{J_D}$ fixes the diagonal active indices $p(i,i)$, whereas a nontrivial element of $H_{J_\tau}$ permutes at least two of these indices. Additionally, for arbitrary coordinates $k = p(i,j) \in J,$ we have that for $i<j$,
\begin{align*}
   (\pi_\tau  \pi_D  \pi_\tau^{-1})(p(i,j)) &= \pi_\tau \left( \pi_D \left( p(\tau^{-1}(i), \tau^{-1}(j)) \right) \right) \\
    &= \pi_\tau \left( \begin{cases} 
        q(\tau^{-1}(i), \tau^{-1}(j)) & \text{if } \epsilon_{\tau^{-1}(i)} \epsilon_{\tau^{-1}(j)} = -1 \\ 
        p(\tau^{-1}(i), \tau^{-1}(j)) & \text{otherwise} 
    \end{cases} \right) \\
    &= \begin{cases} 
        q(\tau(\tau^{-1}(i)), \tau(\tau^{-1}(j))) & \text{if } \epsilon_{\tau^{-1}(i)} \epsilon_{\tau^{-1}(j)} = -1 \\ 
        p(\tau(\tau^{-1}(i)), \tau(\tau^{-1}(j))) & \text{otherwise} 
    \end{cases} \\
    &= \begin{cases} 
        q(i,j) & \text{if } \epsilon_{\tau^{-1}(i)} \epsilon_{\tau^{-1}(j)} = -1 \\ 
        p(i,j) & \text{otherwise.} 
    \end{cases}
\end{align*}
Similarly, it can be seen that for $k = q(i,j)\in J$, \begin{equation*}
     (\pi_\tau  \pi_D  \pi_\tau^{-1})(q(i,j))  = \begin{cases} 
        p(i,j) & \text{if } \epsilon_{\tau^{-1}(i)} \epsilon_{\tau^{-1}(j)} = -1 \\ 
        q(i,j) & \text{otherwise} 
    \end{cases}     \end{equation*}
and $(\pi_\tau \circ \pi_D \circ \pi_\tau^{-1})(k) = k$ for $k \notin J.$ Altogether,  the composition defines a new permutation such that $\pi_\tau \pi_D \pi_\tau^{-1} = \pi_{\tilde{D}} \in H_{J_D}$, parametrized by $\tilde{\epsilon}_i = \epsilon_{\tau^{-1}(i)}$. Therefore, $H_{J_D}$ is a normal subgroup of $\mathrm{Stab}_{S_{mn}}(h_{\mathrm{id}})$. Thus, it follows that the total subgroup $H_J$ of the stabilizer group defined under the index set $J$ is isomorphic to the semidirect product of the two generating groups,
$$H_J \cong H_{J_{D}}\rtimes H_{J_\tau}\cong (\mathbb{Z}_2)^{m-1} \rtimes S_m.$$ 
Thus, $H_J$ has order $ |H_J| = 2^{m-1}\cdot m!$.

Furthermore, since the stabilizer subgroups $H_{I\setminus J}$ and $H_J$ act on strictly disjoint index sets,  $H_{I\setminus J}\cap H_J = \{\mathrm{id}\}$ and their elements commute, it implies that both sets are normal subgroups of the stabilizer group. By Lemma~\ref{lem:signed-permutation}, every element of the stabilizer $\mathrm{Stab}_{S_{mn}}(h_{\mathrm{id}})$ preserves the active index set $J$ and its action on $J$ is induced by a signed simultaneous row-and-column permutation. Since every signed permutation matrix has the form $A =D\Pi_\tau$, the subgroup $H_J$ constructed above exhausts the full stabilizer set. That is, 
$$\mathrm{Stab}_{S_{mn}}(h_{\mathrm{id}}) = \{\pi\in S_{mn}\mid L_\pi = AL_{\mathrm{id}}A^T \textrm{ for a signed permutation matrix } A = D\Pi\}. $$ Thus, every stabilizing permutation preserves $J$ and subsequently preserves $I\setminus J$. This means that each element of the stabilizer decomposes uniquely as a product of a permutation supported on $I\setminus J$ and a permutation supported on $J$. Thus, 
$$\mathrm{Stab}_{S_{mn}}(h_{\mathrm{id}})\cong H_{I\setminus J}\times H_J.$$ 
Therefore, the order of the stabilizer group is given by 
\begin{equation*}
        |\mathrm{Stab}_{S_{mn}}(h_{\mathrm{id}})| = |H_{I\setminus J}| \cdot |H_{J_{D}}|\cdot | H_{J_\tau}| =(mn-m^2)! \cdot 2^{m-1} \cdot m!.
\end{equation*}
Substituting this into Eq.~\eqref{eq:stabilizer_orbit}, it follows that
\begin{equation*}
        \kappa_{m,n} =  \frac{(mn)!}{(mn-m^2)!\cdot 2^{m-1} \cdot m!}.
\end{equation*}

Since  $\partial\mathrm{APPT}_{m,n} =\left\{\lambda\in\mathrm{APPT}_{m,n}:\det(\mathcal{L}(\lambda))=0\right\}$ from Theorem~\ref{thm:boundary_characterization} and $\det(\mathcal{L}(\lambda)) = \prod_{\pi\in \widetilde{S}} \det(L_\pi)=0 $, it follows that 
$ \partial \mathrm{APPT}_{m,n} = \mathrm{APPT}_{m,n} \cap \left( \bigcup_{\pi \in S_{mn}} \mathcal{Z}_{V_{mn}}(\det(L_{\pi}(\lambda)) )\right)$ where $V_{mn}=\{\lambda\in \RR^{mn}: \sum_{i=1}^{mn}\lambda_i=1 \}$.
Choose permutations $\pi_1, \ldots, \pi_{\kappa_{m,n}}\in S_{mn}$ corresponding to the distinct polynomials in the orbit. We may equivalently choose these polynomials from $\widetilde{S}$ as permutations of the $(mn-m^2)$ inactive indices of the spectrum leave the matrix constraint unchanged. Then we have  $$\bigcup_{\pi \in S_{mn}} \mathcal{Z}_{V_{mn}}(\det(L_{\pi}(\lambda))) =\bigcup_{i=1}^{\kappa_{m,n}} \mathcal{Z}_{V_{mn}}(\det(L_{\pi_i}(\lambda))) $$
where each $\det(L_{\pi_i}(\lambda))$ are distinct irreducible polynomials of degree $m$ in the polynomial ring $\mathbb{R}[\lambda]$, by construction. Suppose two algebraic sets generated by any two distinct polynomials $\det(L_{\pi_i}(\lambda)),\det(L_{\pi_j}(\lambda))$ in the affine space $V_{mn}$  are equal. Since the polynomials are irreducible,  $\det(L_{\pi_i}(\lambda))= c\det(L_{\pi_j}(\lambda))$. Evaluating at the maximally mixed spectrum $\mathbf{u}_{mn}= \frac{1}{mn}(1,\ldots,1)$, it follows that
\begin{align*}
    \det(L_{\pi_i}(\mathbf{u}_{mn}))&= c\det(L_{\pi_j}(\mathbf{u}_{mn}))\\
    \left(\frac{2}{mn}\right)^m &= c\left(\frac{2}{mn}\right)^m \neq 0.
\end{align*}
Hence $c=1$ and contradicts the assumption of equal algebraic sets. Thus, the hypersurfaces $\mathcal{Z}_{V_{mn}}(\det(L_{\pi_i}(\lambda)))$ are pairwise distinct for $i =1,\ldots, \kappa_{m,n}$. The affine space $V_{mn}$ has dimension $mn-1$ and for each $i$, the restriction of the polynomial $\det(L_{\pi_i}(\lambda)) $ to $V_{mn}$ is a non-constant irreducible polynomial. Therefore, by \cite[Proposition~1.13]{hartshorne1977}, each algebraic set  $\mathcal{Z}_{V_{mn}}(\det(L_{\pi_i}(\lambda)))$ in the $(mn-1)$-dimensional affine space $V_{mn}$ is an irreducible hypersurface of dimension
$$\dim(\mathcal{Z}_{V_{mn}}(\det(L_{\pi_i}(\lambda))) ) = (mn-1)-1= mn-2,$$
as desired.
\end{proof}

\end{document}